\documentclass[12pt,letterpaper]{article}

\usepackage{natbib}

\usepackage[colorlinks,
            linkcolor=red,
            linktoc=all,
            anchorcolor=red,
            citecolor=blue
            ]{hyperref}

\usepackage[ left=1in, top=1in, right=1in, bottom=1in]{geometry}
\usepackage{graphicx,bm,colonequals,amsmath,amssymb,url,xcolor,bbm}
\usepackage{array,tabularx,multirow}
\usepackage[font={footnotesize}]{caption,subcaption}
\usepackage[utf8]{inputenc}
\usepackage{enumitem}
\usepackage{soul} %
\usepackage{placeins} %
\usepackage[normalem]{ulem}

\graphicspath{{plots/}}

\usepackage{xr}
\makeatletter
\newcommand*{\addFileDependency}[1]{%
  \typeout{(#1)}
  \@addtofilelist{#1}
  \IfFileExists{#1}{}{\typeout{No file #1.}}
}
\makeatother
\newcommand*{\myexternaldocument}[1]{%
    \externaldocument{#1}%
    \addFileDependency{#1.tex}%
    \addFileDependency{#1.aux}%
}
\myexternaldocument{supplement}

\usepackage{mathtools}
\mathtoolsset{showonlyrefs}

\bibpunct[, ]{(}{)}{;}{a}{,}{,}

\usepackage{amsthm}
\newtheoremstyle{propstyle} %
    {2mm}                    %
    {1mm}                    %
    {\itshape}                   %
    {}                           %
    {\scshape}                   %
    {.}                          %
    {.5em}                       %
    {}  %
\theoremstyle{propstyle}
\newtheorem{proposition}{Proposition}
\theoremstyle{propstyle}

\theoremstyle{propstyle}

\theoremstyle{propstyle}
\newtheorem{lemma}{Lemma}
\theoremstyle{propstyle}

\theoremstyle{propstyle}
\newtheorem{assumption}{Assumption}
\theoremstyle{propstyle}
\newtheorem{theorem}{Theorem}

\usepackage[ruled,vlined]{algorithm2e}
\SetKwInput{KwInput}{Input}
\SetKwInOut{KwOutput}{Output}
\usepackage{algorithmic}
\usepackage{array,framed}
\usepackage{float}

\makeatletter
\renewcommand{\paragraph}{%
  \@startsection{paragraph}{4}%
  {\z@}{2ex \@plus 1ex \@minus .2ex}{-1em}%
  {\normalfont\normalsize\bfseries}%
}
\makeatother

\DeclareMathAlphabet\mathbfcal{OMS}{cmsy}{b}{n}

\newcommand{\ba}{\mathbf{a}}

\newcommand{\bb}{\mathbf{b}}
\newcommand{\bt}{\mathbf{t}}
\newcommand{\bs}{\mathbf{s}}

\newcommand{\bx}{\mathbf{x}}
\newcommand{\by}{\mathbf{y}}

\newcommand{\bg}{\mathbf{g}}

\newcommand{\bZ}{\mathbf{Z}}

\newcommand{\bA}{\mathbf{A}}
\newcommand{\bY}{\mathbf{Y}}

\newcommand{\bW}{\mathbf{W}}

\newcommand{\bE}{\mathbf{E}}
\newcommand{\bL}{\mathbf{L}}
\newcommand{\bN}{\mathbf{N}}
\newcommand{\bI}{\mathbf{I}}
\newcommand{\bD}{\mathbf{D}}

\newcommand{\bU}{\mathbf{U}}

\newcommand{\bX}{\mathbf{X}}
\newcommand{\bQ}{\mathbf{Q}}
\newcommand{\bB}{\mathbf{B}}

\newcommand{\bM}{\mathbf{M}}

\newcommand{\cX}{\mathcal{X}}
\newcommand{\cS}{\mathcal{S}}

\newcommand{\bfzero}{\mathbf{0}}

\newcommand{\bftheta}{\bm{\theta}}

\newcommand{\bfbeta}{\bm{\beta}}
\newcommand{\bfepsilon}{\bm{\epsilon}}

\newcommand{\bfSigma}{\bm{\Sigma}}

\newcommand{\bfPsi}{\bm{\Psi}}
\newcommand{\bfOmega}{\bm{\Omega}}
\newcommand{\bfDelta}{\bm{\Delta}}

\newcommand{\cov}{\textbf{Cov}}

\newcommand{\tr}{tr}

\DeclareMathOperator*{\argmin}{arg\,min}

\title{Tensor Covariance Estimation via Kronecker-Structured Sparse Inverse Cholesky}

\author{Wentao Zhan\thanks{Department of Statistics, University of Wisconsin--Madison} \and Matthias Katzfuss\footnotemark[1] \thanks{Corresponding author: \texttt{katzfuss@gmail.com}}}

\date{}

\begin{document}

\maketitle

\begin{abstract}
High-dimensional multi-way (tensor) data pose significant challenges for covariance estimation due to the curse of dimensionality. We introduce a unified framework for scalable estimation of tensor covariances based on a Kronecker-structured sparse inverse Cholesky (KSIC) projection. Our approach is grounded in the geometry of information projection, defining the estimator as the moment-matching projection of a target distribution onto a manifold characterized by sparse, Kronecker-factored inverse Cholesky factors. By leveraging physical or data-driven nearest-neighbor sparsity, KSIC provides a geometry-aware representation that is both statistically interpretable and computationally efficient. Our framework integrates two estimation regimes: a nonparametric estimator that projects the empirical covariance directly onto the manifold, utilizing the KSIC structure to implicitly regularize rank-deficient data; and a parametric estimator that fits generative covariance models (e.g., Mat\'ern) by maximizing the likelihood of their KSIC projections, formulated as a nested double forward Kullback-Leibler minimization.
Theoretically, we establish the conditions for the existence of the KSIC projection and finite-sample concentration rates for the nonparametric regime, proving that the KSIC estimator gainfully exploits cross-mode information and is robust to data scarcity. Numerical experiments demonstrate that the proposed KSIC estimators achieve state-of-the-art accuracy and scalability, particularly in settings with high dimensionality and limited sample sizes. We apply KSIC to spatiotemporal temperature anomalies and functional MRI data, demonstrating its broad applicability across diverse multi-way data domains.
\end{abstract}

\section{Introduction \label{sec:intro}}

\subsection{High-Dimensional Multi-Way Data and Applications}

With advancing observational technologies, high-dimensional multi-way (tensor) data have become ubiquitous across scientific disciplines, including climatology \citep{li2008three}, neuroimaging \citep{bijma2005spatiotemporal}, biology \citep{teng2009statistical}, and economics \citep{hao2021sparse}. Investigating the intricate dependency structures within such datasets requires modeling their multi-way covariances. Despite decades of research in structured covariance estimation \citep[e.g.,][]{rodriguez1974design, Stein2005, genton2007separable}, the extreme scale of modern multi-way observations presents a formidable challenge that continues to drive statistical innovation \citep[e.g.,][]{white2019nonseparable, chen2021space}.

Formally, consider a collection of $K$-way tensor observations $\{\cX_j\in \mathbb{R}^{p_1\times \cdots \times p_K}\}_{j=1}^n$, where the total dimension $p = \prod_{k=1}^K p_k$ can scale into the millions. Methodological advancements in this domain are routinely bottlenecked by two statistical regimes: massive grids where the scale of $p$ renders standard covariance estimators computationally intractable, and small-sample settings where the number of independent replicates is vastly exceeded by the dimension ($n \ll p$). We highlight four applications that typify these bottlenecks.

\textit{First, in satellite remote sensing}, observations are naturally indexed by space, time, and spectral band. The underlying multi-way dependencies are physically heterogeneous: spatial interactions are local, temporal paths exhibit periodicities, and spectral bands display sharp block correlations. Capturing these heterogeneous structures requires a scalable estimator whose mode-specific components adapt to distinct geometries and sparsity patterns without explicitly forming the dense, massive joint covariance matrix.

\textit{Second, in neuroimaging connectivity analysis}, functional MRI (fMRI) data are structured as tensors of spatial locations and time points observed across multiple subjects \citep{noroozi2020tensor}. Mapping functional brain networks requires estimating high-dimensional precision matrices; however, because the sample size $n$ (subjects) is typically small relative to the space-time dimension $p_1 p_2$, the empirical covariance is severely rank-deficient. Overcoming this data scarcity requires a regularized framework that exploits mode-specific geometric structures to borrow strength across modes and yield stable, sparse connectivity maps.

\textit{Third, in computer experiments and digital twins}, surrogates for physical states must predict variables across complex 3D meshes in real-time. The resulting tensors exhibit a hybrid topology: the irregular mesh requires a highly sparse spatial factor, whereas the multiple physical state variables are densely inter-correlated. An ideal covariance estimator must flexibly accommodate this hybrid structure, seamlessly blending sparse structural factors with dense, unconstrained components.

\textit{Finally, in spatial and environmental statistics}, the objective is often spatiotemporal interpolation (kriging) over vast grids. This task requires fitting parametric generative covariance models, such as non-separable Mat\'{e}rn fields, via maximum likelihood. However, evaluating the exact likelihood scales cubically with the total dimension, $\mathcal{O}(p^3)$, presenting an insurmountable computational barrier. Overcoming this requires a fundamentally new formulation that evaluates multi-way likelihoods efficiently without sacrificing parametric flexibility.

Collectively, these diverse scientific challenges underscore the necessity of a flexible, unified covariance framework that can seamlessly navigate both nonparametric regularization under data scarcity and scalable parametric estimation on massive grids. Crucially, while many of these applications feature natural physical coordinates (e.g., spatial distances or time), such a framework must also accommodate data-driven geometries (e.g., correlation-based metrics) to capture complex dependencies when explicit physical layouts are unavailable or uninformative.

\subsection{Limitations of Existing Tensor Covariance Models}

Unrestricted covariance estimation quickly becomes intractable, both computationally and statistically, as the dimension $p$ grows. To mitigate this curse of dimensionality, the \textit{Kronecker product (separable) model} has been widely adopted as a foundational dimension-reduction tool \citep{Dawid1981, hoff2011separable, du2017cbinderdb, drton2021existence, zhang2023covariance, mayrhofer2025robust}. By assuming the joint covariance factors into a Kronecker product of mode-specific marginal covariances, this approach drastically reduces the parameter space \citep{dutilleul1999mle, werner2008estimation, Manceur2013TensorNormal}. However, strict separability is often an overly rigid physical assumption. Furthermore, estimating these dense marginal factors still scales poorly when individual mode dimensions are large \citep{dutilleul1999mle, lu2005likelihood, srivastava2008models, soloveychik2016gaussian, drton2021existence}. While extensions such as sums of Kronecker products \citep{greenewald2013kronecker, tsiligkaridis2013covariance, Cao2022, zhou2025kronecker} and core shrinkage \citep{hoff2023core} increase flexibility, they obscure the interpretability of marginal structures and are generally ill-suited for precision matrix estimation.

When the inferential focus shifts from marginal covariance to conditional dependence, \textit{Gaussian graphical models (GGMs)} offer an alternative by inducing sparsity directly on the precision matrix. The graphical lasso \citep{friedman2008sparse} and its multi-way extensions, including the Kronecker graphical lasso \citep{tsiligkaridis2013convergence}, tensor graphical lasso \citep{greenewald2019tensor}, and Sylvester graphical lasso \citep{wang2020sylvester}, have proven highly effective for structural learning and network recovery. Nevertheless, these regularization-based methods face two critical limitations. First, they are predominantly descriptive; they excel at identifying structural zeros but are not designed to estimate the generative parametric kernels (e.g., spatial range or smoothness) required for physical modeling. Second, despite significant algorithmic advancements \citep{lyu2019tensor, min2022fast}, optimizing the $\ell_1$-penalized likelihood remains computationally prohibitive for massive tensors, as the underlying subproblems frequently retain cubic complexity with respect to the mode dimensions.

A third, highly scalable paradigm relies on the \textit{sparse inverse Cholesky (SIC)} factorization. Rooted in Vecchia's approximation \citep{Vecchia1988} and mathematically justified by the screening effect in geostatistics \citep{Stein2002}, SIC induces sparsity directly on the Cholesky factor of the precision matrix based on a prescribed variable ordering. This framework dramatically reduces the cost of exact likelihood evaluation and sampling from $\mathcal{O}(p^3)$ to $\mathcal{O}(p)$ \citep{Schafer2020, katzfuss2021general, datta2022nearest}. While SIC formulations have revolutionized spatial and spatiotemporal modeling, standard Vecchia approximations do not natively accommodate the tensor-product geometry of multi-way data. Applying them to tensor arrays typically requires an unnatural flattening of the multidimensional lattice. This vectorization destroys the structural Kronecker geometry, obscures mode-specific conditional dependencies, and compromises both the interpretability and the computational efficiency of the estimator.

\subsection{Kronecker-Structured Sparse Inverse Cholesky Framework}

To bridge these methodological gaps, we introduce the Kronecker-Structured Sparse Inverse Cholesky (KSIC) framework. KSIC provides a unified approach to high-dimensional tensor covariance estimation, grounded in the geometry of information projection. Our primary contributions are four-fold:

\begin{enumerate}
    \item A Geometric Framework via Information Projection: We define the KSIC estimator as the information projection of a target distribution onto a manifold of sparse, Kronecker-factored precision matrices. By minimizing the forward Kullback-Leibler (KL) divergence, KSIC elegantly extends the computational benefits of the SIC to tensor product spaces, yielding a structured approximation that optimally captures the specified multi-way dependencies using nearest-neighbor sparsity based on either physical or data-driven geometries.
    
    \item Unification of Nonparametric and Parametric Regimes: The framework integrates two traditionally distinct estimation tasks under one umbrella. For nonparametric estimation, we project the empirical covariance directly onto the KSIC manifold, utilizing the Kronecker structure to inherently regularize rank-deficient data (acting as structural data augmentation). For parametric estimation, we fit generative covariance models (e.g., Mat\'{e}rn) by maximizing the likelihood of their structured KSIC projections, formulated as a novel, nested double forward KL optimization problem.
    
    \item Scalable Optimization and Robustness to Data Scarcity: We develop an efficient block coordinate descent (BCD) algorithm to compute the KSIC projection that scales gracefully to massive dimensions. Furthermore, by explicitly exploiting cross-dimensional geometries, KSIC natively pools information across modes. This structural regularization makes the estimator highly robust to data scarcity, enabling stable estimation even with a minimal number of replicates.
    
    \item We provide a comprehensive theoretical analysis of KSIC, establishing sufficient conditions for the existence of the projection. Under the nonparametric regime, this translates to a novel finite-sample threshold for the minimum number of samples required for stable estimation, proving that KSIC is highly robust to data scarcity. Furthermore, we establish that KSIC achieves an optimal convergence rate that asymptotically matches or improves upon existing estimators, providing rigorous theoretical justification for the necessity of modeling cross-dimensional dependencies.
\end{enumerate}

Our numerical experiments corroborate these theoretical advances, demonstrating that KSIC is substantially more versatile and accurate than existing separable models. In summary, by unifying the physical interpretability of generative models with the extreme scalability of sparse matrix algebra, KSIC provides an ideal tool for modern multi-way data analysis.

The remainder of the article is organized as follows. In Section \ref{sec:preliminaries}, we review separable covariance models for multi-way data and the SIC factorization. Section \ref{sec:KSIC} formally proposes the KSIC projection, detailing the approximation algorithms and their applications in different estimation regimes. Section \ref{sec:sim} demonstrates the finite-sample performance of KSIC via simulation studies. Section \ref{sec:real} provides real-data examples where KSIC is applied to spatiotemporal climate arrays and fMRI data.

\section{Preliminaries and Background}
\label{sec:preliminaries}

\subsection{Multi-Way Data and Separable Covariance Models}
\label{subsec:multi-way-review}

We adopt standard tensor notation throughout the manuscript. Scalars are denoted by lowercase letters (e.g., $x$), vectors by bold lowercase ($\mathbf{x}$), matrices by bold uppercase ($\mathbf{X}$), and tensors by calligraphic letters ($\mathcal{X}$). The Kronecker product is denoted by $\otimes$, and the operator $\text{vec}(\cdot)$ stacks the elements of a tensor into a column vector.

Consider an order-$K$ tensor observation $\cX \in \mathbb{R}^{p_1\times \cdots \times p_K}$, with total dimension $p = \prod_{k=1}^K p_k$. In this manuscript, we assume $\cX$ follows a mean-zero tensor-variate normal distribution, such that $\text{vec}(\cX) \sim \mathcal{N}_p(\bfzero, \bfSigma)$, where $\bfSigma$ is the $p\times p$ joint covariance matrix. 

Unrestricted estimation of $\bfSigma$ is typically prohibitive. A prevalent and parsimonious structural assumption is the \textit{separable} (or Kronecker product) covariance model \citep{hoff2011separable}, which assumes $\bfSigma$ decomposes as $\bfSigma = \bfSigma_K \otimes \cdots \otimes \bfSigma_1$, where $\bfSigma_k \in \mathbb{R}^{p_k \times p_k}$ is the marginal covariance of the $k$-th mode. A major advantage of separable covariance models is their clear interpretation of mode-specific marginal covariance structures. This feature is particularly useful when the scientific goal is to identify similarities or correlations among observations indexed by a given mode \citep{hoff2011separable}. 

However, strict separability is often too restrictive for general multi-way covariance modeling. We emphasize that the KSIC projection and most of our results apply to general covariance models. The separable structure is treated as an important special case that simplifies some theoretical results and provides additional intuition for our geometric framework.

\subsection{Sparse Inverse Cholesky (SIC) Approximation}
\label{subsec:sic_review}

The standard SIC approximation serves as the atomic unit for our Kronecker-structured framework. SIC is rooted in the Vecchia approximation, which was originally proposed to approximate a spatial joint density via ordered conditional distributions with reduced conditioning sets \citep{Vecchia1988}. This implicitly defines a SIC factorization of the precision matrix, $\bfSigma^{-1} \approx \hat{\bL}\hat{\bL}^\top$. Here, $\hat{\bL}$ is lower triangular, and its $i$th column is non-zero only at the diagonal index $i$ and a subset of sub-diagonal indices $\bs_i \subset \{i+1,\ldots,p\}$ corresponding to the chosen conditioning sets \citep[e.g.,][]{Datta2016a, katzfuss2021general}. We use $\bs^*_i := \{i\}\cup\bs_i$ to denote the full set of non-zero elements in column $i$, and $m := \max_i|\bs_i|$ for the maximum size of the conditioning set.
Let 
\begin{equation}
\mathcal{S} = \{\bL \in \mathbb{R}^{p \times p}: \bL_{j,i} \neq 0 \Rightarrow j \in \bs^*_i \}    
\end{equation} 
be the manifold of lower triangular matrices with this fixed sparsity pattern. A foundational result establishes that the SIC approximation $\hat{\bL}$ is the unique minimizer of the forward Kullback-Leibler (KL) divergence from the true model to the approximation \citep{Schafer2020}:
\begin{equation}
    \hat{\bL} = \mathop{\mathrm{arg\,min}}_{\bL \in \mathcal{S}} \mathrm{KL}\left( \mathcal{N}(\bfzero, \bfSigma) \;\|\; \mathcal{N}(\bfzero, (\bL\bL^\top)^{-1}) \right) := \Pi(\bfSigma,\mathcal{S}).
    \label{eq:forward_kl_obj}
\end{equation}

We emphasize that $\Pi$ is an information-geometric moment projection (M-projection). Unlike the reverse KL divergence (often used in variational inference, which is mode-seeking and tends to underestimate variance), the M-projection perfectly preserves the local conditional distributions of the target—meaning the implied conditional variances and regression coefficients of each variable given its conditioning set match the truth exactly.

Remarkably, the optimization problem \eqref{eq:forward_kl_obj} decouples column-wise and admits a closed-form solution that can be computed efficiently. The non-zero entries of the $i$th column of $\hat{\bL}$ are given by:
\begin{equation}
\hat{\bL}_{\bs^*_i,i} = \frac{\bfbeta_i}{(\mathbf{e}_1^{\top}\bfbeta_i)^{1/2}},
    \label{eq:sic_closed_form}
\end{equation}
where $\bfbeta_i = (\bfSigma_{\bs^*_i, \bs^*_i})^{-1} \mathbf{e}_1$, and $\mathbf{e}_1$ is the indicator vector with the first entry equal to one and all other entries equal to zero (assuming the first entry of $\bs^*_i$ is $i$). This allows the entire approximation to be computed in $\mathcal{O}(p m^3)$ time, scaling linearly with the total dimension $p$.

\section{Kronecker-Structured Sparse Inverse Cholesky (KSIC)}
\label{sec:KSIC}

In this section, we synthesize the concepts developed above to introduce the KSIC projection. In subsection~\ref{subsec:KSIC}, we present the main optimization algorithm, demonstrating how the Kronecker product model and the SIC approximation naturally combine to yield a computationally efficient procedure. In subsection~\ref{subsec:property}, we describe theoretical properties of KSIC projection and the BCD algorithm. In subsection~\ref{subsec:non-parametric}, we study the KSIC projection under a nonparametric regime, showing that in data-scarce settings, the structural constraints imposed by KSIC effectively regularize the estimator and improve upon the sample-size thresholds required by the classical MLE \citep{soloveychik2016gaussian, drton2021existence}. Finally, in subsection~\ref{subsec:parametric_estimation}, we describe how the KSIC projection can be incorporated into parametric covariance estimation through a highly scalable projected-likelihood formulation.

\subsection{KSIC Projection via Block Coordinate Descent}
\label{subsec:KSIC}

For a given tensor $\mathcal{X} \in \mathbb{R}^{p_1 \times \cdots \times p_K}$, we define the KSIC manifold as the set of Cholesky factors given by the Kronecker product of sparse lower-triangular matrices:
\begin{equation}
\mathcal{S}_{\text{KS}} =
\{\bL \in \mathbb{R}^{p \times p}: \bL = \bigotimes_{k = 1}^K\bL_k, \; \bL_k \in \mathcal{S}_k \subset \mathbb{R}^{p_k \times p_k}\}, 
\label{eq:ksicsparsity}    
\end{equation}
with fixed sparse matrix subspaces $\mathcal{S}_k$'s. For modes possessing geometric or physical information (e.g., space or time), these sparsity patterns typically correspond to nearest neighbors; for modes without such geometry, they may be diagonal, dense, or rely on neighbors derived from correlation distances \citep[e.g.,][]{Kang2021}.

For a given KSIC manifold \eqref{eq:ksicsparsity} and a $p\times p$ target covariance $\bfSigma$, we define the KSIC projection operator via forward KL minimization:
\begin{equation}
    \Pi_{\text{KS}}(\bfSigma,\mathcal{S}_{\text{KS}}) = \mathop{\mathrm{arg\,min}}_{\bL \in \mathcal{S}_{\text{KS}}} \mathrm{KL}\left( \mathcal{N}(\bfzero, \bfSigma) \;\|\; \mathcal{N}(\bfzero, (\bL\bL^\top)^{-1}) \right). 
    \label{eq:ksicproj}
\end{equation}
While this projection cannot be computed in a single closed-form step, the optimization can be efficiently solved using block coordinate descent (BCD). This consists of iteratively optimizing the KL divergence with respect to one factor $\bL_k$ while holding the remaining factors $\bL_{-k} = \bigotimes_{l \neq k}\bL_l$ fixed. The key to the algorithm's efficiency is formally reducing the global optimization problem to an individual, lower-dimensional SIC projection during each iteration.

\begin{proposition}\label{prop:pseudocov}
Fix all factors except $\bL_k$. Then the conditional minimization problem in \eqref{eq:ksicproj} reduces exactly to a standard SIC projection onto $\cS_k$:
\begin{equation}
\operatorname*{arg\,min}_{\bL_k:\, \bL \in \mathcal{S}_{\text{KS}}}
\mathrm{KL}\left(
\mathcal{N}(\bfzero, \bfSigma)
\;\|\;
\mathcal{N}(\bfzero, (\bL\bL^\top)^{-1})
\right)
=
\Pi(\tilde\bfSigma_k,\mathcal{S}_k),
\label{eq:pseudocov}    
\end{equation}
where $\Pi$ denotes the closed-form SIC projection defined in
\eqref{eq:forward_kl_obj}--\eqref{eq:sic_closed_form}. The pseudo-covariance matrix $\tilde\bfSigma_k \in \mathbb{R}^{p_{k} \times p_{k}}$ is given by the weighted sum of mode-$k$ covariance submatrices:
\begin{equation}\label{eq:pseudocov-comp}  
\tilde{\bfSigma}_{k}
=
\frac{1}{p_{-k}}
\sum_{i,j=1}^{p_{-k}}
\left(
\bL_{-k}\bL^\top_{-k}
\right)_{i,j}
\bfSigma^{(k)}_{(i,j)},
\end{equation}
where $\left(\bL_{-k}\bL^\top_{-k}\right)_{i,j}$ is the $(i,j)$-th entry of $\bL_{-k}\bL^\top_{-k}$ with $\bL_{-k} = \bigotimes_{l \neq k}\bL_l$, and $\bfSigma^{(k)}_{(i,j)}$ is the $p_k \times p_k$ submatrix of $\bfSigma$ obtained by fixing the indices of all modes other than $k$ to $i$ and $j$. Here, both $i$ and $j$ index elements of the product set $\bigotimes_{l \neq k}\{1,\ldots,p_l\}$.
\end{proposition}

Proposition~\ref{prop:pseudocov} yields an explicit construction of a pseudo-covariance matrix that perfectly absorbs the structural contributions of all fixed factors, thereby reducing the conditional update of $\bL_k$ to a standard SIC projection. 

When the target covariance $\bfSigma$ is defined empirically from a limited number of samples, we can bypass \eqref{eq:pseudocov-comp} and compose the pseudo-covariance even more efficiently using its low-rank structure:

\begin{proposition}\label{prop-low_rank}
When $\bfSigma = \sum_{r = 1}^n\bx_r\bx_r'$, where $\bx_r \in \mathbb{R}^{p}$, the pseudo-covariance matrix in \eqref{eq:pseudocov} can be expressed as:
\[
\tilde{\bfSigma}_{k} = \frac{1}{p_{-k}}\sum_{r=1}^n\bX^{(k)}_{r}\bL_{-k}\bL^\top_{-k}\bX^{(k)\top}_{r},
\]
where $\bX_r^{(k)}$ denotes a $p_k \times p_{-k}$ matrix obtained by unfolding $\bx_r$ along mode $k$.
\end{proposition}

\begin{algorithm}[htbp]
\caption{KSIC-BCD}
\KwInput{Multi-way covariance matrix $\bfSigma$ (or an oracle for its entries), sparsity sets $\{\cS_1, \ldots, \cS_K\}$.}
\KwOutput{KSIC projection $\bL = \Pi_{\text{KS}}(\bfSigma,\mathcal{S}_{\text{KS}})$.}
\begin{algorithmic}[1]
\STATE Initialize $\bL_k$ for $k = 1, \ldots, K$. (Typically, $\bL_k = \Pi(\bfSigma_k,\mathcal{S}_k)$ for some $\bfSigma_k$.)
\WHILE{not converged}
\FOR{$k=1,2,\ldots, K$}
     \STATE Compute the required entries of $\tilde{\bfSigma}_k$ using $\bL_{-k}$ via Proposition \ref{prop:pseudocov} or \ref{prop-low_rank}. 
     \STATE Compute the SIC projection $\bL_k \gets  \Pi(\tilde\bfSigma_k,\mathcal{S}_k)$ using \eqref{eq:forward_kl_obj}--\eqref{eq:sic_closed_form}.
\ENDFOR
\ENDWHILE
\RETURN $\bL = \bigotimes_{k = 1}^K\bL_k$.
\end{algorithmic}
\label{alg-main}
\end{algorithm}

Algorithm~\ref{alg-main} summarizes the KSIC-BCD procedure. The algorithm accommodates any suitable initialization strategy. A natural default choice is the identity initialization $\bfSigma_k = \bI_{p_k}$, which requires no additional computation. Under this initialization, the first-round mode-specific covariance estimates correspond to averaged covariances across fibers aligned along all other modes. When prior covariance information is available, initialization can also be performed in a geometry-aware manner through an initial SIC projection.

Compared with other BCD-type algorithms for tensor covariance analysis \citep[e.g.,][]{dutilleul1999mle, werner2008estimation, lyu2019tensor}, whose computational costs typically grow super-linearly with the total dimension $p$, the primary advantage of the KSIC-BCD algorithm is its extreme scalability. The sparsity of the SIC factors improves computational efficiency in several profound ways. 
First, the BCD completely avoids dense matrix inversion. As reviewed in Section~\ref{subsec:sic_review}, the SIC representation reduces the mode-wise inversion cost from $\mathcal{O}(p_k^3)$ to $\mathcal{O}(p_k)$. Second, computing the SIC projection $\Pi(\tilde\bfSigma_k,\mathcal{S}_k)$ only requires the local submatrices of $\tilde\bfSigma_k$ determined by the sparse index sets. Therefore, the full pseudo-covariance matrix never needs to be explicitly instantiated. 
Most importantly, computing the necessary entries of the pseudo-covariance is remarkably fast even when the generative covariance $\bfSigma$ is fully dense and full-rank (Proposition \ref{prop:pseudocov}). Because each mode-specific factor $\bL_l$ is sparse, the intermediate matrix $\bL_l\bL_l^\top$ contains at most $\mathcal{O}(m_l^2 p_l)$ non-zero entries. Consequently, their Kronecker product $\bL_{-k}\bL_{-k}^\top$ is highly sparse, possessing at most $\mathcal{O}(p_{-k} \prod_{l \neq k} m_l^2)$ non-zeros. By only evaluating the weighted sum in \eqref{eq:pseudocov-comp} for the specific indices required by the sparse matrix  subspace $\cS_k$, the computation perfectly circumvents the dense matrix $\bfSigma$. Assuming the entries of $\bfSigma$ can be queried in $\mathcal{O}(1)$ time (as is standard for parametric covariance kernels), these properties yield an overall computational cost that scales strictly linearly with the total dimension $p$, up to factors depending on the tensor order and the bounded conditioning-set sizes.

\begin{proposition}\label{prop:cost}
Consider a $p\times p$ multi-way covariance matrix $\bfSigma$ of order $K$ and sparse matrix  subspaces $\{\cS_1,\ldots,\cS_K\}$. For each mode $k$, define the maximum conditioning-set size as
\[
m_k := \max_{i=1,\ldots,p_k} |\bs_i^k|,
\]
where $\bs_i^k$ denotes the nonzero index set for the $i$-th column implied by $\cS_k$. Let $m := \prod_{l=1}^K m_l$ denote the global maximum conditioning-set size. Suppose that the $m_k$'s and the number of iterations in Algorithm~\ref{alg-main} are fixed and independent of $p$. Then the total computational cost of Algorithm~\ref{alg-main} is 
$$
\mathcal{O}\left( K m^2 p + \sum_{k=1}^K m_k^3 p_k \right).
$$
Assuming $m_k \ll p_k$ for all $k$, this complexity reduces to $\mathcal{O}(Kp)$.
\end{proposition}

This linear scalability seamlessly extends to the settings where the target covariance exhibits a low-rank structure:

\begin{proposition}\label{prop:cost2}
Under the same setting as Proposition \ref{prop-low_rank}, let $n$ denote the rank of $\bfSigma$. Given sparse matrix  subspaces $\{\cS_1,\ldots,\cS_K\}$ and the corresponding maximum conditioning-set sizes $\{m_1,\ldots,m_K\}$ defined in Proposition~\ref{prop:cost}, the total computational cost of the modified Algorithm~\ref{alg-main} is 
$$
\mathcal{O}\left( n \Big(\sum_{k=1}^K m_{-k}\Big) p + n \Big(\sum_{k=1}^K m^2_k \Big) p + \sum_{k=1}^K m^3_k p_k \right),
$$
where $m_{-k} = \prod_{l \neq k} m_l$.
Assuming $m_k \ll p_k$ for all $k$, this complexity reduces to $\mathcal{O}(Knp)$.
\end{proposition}
Apart from the explicit factor $n$ accounting for the rank of $\bfSigma$, the computational complexity in Proposition~\ref{prop:cost2} is essentially equivalent to the $\mathcal{O}(Kp)$ rate established in Proposition~\ref{prop:cost}.

In Algorithm~\ref{alg-main}, the convergence criterion is typically set as $\|\bL^{(t)}_{k} - \bL^{(t-1)}_k\|_F/\|\bL^{(t-1)}_k\|_F < \epsilon$ for all $k$, where $\epsilon$ is a small tolerance, such as $10^{-3}$. In our numerical experience (e.g., see Appendix \ref{app:sim-niter}), convergence is usually achieved in around $5$ iterations across different scenarios, which supports the fixed-iteration assumption used in the preceding computational-complexity analysis. 

The preceding results confirm that the geometry-aware SIC projection provides the key computational advantage that distinguishes our approach from existing dense tensor covariance methods. The remainder of this section addresses two important practical and theoretical considerations regarding the implementation of KSIC. First, the existence of the KSIC projection and the convergence of the BCD algorithm should be carefully investigated, particularly when $\bfSigma$ is singular (Section \ref{subsec:property}). Second, because the exact true multi-way covariance $\bfSigma$ is rarely known in practice, we detail how Algorithm~\ref{alg-main} is extended to enable efficient inference under nonparametric empirical settings (Section \ref{subsec:non-parametric}) and parametric maximum-likelihood settings (Section \ref{subsec:parametric_estimation}).

\subsection{Properties of KSIC-BCD}
\label{subsec:property}
We start by verifying that the KSIC projection behaves exactly as expected in the ideal scenario where the true multi-way covariance matrix $\bfSigma$ is inherently separable. In this case, the pseudo-covariance matrix derived in Proposition~\ref{prop:pseudocov} is strictly proportional to the corresponding true marginal covariance. 

\begin{proposition}\label{prop:separable}
Suppose that $\bfSigma = \bigotimes_{k=1}^K \bfSigma_k$ is separable. Then the $k$-th pseudo-covariance matrix in \eqref{eq:pseudocov} is proportional to the marginal covariance matrix of the $k$-th mode:
\[
\tilde{\bfSigma}_{k} \propto \bfSigma_k.
\]
Consequently, the global KSIC projection in \eqref{eq:ksicproj} completely decouples into $K$ independent, closed-form SIC projections:
\[
\Pi_{\text{KS}}(\bfSigma,\mathcal{S}_{\text{KS}})
=
\bigotimes_{k=1}^K \hat{\bL}_k,
\quad
\text{where } \hat{\bL}_k = \Pi(\bfSigma_k,\mathcal{S}_k).
\]
\end{proposition}

Proposition~\ref{prop:separable} provides foundational justification for our geometric approach. It proves that under separability, the global optimization problem decomposes exactly into marginal SIC projections. Consequently, the KSIC-BCD algorithm converges in exactly one full cycle of mode-wise updates.

When the input covariance matrix $\bfSigma$ is strictly positive definite (i.e., full rank, where its rank $n = p$), the trace term in the forward-KL objective \eqref{eq:ksicproj} bounds the parameters away from infinity, while the log-determinant term acts as a strict barrier preventing the factors from approaching the singular boundary. Consequently, the objective is coercive, and the existence of a global KSIC projection strictly within the non-singular manifold is guaranteed. However, in many multi-way applications, the ambient dimension $p = \prod_{k=1}^K p_k$ is massive, rendering target matrices like the empirical covariance highly rank-deficient ($n \ll p$). Our following theoretical analysis focuses on this challenging low-rank regime. Based on the low-rank representation of $\bfSigma$ introduced in Proposition~\ref{prop-low_rank}, we establish sufficient rank conditions for the validity of the KSIC-BCD block update and the global existence of the KSIC projection. First, we formalize this setting:
\begin{assumption}\label{asmp-exist-2}
$\bfSigma$ is a generic low-rank covariance matrix in the sense that 
$
\bfSigma = \sum_{r=1}^n\bx_r\bx^\top_r,
$
where each $\bx_r$ is sampled from an absolutely continuous distribution supported on $\mathbb{R}^p$.
\end{assumption}

\begin{theorem}[Validity of One-Step Update]\label{thm-step-existence}
Consider a generic multi-way covariance under Assumption \ref{asmp-exist-2}. For a fixed mode $k$, if the fixed factors $\bL_l$ are non-singular for all $l\neq k$, then the resulting SIC projection is strictly positive definite a.s.,
\[
\Pi(\tilde{\bfSigma}_k, \cS_k) \in \mathbf{GL}(p_k),
\]
provided that the rank satisfies $n \geq (m_k+1)/p_{-k}$.
\end{theorem}

Theorem~\ref{thm-step-existence} provides a sufficient condition under which the KSIC-BCD algorithm can successfully update $\mathbf{L}_k$ without encountering singularities. In particular, if $n \geq \max_k\{(m_k+1)/p_{-k}\}$, then Algorithm~\ref{alg-main} can perform all mode-wise updates indefinitely. However, this \textit{algorithmic threshold} only ensures that individual block updates are well-defined; it does not guarantee the existence of a global KSIC minimizer. In fact, the forward-KL objective is much less tractable under a multi-way structure, even in the simpler case of the MLE. While the algorithmic threshold establishes operational stability, deterministically guaranteeing the existence of the global KSIC projection is highly non-trivial. 

In Theorem~\ref{thm-existence}, we establish a formal existence threshold for the general $K$-mode KSIC projection.

\begin{theorem}[Existence of the KSIC Projection]\label{thm-existence}
Consider a generic order-$K$ multi-way covariance matrix $\bfSigma$ under Assumption \ref{asmp-exist-2}. If 
\[
n > \max_{{k_1}\neq {k_2}}\left\{\frac{1}{p^2_{k_1}} + \frac{1}{p^2_{k_2}}\right\}p,
\]
then the minimum of the forward-KL objective in \eqref{eq:ksicproj} exists almost surely and is attained at non-singular factors $\bL_k\in \cS_k\cap \mathbf{GL}(p_k)$ for all $k = 1, \ldots, K$. In particular, when $K = 2$, the condition simplifies to $n > p_1/p_2 + p_2/p_1$.
\end{theorem}

Theorem~\ref{thm-existence} guarantees that the forward-KL objective possesses a well-defined global minimum over the KSIC manifold. For matrix-variate data ($K=2$), this result retrieves the sharp, state-of-the-art sample-size threshold in matrix-variate Gaussian covariance estimation \citep{soloveychik2016gaussian, drton2021existence}. 

For higher-order tensors ($K > 2$), Theorem~\ref{thm-existence} provides the first rigorous guarantee of forward-KL projection existence under multilinear Kronecker structures. We note that this theoretical threshold represents a conservative sufficient condition, as bounding the non-linear operator norms across the coupled modes without bipartite Cholesky symmetry becomes mathematically challenging under high order. 

The threshold in Theorem~\ref{thm-existence} depends only on the dimensions $p_k$, as its derivation relies on bounding the worst-case unconstrained operator norms. In fact, as established next in Proposition~\ref{prop-exist-nest}, the structural inverse-Cholesky sparsity enforced by KSIC effectively regularizes the optimization landscape, enabling the projection to exist at sample sizes significantly below this theoretical bound in practice. We will empirically validate this argument through numerical simulations in Section \ref{sec:sim} and real-world applications in Section \ref{sec:real}.
\begin{proposition}\label{prop-exist-nest}
Consider a generic multi-way covariance matrices $\bfSigma$ under Assumption \ref{asmp-exist-2}. Consider the sparsity patterns $\{\cS^a_k, \cS^b_k\}$ with
\[
\bs^a_i \subseteq \bs^b_i, \quad \forall  i = 1, \ldots, p_k \text{ and } k = 1, \ldots, K.
\]
If the KSIC projection onto $\cS^b_{\text{KS}}$ exists, so does the KSIC projection onto $\cS^a_{\text{KS}}$.
\end{proposition}
The proposition shows that stronger sparsity constraints reduce the effective dimension of the local projection problems and therefore allow the method to operate with lower-rank covariance inputs. In particular, projection onto the diagonal matrix class always exists and only requires $n \geq 1$. For general KSIC projections, the sample-size threshold lies between the threshold for the full-conditioning case and the algorithmic threshold determined by the maximum conditioning-set sizes.

Together, these results confirm that the KSIC projection exists as a well-defined global minimum, establishing justification for the validity of the KSIC-BCD algorithm. Because the conditions are dictated purely by the rank $n$, these guarantees map seamlessly onto nonparametric covariance estimation, where $n$ represents the number of available multi-way samples, as detailed in the following subsection.

\subsection{Nonparametric Estimation (Direct Projection)}
\label{subsec:non-parametric}

\subsubsection{Empirical Estimation and Implicit Data Augmentation}

Consider the empirical setting where we observe independent multi-way samples $\mathcal{X}_1,\ldots,\mathcal{X}_n \in \mathbb{R}^{p_1\times\cdots\times p_K}$ drawn from a mean-zero distribution with covariance matrix $\bfSigma$. Let $\bX_r^{(k)} \in \mathbb{R}^{p_k \times p_{-k}}$ denote the matrix obtained by unfolding tensor $\mathcal{X}_r$ along mode $k$. 

We propose a nonparametric estimator of $\bfSigma$ by applying the KSIC projection directly to the empirical covariance matrix $\bfSigma_\text{emp} = \frac{1}{n}\sum_{r=1}^n \text{vec}(\mathcal{X}_r)\text{vec}'(\mathcal{X}_r)$:
\begin{equation}
    \hat{\bfSigma}^{-1} = \hat\bL\hat\bL^\top, \quad \text{where} \; \hat\bL = \Pi_{\text{KS}}(\bfSigma_\text{emp},\mathcal{S}_{\text{KS}}).
    \label{eq:nonparam}
\end{equation}
Computationally, this projection is carried out efficiently using the low-rank BCD formulation in Proposition \ref{prop-low_rank}. For initialization, we typically use the identity matrix when $n p_{-k} \leq p_k(p_k-1)/2$. Otherwise, we use the naive marginal empirical estimator $\hat{\bfSigma}_{k} = \frac{1}{np_{-k}}\sum_{r=1}^n\bX^{(k)}_{r}\bX^{(k)\top}_{r}$. In addition, given the nonparametric regime, the size of conditioning set corresponds to the level of regularization, whose choice should ideally be data-driven. In practice, we split part of the training data as validation data and select the $m_k$'s that achieve the best validation Log-score. We call this procedure the \emph{adaptive choice} on conditioning set. %

Nonparametric estimation using the KSIC-BCD algorithm is particularly powerful for rank-deficient empirical covariance matrices, especially in the small-$n$, large-$p$ regime where $n \ll \prod_k p_k = p$. As modern data-collection technologies generate richer measurements, explicitly modeling these massive dependence structures is crucial, but researchers are often limited to a very small number of statistically independent replicates. 

The KSIC framework mitigates this severe data scarcity through a mechanism we term \textit{implicit data augmentation}. The Kronecker structure allows observations along modes other than $k$ to act as pseudo-replicates for estimating the mode-$k$ covariance structure. Specifically, in the low-rank pseudo-covariance construction (Proposition~\ref{prop-low_rank}), the raw observations $\bX_{r}^{(k)}$ are deconvolved into $\bX_{r}^{(k)}\bL_{-k}$. By utilizing the estimated precision factors along the other modes to whiten the data, the estimator fully exploits shared information across the tensor. 

Implicit data augmentation has long been recognized conceptually in multi-way covariance estimation. For example, in the classical BCD (or flip-flop) algorithm for the Kronecker MLE \citep{dutilleul1999mle}, samples are similarly deconvolved in each update step:
$
\hat{\bfSigma}_k = \frac{1}{np_{-k}}\sum_{r = 1}^n\bX^{(k)}_r\hat{\bfSigma}^{-1}_{-k}\bX^{(k)\top}_r$.
As a consequence of this augmentation, the minimum sample size required for the standard Kronecker MLE to exist is much smaller than the unrestricted parameter count $p(p+1)/2$. This sample-size threshold has been studied extensively \citep{dutilleul1999mle, lu2005likelihood, srivastava2008models, Manceur2013TensorNormal}; for $K=2$, the state-of-the-art threshold required for the MLE to exist is $n \geq p_1/p_2 + p_2/p_1$ \citep{soloveychik2016gaussian, drton2021existence}, which aligns exactly with the threshold in Theorem \ref{thm-existence}. However, the algorithmic threshold for MLE is $n \geq \max\{p_1/p_2, p_2/p_1\}$, which is much larger than our threshold $n \geq \max_k\{(m_k+1)/p_k\}$ in Theorem \ref{thm-step-existence}. 

Our KSIC framework pushes the frontier in two critical ways. 
First, the structural sparsity further regularizes the estimator, introducing more robustness under data scarcity. Equivalent, it reduces the sample size required to obtain a valid, positive-definite estimator in both computational sense (Theorem \ref{thm-step-existence}) and theoretical sense (Proposition \ref{prop-exist-nest}). 

Second, we provide, to our knowledge, the first theoretical result that explicitly reflects this implicit data augmentation directly within a finite-sample convergence rate, which we detail in Section~\ref{subsec:np-asym}.

\subsubsection{Practical Regularization in Extreme Data Scarcity}
\label{subsec:np-scarce}

In dimensional imbalances where the sample size $n$ falls below the existence threshold in Theorem \ref{thm-existence}, there are two practical strategies to resolve this. 

The first approach is to dynamically trim the conditioning sets associated with $\cS_k$. By enforcing a sparser conditioning structure, the dimensions of the local regression problems are reduced, which helps ensure that the corresponding SIC projections remain well defined. In our experiments, we find that even near the algorithmic threshold $\max_k\{(m_k+1)/p_k\}$, KSIC-BCD generally converges quickly and yields stable estimation with simple constraints on the maximum conditioning-set size $m_k$, which offers a practical balance between feasibility and accuracy. Crucially, in many multi-way applications, the product of the ambient dimensions $p_{-k}$, is much larger than the local conditioning-set size $m_k$, and the fractional $\max_k\{(m_k+1)/p_k\}$ is often strictly less than 1, meaning that the KSIC-BCD procedure can yield a stable covariance estimator even from a single tensor observation ($n=1$). This represents a substantial improvement over the Kronecker MLE, which is not even computationally feasible once $n < \max_k\{p_k/p_{-k}\}$.

The second approach is to explicitly increase the effective rank of the pseudo-covariance matrix by adding a small diagonal (nugget) term. This strategy is widely used in covariance-related computations and, in the present setting, acts analogously to increasing the effective sample size $n$. In the numerical experiments presented in Sections~\ref{sec:sim} and~\ref{sec:real}, we apply minimal nugget regularization whenever $n$ falls below the existence threshold. The results show that the KSIC estimator consistently outperforms the Kronecker MLE under severe data scarcity, often producing accurate and stable covariance estimates even from a single independent replicate ($n=1$).

\subsubsection{Asymptotic Efficiency}
\label{subsec:np-asym}

Beyond the data-scarce regime, we establish the asymptotic consistency of the KSIC covariance estimator as the sample size $n$ increases. We provide a step-wise result demonstrating that the SIC projections computed from the empirical data tightly bound the corresponding population SIC projections based on the true multi-way covariance $\bfSigma$. 

Suppose that $\cX_1, \ldots, \cX_n \in \mathbb{R}^{p_1\times \cdots\times p_K}$ are independent order-$K$ replicates with $\text{vec}(\cX_r) \sim \mathcal{N}(\bfzero, \bfSigma)$, where $\lambda_{min}(\bfSigma) > 0$. Let $\bL_{-k}\bL_{-k}^\top$ be the fixed working precision matrix for all modes except $k$. Construct the empirical pseudo-covariance according to Proposition~\ref{prop-low_rank} as $\bar{\bfSigma}_{k} = \frac{1}{np_{-k}}\sum_{r=1}^n\bX_r^{(k)}\bL_{-k}\bL_{-k}^\top\bX_r^{(k)\top}$, and population pseudo-covariance according to Proposition~\ref{prop-low_rank} as $\tilde{\bfSigma}_{k}$. Let $\hat{\bL}_{k} = \Pi(\bar{\bfSigma}_{k}, \cS_k)$ be the empirical one-step SIC estimator, and let $\tilde{\bL}_{k} = \Pi(\tilde{\bfSigma}_{k}, \cS_k)$ be the corresponding population one-step SIC estimator.

To ensure regularity of the objective, we impose the following assumption on the SIC approximation error.
\begin{assumption}\label{asmp-rate-1}
Let $|\cS_k|$ denote the number of nonzero elements in $\bL_k\bL_k^\top$ allowed by the sparse matrix  subspace $\cS_k$. The SIC estimator satisfies
\begin{equation} 
\|\tilde{\bfSigma}_{k}- (\tilde{\bL}_{k}\tilde{\bL}_{k}^\top)^{-1}\|_{\max}|S_k|^{1/2} \leq \frac{\lambda_{min}(\bfSigma)\tr(\bL_{-k}\bL^\top_{-k})}{8p_{k}}. 
\end{equation}
\end{assumption}
This assumption requires the SIC projection to approximate the pseudo-covariance matrix with sufficiently small error and primarily serves as a technical condition for the proof. In practice, for common covariance families, the SIC approximation error can often be substantially reduced by moderately enlarging the conditioning sets. An illustrative example is provided in Appendix~\ref{app:np-asym}.

\begin{theorem}\label{thm-rate}
Define $\bE_{(i^{(k)},j^{(k)})} = \bL^\top_{-k}\bfSigma^{(-k)}_{(i^{(k)},j^{(k)})}\bL_{-k}$ as the discrepancy matrix between the true sub-covariance matrix and the working precision factor. Under Assumption \ref{asmp-rate-1}, the estimated precision matrix satisfies the following bound for some constant $C_1 >0$:
\begin{equation}\label{eq:main-rate}
    \|\hat{\bL}_{k}\hat{\bL}^\top_{k} - \tilde{\bL}_{k}\tilde{\bL}^\top_{k}\|_F = 
    \mathcal{O}_p\left( \left( C_1\max\limits_{i^{(k)},j^{(k)}}\|\bE_{(i^{(k)},j^{(k)})}\|_2\sqrt{\frac{\log p_{k}}{np_{-k}}} + 
    \|\tilde{\bfSigma}_{k}- (\tilde{\bL}_{k}\tilde{\bL}_{k}^\top)^{-1}\|_{\max} \right) \big|\cS_{k}\big|^{1/2} \right).
\end{equation}
\end{theorem}

Theorem~\ref{thm-rate} provides a finite-sample error bound for the estimation step. The bound fundamentally relies on two sources of error, both of which are scaled by the complexity of the chosen sparsity pattern $|\cS_{k}|^{1/2}$ (which typically scales as $\mathcal{O}(p_k^{1/2})$). The first internal term reflects the statistical sampling error, where the denominator $n p_{-k}$ explicitly quantifies the \textit{implicit data augmentation} provided by the multi-way structure. The multiplier $\max\|\bE\|_2$ dictates how accurately the fixed factors $\bL_{-k}$ capture the true cross-mode dependence; a highly accurate standby factor drives this discrepancy down, accelerating the convergence of the mode-$k$ estimate. This formalizes why iterative deconvolution significantly outperforms the naive marginal estimator (which implicitly assumes $\bL_{-k} = \bI_{p_{-k}}$). The complexity term $|\cS_k|^{1/2}$ scales tightly with the prescribed sparsity and is typically at the order of $p_k^{1/2}$.
The second internal term, $\|\tilde{\bfSigma}_{k}- (\tilde{\bL}_{k}\tilde{\bL}_{k}^\top)^{-1}\|_{\max}$, isolates the deterministic approximation error introduced by the SIC spatial projection itself. This error is independent of the sample size and naturally vanishes if the true pseudo-covariance perfectly respects the sparsity pattern $\cS_k$. 

This step-wise rate is highly consistent with established results in sparse precision estimation. For instance, if the standby factors are uniformly bounded and $|\cS_k| \sim \mathcal{O}(p_k)$, the estimation error precisely matches the optimal rate $\mathcal{O}_p\left(\sqrt{p_k \log p_k / (np_{-k})}\right)$ recovered in recent Kronecker graphical Lasso \citep{lyu2019tensor} and related literature \citep{cai2016estimating}. For the exactly separable case, where the global KSIC projection admits a closed form (Proposition~\ref{prop:separable}), fully specified global convergence rates can be derived. We detail these rigorous asymptotic conditions and global results in Appendix~\ref{app:np-asym}.

\subsection{Parametric Estimation via Nested Projection}
\label{subsec:parametric_estimation}

In many scientific applications, the covariance structure is governed by a generative model $\bfSigma_{\bftheta}$ parameterized by a low-dimensional vector $\bftheta \in \Theta \subseteq \mathbb{R}^d$ (e.g., a Mat\'ern spatial field or a non-separable spatio-temporal process). Direct maximum likelihood estimation (MLE) is often computationally prohibitive due to the $\mathcal{O}(p^3)$ cost of decomposing dense covariance matrices. 

We circumvent this bottleneck by defining our estimator as the solution to a nested, double-forward-KL projection. Specifically, we estimate $\bftheta$ by minimizing the divergence from the empirical data distribution to a structured KSIC approximation of the parametric model:
\begin{align}
    \hat{\bftheta} &= \mathop{\mathrm{arg\,min}}_{\bftheta \in \Theta} \; \mathrm{KL} \left( \mathcal{N}(\bfzero, \bfSigma_\text{emp}) \;\|\; \mathcal{N}(\bfzero, (\hat\bL_{\bftheta}\hat\bL_{\bftheta}^\top)^{-1}) \right), \label{eq:outer} \\
    \text{subject to} \quad \hat\bL_{\bftheta} &= \mathop{\mathrm{arg\,min}}_{\bL \in \mathcal{S}_{\text{KS}}} \; \mathrm{KL}\left( \mathcal{N}(\mathbf{0}, \bfSigma_{\bftheta}) \;\|\; \mathcal{N}(\bfzero, (\bL\bL^\top)^{-1}) \right) = \Pi_{\text{KS}}(\bfSigma_{\bftheta},\mathcal{S}_{\text{KS}}). \label{eq:inner}
\end{align}
The inner minimization \eqref{eq:inner} represents the spatial M-projection of the parametric model $\bfSigma_{\bftheta}$ onto the scalable KSIC manifold $\mathcal{S}_{\text{KS}}$. The outer minimization \eqref{eq:outer} matches this projected model to the empirical data. While the outer objective is mathematically equivalent to maximizing the log-likelihood of the data under the approximated precision matrix, framing it as a Gaussian-to-Gaussian KL divergence highlights the geometric symmetry of the nested framework.

If we use $J(\bftheta)$ to denote the projected negative log-likelihood objective in \eqref{eq:outer_loss}, and $J^0(\bftheta)$ for the exact objective, 
\begin{equation}
\begin{split}
J(\bftheta) &=
\tr\left(
\bfSigma_{\mathrm{emp}}
\hat\bL_{\bftheta}\hat\bL_{\bftheta}^\top\right)
-
\log\det\left(
\hat\bL_{\bftheta}\hat\bL_{\bftheta}^\top
\right),\\
J^0(\bftheta) &=
\tr\left(
\bfSigma_{\mathrm{emp}}\bfSigma_{\bftheta}^{-1}
\right) 
-
\log\det\left(
\bfSigma_{\bftheta}^{-1}
\right).\\
\end{split}
\end{equation}
The discrepancy between the projected estimator $\hat{\bftheta}$ and the exact MLE $\hat{\bftheta}_{\mathrm{MLE}}$ is therefore governed by the accuracy of the KSIC projection embedded in the inner problem \eqref{eq:inner}. The following theorem formalizes this relationship by bounding the approximation error induced by replacing the exact covariance model $\bfSigma_{\bftheta}$ with its KSIC projection.
\begin{theorem}[Stability under inner KL approximation]
\label{thm:parametric_stability}
Suppose that the exact objective satisfies the quadratic-growth condition for some $\mu > 0$:
$$
J^0(\bftheta)-J^0(\hat{\bftheta}_{MLE})
\geq
\frac{\mu}{2}
\|\bftheta-\hat{\bftheta}_{MLE}\|_2^2,
\quad
\bftheta\in\Theta.
$$
If the inner objective is uniformly accurate in the sense that for some $\delta\geq0$:
$$
\sup_{\bftheta\in\Theta}
\mathrm{KL}\left(
\mathcal N\left(\bfzero,\bfSigma_{\bftheta}\right)
\,\|\,
\mathcal N\left(
\bfzero,(\hat\bL_{\bftheta}\hat\bL_{\bftheta}^\top)^{-1}
\right)
\right)
\leq\delta.
$$
Then
$$
\|\hat{\bftheta}-\hat{\bftheta}_{MLE}\|_2
\leq
2\left[
\frac{
2\|\bD_{\bftheta}-\bI\|_F\big(\delta+\sqrt{\delta(1+\delta)}\big)+\delta
}{\mu}
\right]^{1/2},
$$
where $\bD_{\bftheta} =
\bfSigma_{\bftheta}^{-1/2}
\bfSigma_{\mathrm{emp}}
\bfSigma_{\bftheta}^{-1/2}$.
\end{theorem}

The theorem shows that the distance between $\hat{\bftheta}$ and $\hat{\bftheta}_{\mathrm{MLE}}$ decomposes into two distinct components. The first component is driven by the statistical estimation error $\bD_{\bftheta}-\bI$, which typically vanishes as the sample size $n \to \infty$. The second component, represented by $\delta$, measures the geometric approximation error induced by projecting the true parametric model onto the KSIC manifold. 
Because $\delta$ represents a structural approximation, it does not naturally vanish with $n$. For a non-separable generative family $\bfSigma_{\bftheta}$, the exact inverse Cholesky factor cannot be perfectly factorized into a Kronecker product. Consequently, $\delta$ encapsulates the error from two constraints: the structural bottleneck of the Kronecker product, and the imposed sparsity of the conditioning sets $\cS_k$. 

It is important to emphasize that our approach fundamentally differs from naively assuming the generative model $\bfSigma_{\bftheta}$ is separable. By  projecting the true non-separable model, our estimator acts as a multi-way, Kronecker-structured Vecchia approximation. Expanding the sizes of the conditioning sets $\cS_k$ allows the user to systematically eliminate the sparsity-induced error, isolating only the irreducible structural mismatch of the optimal Kronecker surrogate. 
For many complex spatial and spatio-temporal processes, this optimal Kronecker backbone captures the vast majority of the true correlation structure, keeping the baseline $\delta$ tightly bounded. This establishes a principled, user-controlled trade-off between computational scalability and statistical fidelity. The empirical results in Sections~\ref{sec:sim-KLD} and~\ref{sec:sim-par} reinforce this theory, demonstrating that the KSIC projection accurately captures complex non-separable parametric families and yields highly accurate parameter estimation despite the Kronecker bottleneck.

Solving the outer optimization \eqref{eq:outer} via explicit gradient descent is challenging. Applying the chain rule to the outer objective yields:
\begin{equation}
\textstyle
\nabla_{\bftheta} \mathrm{KL} \left( \mathcal{N}(\bfzero, \bfSigma_\text{emp}) \;\|\; \mathcal{N}(\bfzero, (\hat\bL_{\bftheta}\hat\bL_{\bftheta}^\top)^{-1}) \right) = \tr\left((\hat{\bL}^{-\top}_{\bftheta} - \bfSigma_\text{emp}\hat{\bL}_{\bftheta})\frac{d\hat{\bL}^{\top}_{\bftheta}}{d\bftheta}\right).
\end{equation}
The primary difficulty lies in the Jacobian term $d\hat{\bL}^{\top}_{\bftheta}/d\bftheta$. Because the mapping $\bftheta \mapsto \hat\bL_{\bftheta}$ is defined implicitly via the BCD projection algorithm, an analytical closed-form gradient is generally unavailable. 

\begin{algorithm}[htbp]
\caption{KSIC Projected-Likelihood Parametric Estimation}
\KwInput{Tensor replicates $\{\cX_r\}_{r=1}^n$, parametric model covariance function $\bfSigma_{\bftheta}$, sparsity sets $\{\cS_1, \ldots, \cS_K\}$, optimization routine $\mathcal{A}$ (e.g., L-BFGS-B).}
\KwOutput{Projected MLE $\hat{\bftheta}$.}
\begin{algorithmic}[1]
\STATE \textbf{Define Objective Oracle $J(\bftheta)$:}
\STATE \quad 1. Compute KSIC projection $\hat{\bL}_{\bftheta}= \bigotimes_{k=1}^K \hat{\bL}_k  = \Pi_{\text{KS}}(\bfSigma_{\bftheta},\cS_{\text{KS}})$ using Algorithm \ref{alg-main}. 
\STATE \quad 2. Evaluate projected negative log-likelihood:
\[ \textstyle J(\bftheta) = \frac{1}{n}\sum_{r=1}^n \big\| (\bigotimes_{k=1}^K \hat\bL_k^\top ) \text{vec}(\cX_r) \big\|_2^2 - 2 \sum_{k=1}^K p_{-k} \sum_{i=1}^{p_k} \log(\hat\bL_{k,ii}). \]
\STATE \textbf{Optimize:}
\STATE Pass the oracle $J(\bftheta)$ to the optimizer $\mathcal{A}$ to find $\hat{\bftheta} = \arg\min_{\bftheta \in \Theta} J(\bftheta)$. Gradients for $\mathcal{A}$ can be supplied via finite differences.
\RETURN $\hat{\bftheta}$
\end{algorithmic}
\label{alg-par}
\end{algorithm}

Rather than deriving explicit analytical gradients, we exploit the linear $\mathcal{O}(Kp)$ computational efficiency of the inner KSIC projection (Proposition~\ref{prop:cost}). As outlined in Algorithm~\ref{alg-par}, we pass the raw multi-way data replicates $\{\cX_r\}_{r=1}^n$ directly as inputs to avoid the prohibitive $\mathcal{O}(p^2)$ memory and computational footprint of constructing the full empirical covariance matrix $\bfSigma_\text{emp}$. 
Crucially, by exploiting the properties of Kronecker products, the matrix-vector multiplication in step 2 of the oracle reduces to a sequential sequence of mode-$k$ tensor-matrix products across the factors: $\cX_r \times_1 \hat\bL_1^\top \times_2 \hat\bL_2^\top \cdots \times_K \hat\bL_K^\top$. Similarly, because each factor matrix $\hat\bL_k$ remains lower triangular, the log-determinant safely decomposes into a linear sum of their diagonal entries. 

This projected-likelihood formulation yields a highly scalable objective evaluation that integrates seamlessly with standard numerical optimization routines like L-BFGS-B \citep{byrd1995limited}. The effectiveness of this approach is demonstrated empirically in Section~\ref{sec:sim-par}, confirming its accuracy even under complex non-separable settings. While we focus on numerical gradient approximation for its simplicity and robustness, exact gradients can alternatively be derived using the implicit function theorem; we detail this advanced implicit-differentiation approach in Appendix~\ref{subsec:implicit_diff} for interested readers.

\section{Simulation Studies}\label{sec:sim}

We conduct extensive simulation experiments to study the asymptotic behavior of KSIC and its advantages over existing methods. Section~\ref{sec:sim-KLD} evaluates the approximation error of the KSIC projection. Section~\ref{sec:sim-np} presents results for nonparametric estimation, and Section~\ref{sec:sim-par} presents results for parametric estimation. Further experimental details are provided in Appendix~\ref{app:sim-detail}.

\subsection{KSIC Projection}\label{sec:sim-KLD}

To understand the mechanism of the KSIC framework, it is essential to assess how well the KSIC projection $\Pi_{\text{KS}}(\bfSigma, \cS_{\text{KS}})$ approximates a general multi-way covariance matrix $\bfSigma$. In this section, we conduct experiments to compare the approximation accuracy of different methods. In addition to KSIC, we include the standard SIC approach as a baseline.

In our experiments, we construct a family of multi-way covariance matrices with $K = 2$ modes. The first mode is defined on $[0,1]\times[0,1] \subset \mathbb{R}^2$, and the second mode is defined on $[0,1] \subset \mathbb{R}$. The dimensions are $30$ and $100$ respectively, so $p = 30\times 100 = 3{,}000$. After randomly ordering the indices, the covariance matrix is generated from a covariance function evaluated at the sampled locations in the two modes. Specifically, we use a variant of the stationary non-separable covariance functions of \cite{Cressie1999}, example~6. For $x_1, x_2 \in \mathbb{R}^2$ and $y_1, y_2 \in \mathbb{R}$, the covariance between random variables indexed by the multi-way coordinates $[x_1,y_1]$ and $[x_2,y_2]$ is defined as
\begin{equation}\label{eq-sim-cov}
\cov([x_1, y_1], [x_2, y_2]|\bftheta) = \sigma^2\exp\big(-\phi(\|x_1 - x_2\|^2 + \|y_1 - y_2\| + \delta\|x_1 - x_2\|^2\|y_1 - y_2\|)\big),
\end{equation}
where $\bftheta := \{\sigma^2, \phi, \delta\}$. This covariance function is positive definite for positive parameter values $\bftheta$, and it becomes separable when $\delta = 0$.

To ensure a fair comparison across methods, we define nearest-neighbor conditioning sets in a comparable manner. After randomly ordering the simulated points, for KSIC, the marginal conditioning sets are chosen as (previously ordered) nearest-neighbor sets within each mode. For SIC, the conditioning set for each multi-way point is chosen from (previously ordered) nearest neighbors in the full product space $\mathbb{R}^2 \times \mathbb{R}$. Given a marginal conditioning-set size $m$ for KSIC, we compare SIC using $m$ conditioning elements (\emph{SIC-$m$}) and $m^2$ conditioning elements (\emph{SIC-$m^2$}). The former is comparable to KSIC (\emph{KSIC-$m$}) in computational cost, whereas the latter is comparable to \emph{KSIC-$m$} in the total size of the induced conditioning set. We vary both the covariance parameters $\bftheta$ and the neighbor size $m$ in the experiments.

\begin{figure}[htbp]
\centering
    \begin{subfigure}{.48\textwidth}
	\centering
 	\includegraphics[width =.98\linewidth]{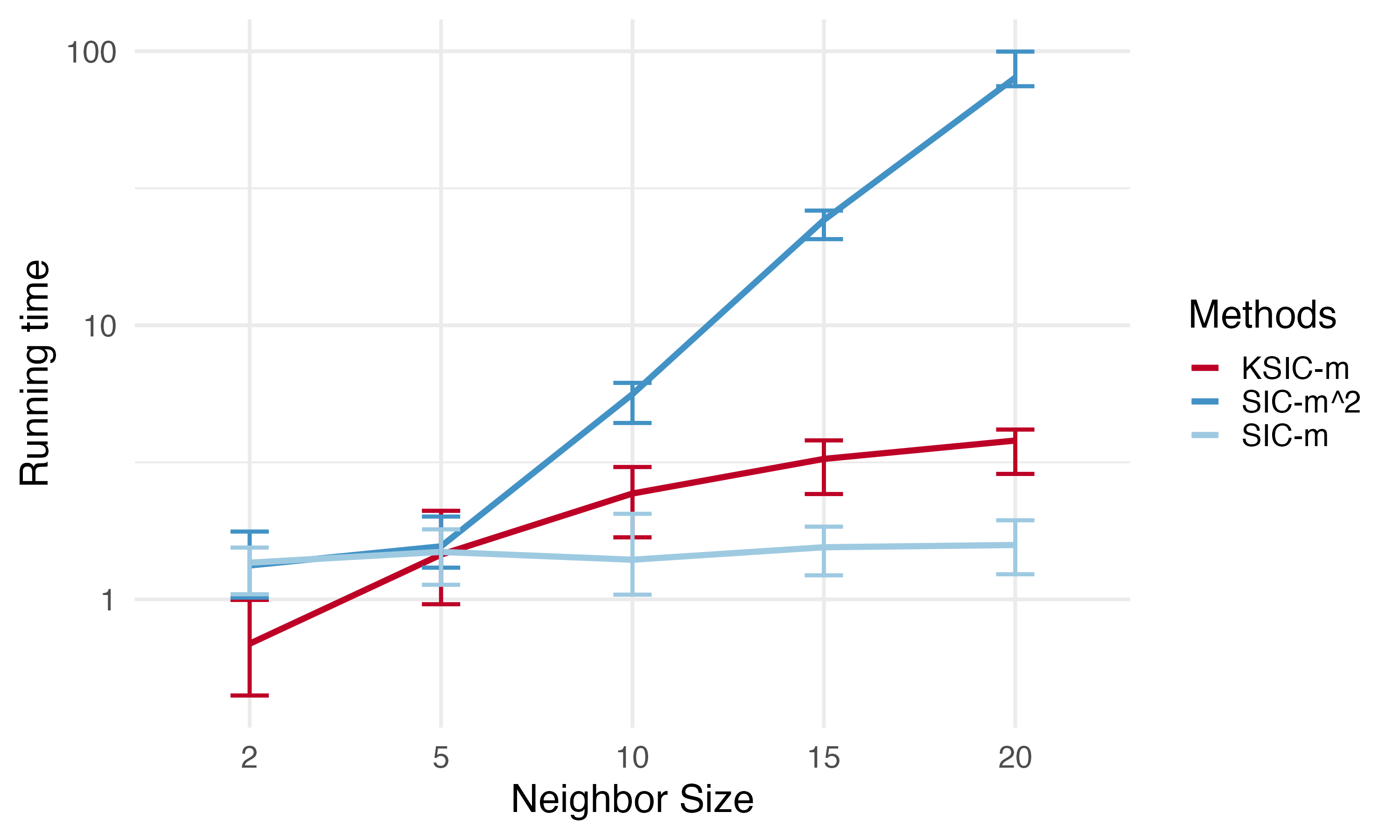}
	\caption{Running time for different algorithms}
	\end{subfigure}%
    \hfill
	\begin{subfigure}{.48\textwidth}
	\centering
  	\includegraphics[width =.98\linewidth]{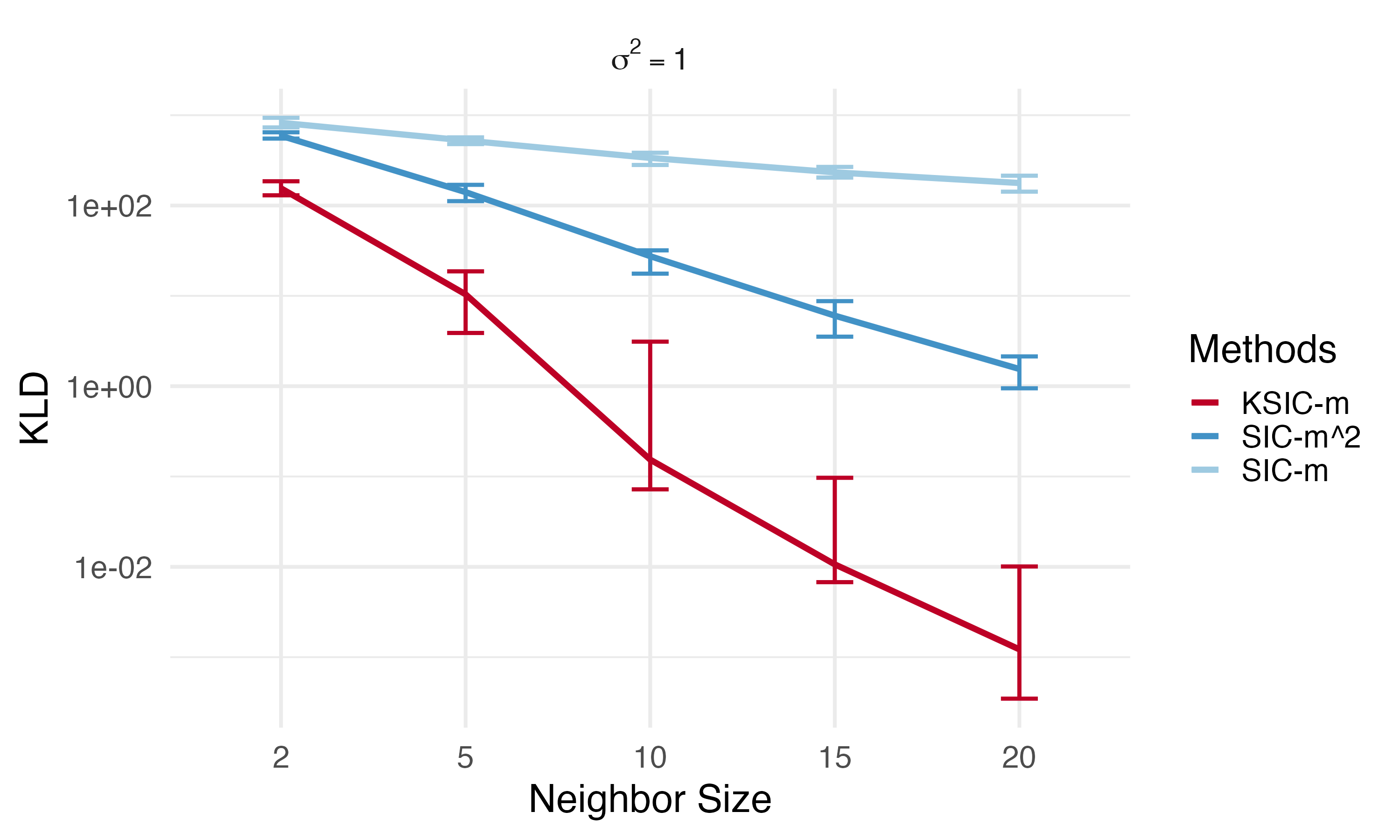}
	\caption{KL Divergence when $\bftheta = \{1, 1, 0\}$}
	\end{subfigure}%
\vspace{0.5em}
\centering
	\begin{subfigure}{.48\textwidth}
	\centering
 	\includegraphics[width =.98\linewidth]{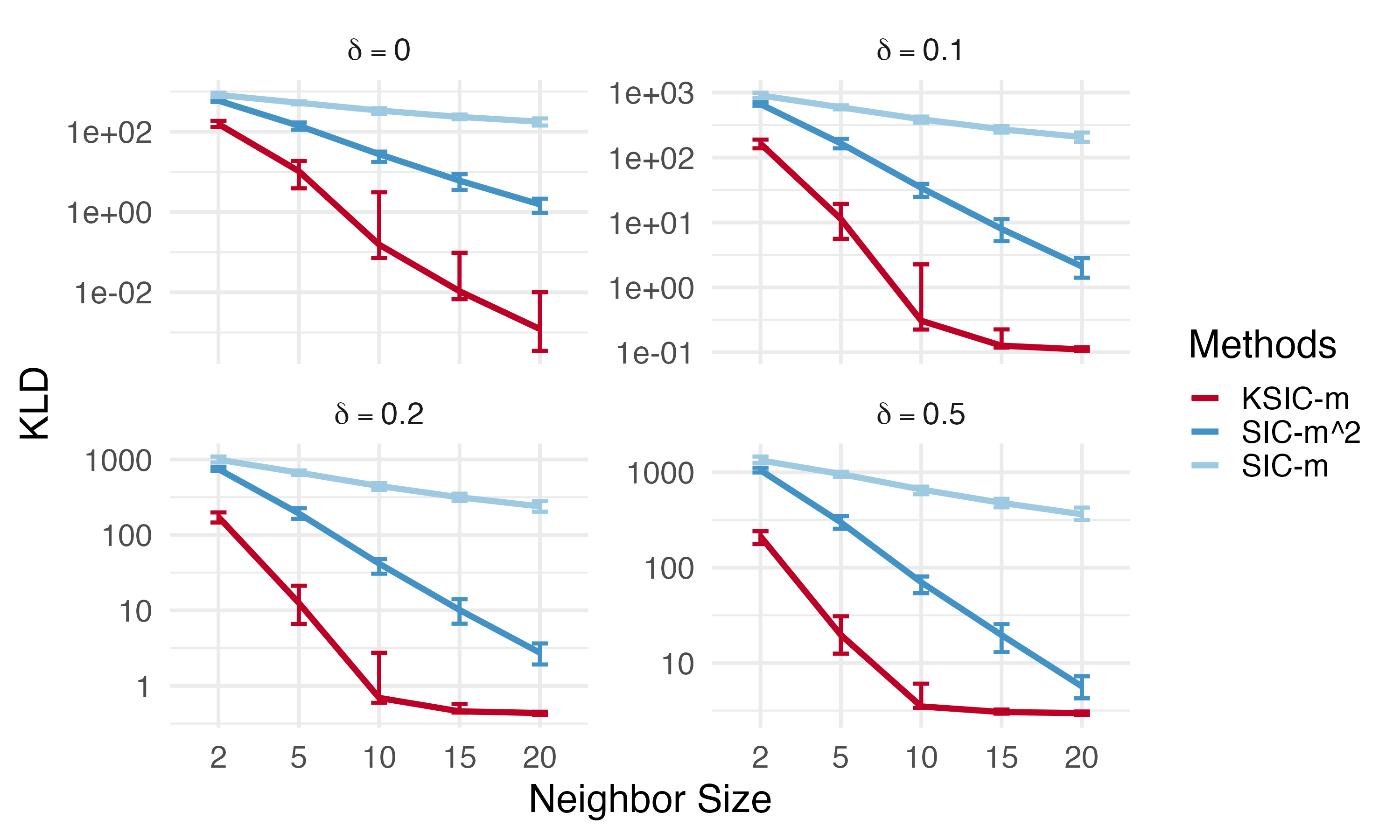}
	\caption{KL Divergence across $\delta$}	
	\end{subfigure}%
    \hfill
	\begin{subfigure}{.48\textwidth}
	\centering
  	\includegraphics[width =.98\linewidth]{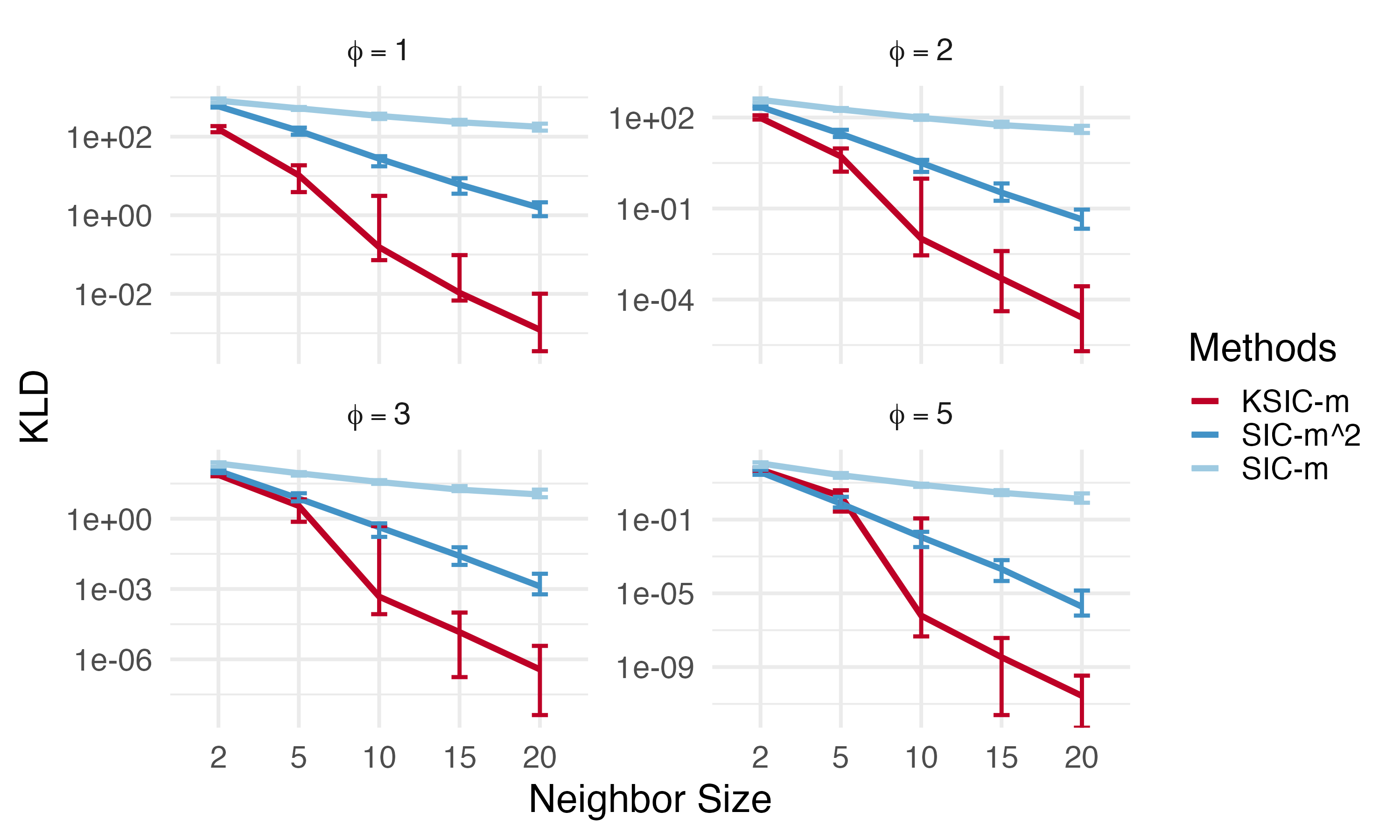}
	\caption{KL Divergence across $\phi$}
	\end{subfigure}%
  \caption{KL Divergence between the projection and true covariance}
\label{fig:KLD}
\end{figure}
In Figure~\ref{fig:KLD} (a), we compare the running time among the approaches. In the other sub-figures, we use $\bftheta = \{1, 1, 0\}$ as the baseline parameter setting, vary $\phi$ in subfigure (c), and $\delta$ in subfigure (d). The y-axis reports the KL divergence between $\bfSigma_{\bftheta}$ and either $\Pi_{\text{KS}}(\bfSigma_{\bftheta}, \cS_{\text{KS}})$ or $\Pi(\bfSigma_{\bftheta}, \cS)$, while the x-axis represents the marginal conditioning-set size $m$. Each panel corresponds to one value of the parameter varied in the corresponding sub-figure. Across all experiments, the KL divergence decreases as $m$ increases, and \emph{SIC-$m^2$} consistently outperforms \emph{SIC-$m$}, confirming the nesting property of SIC approximations, which says that enlarging the conditioning sets decrease the KL divergence \citep{Guinness2016a}. We further highlight that \emph{KSIC} consistently outperforms \emph{SIC-$m^2$}, even though the two approaches use conditioning sets of comparable induced size. Notably, in Figure~\ref{fig:KLD}(c), \emph{KSIC} continues to perform better under a strongly non-separable covariance setting with $\delta = 0.5$ and a large conditioning-set size $m^2 = 400$, even after the performance curve begins to plateau. For reference, nearest-neighbor conditioning sets of size $10$ to $20$ are often considered effective for SIC approximations of covariance matrices of size around $1000 \times 1000$ \citep{zhan2025neural}. In multi-way applications, however, maintaining an $m^K$-sized conditioning set can make the standard SIC approach computationally prohibitive, since its cost grows on the order of $\mathcal{O}(m^{3K}p)$. In contrast, the corresponding \emph{~KSIC} complexity grows on the order of $\mathcal{O}(Km^{2K}p + Km^3p_k)$, making it substantially more efficient even when $K = 2$.

\subsection{Nonparametric estimation}\label{sec:sim-np}
For the nonparametric estimation experiments, we compare KSIC (\emph{KSIC}) with several competing approaches representing different estimation paradigms: 1) naive marginal covariance estimation (\emph{Naive}); 2) naive marginal covariance estimation regularized by Vecchia's approximation (\emph{Naive-Vecchia}); 3) the classical Kronecker MLE of \cite{dutilleul1999mle} (\emph{MLE}); 4) the matrix minimum covariance determinant estimator of \cite{mayrhofer2025robust} (\emph{MMCD}); 5) graphical Lasso estimation (\emph{Glasso}); and 6) robust estimator using bandable covariance from \cite{zhang2023covariance} (\emph{Robust}). The \emph{Naive} approach treats observations along those modes as independent replicates, and computes the sample covariance as the marginal covariance estimate. It is widely used in scientific applications because it is simple to implement and straightforward to interpret. 
\emph{MMCD} is a covariance estimation procedure for matrix-variate data robust to outlier, equivalent to maximizing a weighted matrix-normal likelihood under constraints with adaptively chosen weights. 
The \emph{Glasso} approach is a variant of the method in \cite{lyu2019tensor}, focusing on learning a graphical representation on each mode (see Appendix \ref{app:GGM}). 
The \emph{Robust} approach  is a distribution-free regularized covariance estimation methods for matrix data under a separability condition and a bandable covariance structure..

Most of the approaches for comparisons are special cases of the \emph{KSIC}. For example, the \emph{MLE} can be viewed as a dense version of \emph{KSIC} using full conditioning sets (i.e., $m_k = p_k$). The \emph{Naive-Vecchia} estimator can also be viewed as a one-step \emph{KSIC}, initializing $\bfSigma_k$ with identity matrices. The \emph{Naive} estimator combines these simplifications, with identity-initialization, one-step update, and full conditioning sets. 
In terms of the algorithmic thresholds, the \emph{Naive} estimator %
has the same threshold as the \emph{MLE}, and \emph{Naive-Vecchia} has the same threshold as \emph{KSIC}. 
For \emph{MLE} and \emph{MMCD}, results under sample size threshold is unavailable as their implementation prohibits small sample size.  For \emph{Naive-Vecchia}, \emph{Naive-Vecchia} and \emph{KSIC}, nugget regularization is applied when the sample size is below the existence threshold.
Additionally for \emph{Naive-Vecchia} and \emph{KSIC}, we choose $m$ adaptively by using $1/3$ training data for validation (see Section~\ref{subsec:non-parametric}). For settings with $n<3$ replicates, we set $m_k = 5$ by default.
All experiments are repeated under multiple random seeds for the data-generation process to account for simulation variability.

Specifically, we simulate data with $K=2$ modes from a separable full covariance matrix $\bfSigma$. We fix $p_1$ to $30$, vary $p_2$ and $n$. The marginal covariances are exponential covariance matrices generated from $p_k$ locations in $[0,1]\times[0,1] \subset \mathbb{R}^2$ with random ordering. We set the marginal variance to $1$ and range parameter to $0.5$ in this section and present results for other geometries in Appendix~\ref{app:sim-np-add}. For the SIC-based approaches, \emph{KSIC} and \emph{Naive-Vecchia}, the conditioning set is defined by the $m_k$-nearest neighbors in the marginal space. Estimation performance is evaluated using two metrics: 1) the Frobenius norm of the difference between the estimator and the true marginal covariance $\bfSigma_1$ (Frobenius loss); and 2) the negative log-likelihood of the fitted full covariance evaluated on the data (Log-score). The Frobenius loss measures the accuracy of marginal covariance estimation, whereas the Log-score is a strictly proper scoring rule \citep[e.g.,][]{Gneiting2014} reflecting the global goodness of fit. %
Representative results are shown in Figures~\ref{fig:np-dimen} and~\ref{fig:np-rep}. 
In the Figure \ref{fig:np-dimen} (b) and \ref{fig:np-rep} (b), a red dashed horizontal line denotes an upper truncation limit for methods yielding excessively high log-scores.

\begin{figure}[htbp]
\centering
	\begin{subfigure}{.48\textwidth}
	\centering
  	\includegraphics[width =.98\linewidth]{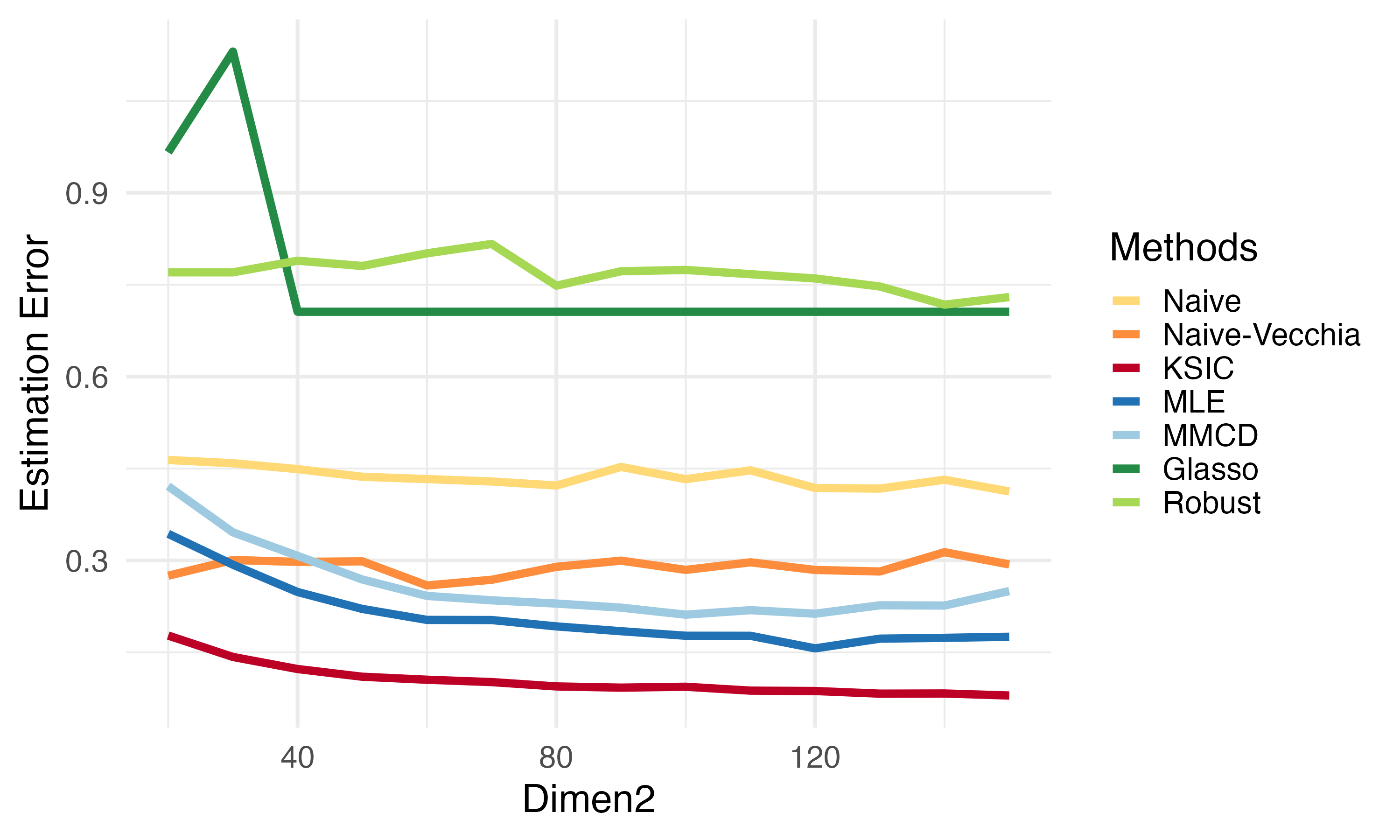}
	\caption{Frobenius loss}
	\end{subfigure}%
\hfill
	\begin{subfigure}{.48\textwidth}
	\centering
 	\includegraphics[width =.98\linewidth]{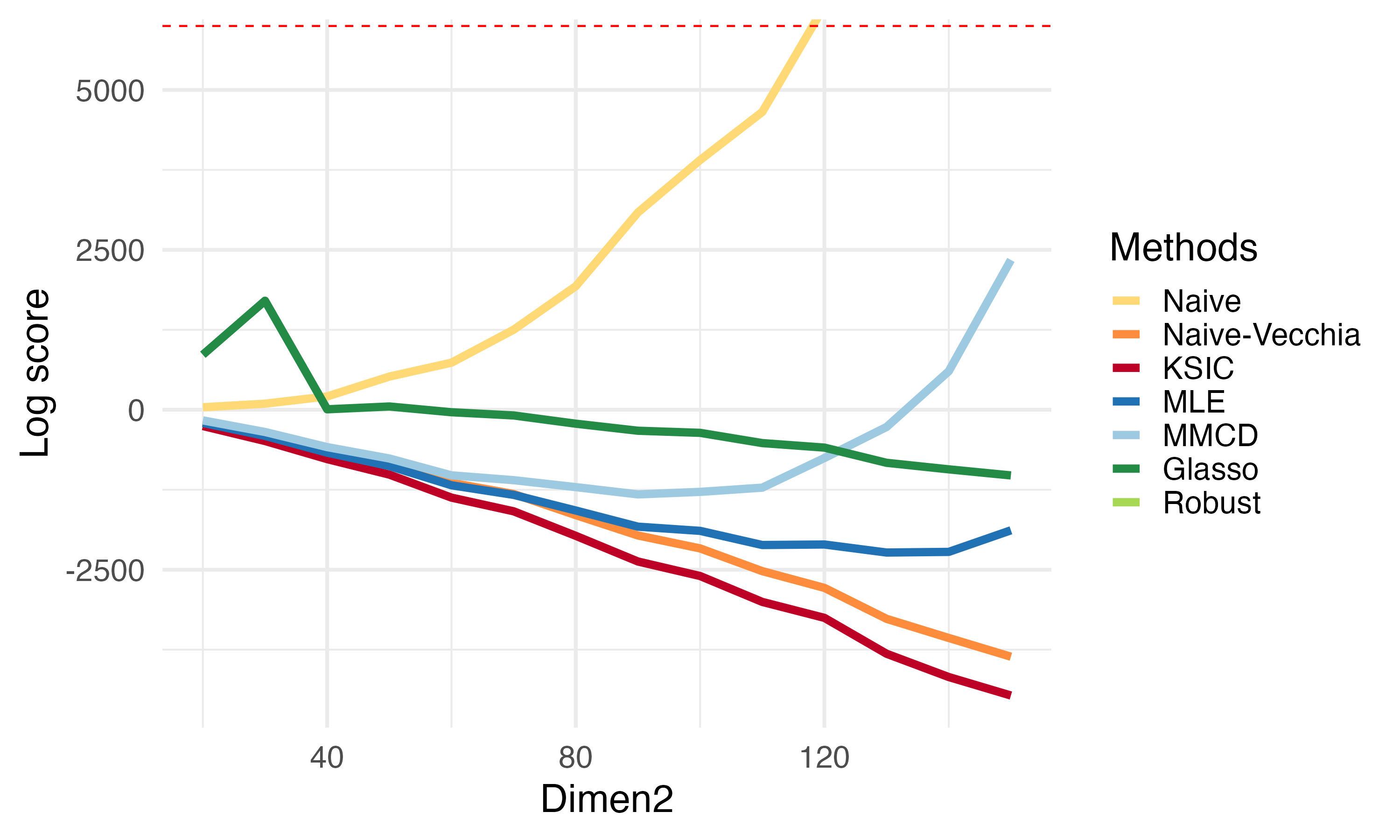}
	\caption{Log-score}
	\end{subfigure}%
\vspace{0.5em}
	\begin{subfigure}{.48\textwidth}
	\centering
 	\includegraphics[width =.98\linewidth]{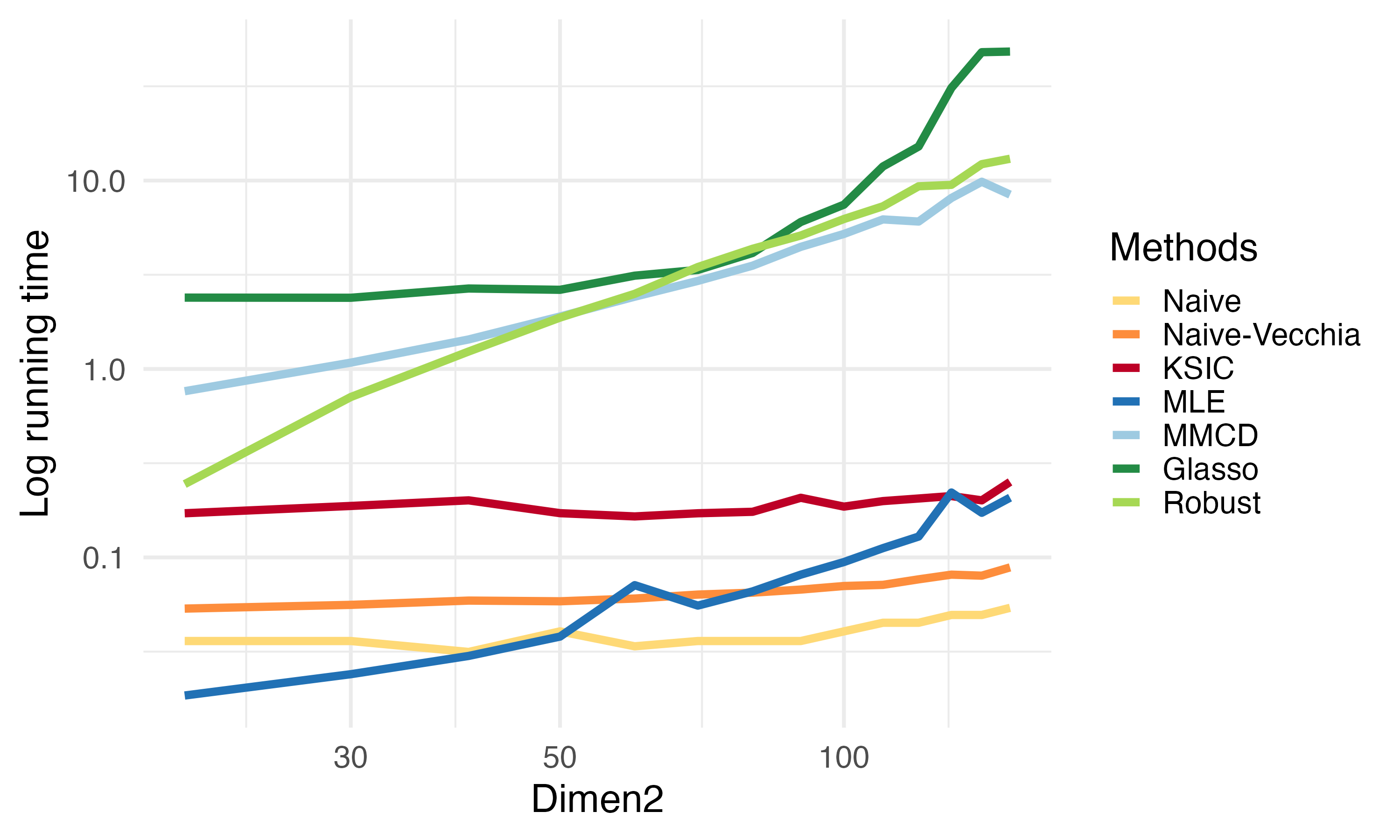}
	\caption{Running time}
	\end{subfigure}%
\caption{Estimation error, log-score, running time for different approaches v.s.\ ambient dimension.}
\label{fig:np-dimen}
\end{figure}
Figure~\ref{fig:np-dimen} illustrates how estimation performance changes with the ambient dimension $p_2$, with $n = 10$ fixed. Figures~\ref{fig:np-dimen}(a) and (b) compare the Frobenius loss and Log-score across competing methods, respectively, while Figure~\ref{fig:np-dimen}(c) reports the running time of different approaches. Across different values of $p_2$, the \emph{KSIC} estimator outperforms the competing methods under both accuracy metrics. The discrepancies from \emph{KSIC} to \emph{Glasso} and \emph{Robust} are the most pronounced, which is expected since 1): \emph{Glasso} is primarily designed to learn graphical dependence structures in an unsupervised paradigm rather than to estimate marginal covariance matrices directly; 2) \emph{Robust} assumes a bandable structure on the covariance matrix that ignores the geometry information in each mode. Among the other covariance estimation approaches, the \emph{Naive} estimator yields the most biased estimate of $\bfSigma_1$ because it ignores the dependence structure along the other mode and therefore fails to exploit information from the ambient modes. This bias does not vanish as $p_2$ increases. The \emph{Naive-Vecchia} approach mitigates this bias through sparse regularization, but it still does not achieve accurate estimation. \emph{MLE} and \emph{MMCD} are the closest competitors to \emph{KSIC}; both show some degree of convergence in Frobenius loss as $p_2$ increases, which is consistent with Theorem~\ref{thm-rate}. %
Moreover, as $p_2$ grows, the full covariance estimates from \emph{MMCD} and \emph{MLE} deteriorate substantially when the ratio between dimensions approaches the MLE existence threshold $n\geq p_2/p_1 + p_1/p_2$ discussed in Section~\ref{subsec:non-parametric}. This contrast highlights the advantage of \emph{KSIC} in imbalanced and high-dimensional settings. In terms of running time, when $p_2$ is small, the \emph{MLE} is the most efficient due to its highly-optimized implementation in \texttt{R}. However, all approaches except \emph{KSIC} and the \emph{Naive-Vecchia} estimator grow at least on the order of $\mathcal{O}(p_2^3)$, making them computationally infeasible in high-dimensional settings. 

\begin{figure}[htbp]
\centering
	\begin{subfigure}{.48\textwidth}
	\centering
  	\includegraphics[width =.98\linewidth]{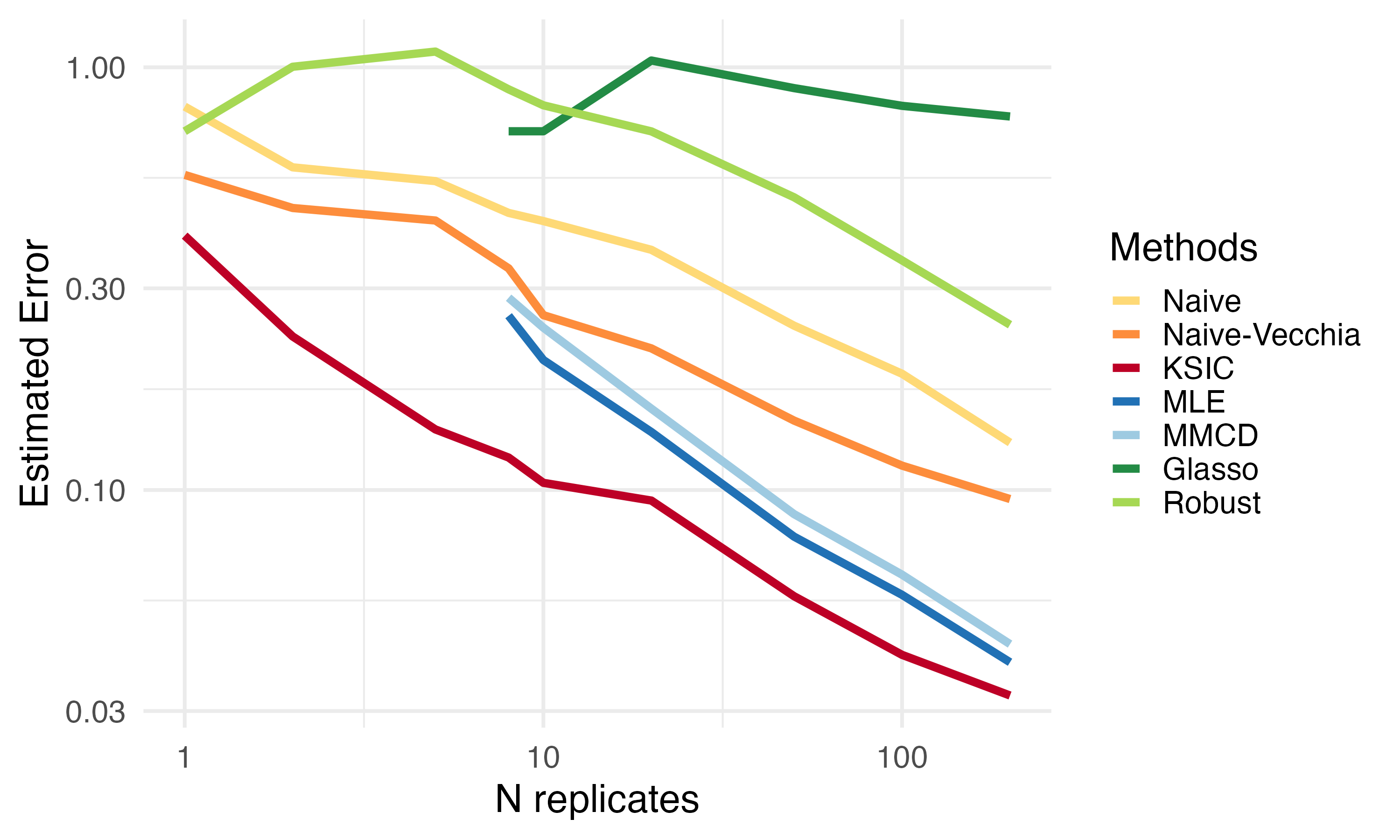}
	\caption{Frobenius loss}
	\end{subfigure}%
\hfill
	\begin{subfigure}{.48\textwidth}
	\centering
 	\includegraphics[width =.98\linewidth]{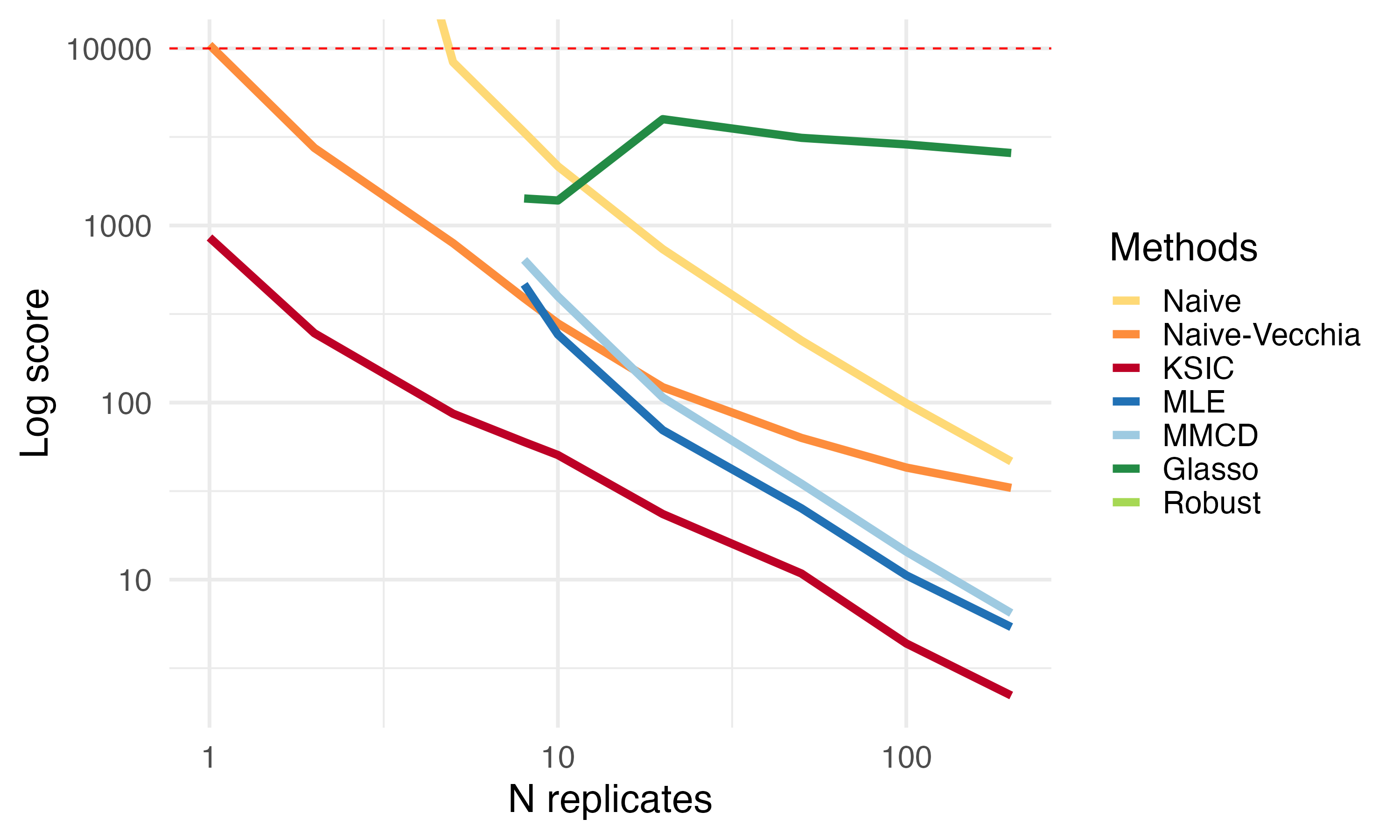}
	\caption{Log-score}
	\end{subfigure}%
\caption{Estimation error and log-score for different approaches v.s.\ number of replicates.}
\label{fig:np-rep}
\end{figure}
Figure~\ref{fig:np-rep} follows the same layout as Figure~\ref{fig:np-dimen}, except that the x-axis now represents the number of training replicates $n$, with $p_2 = 60$ fixed across all experiments. Across $n$, \emph{KSIC} is at least comparable to the competing methods in terms of both Frobenius loss and Log-score. Its advantage is most pronounced for small $n$, especially when $n$ is below the sample-size threshold of \emph{MLE}, where \emph{MLE}, \emph{MMCD}, and \emph{Glasso} become infeasible. Compared with Naive-type approaches, \emph{KSIC} provides more efficient estimation as $n$ increases. This behavior is consistent with the convergence rate in Theorem~\ref{thm-rate}, where the estimation error depends on the discrepancy term that quantifies the goodness of fit along the ambient modes. When $p_{-k}$ is fixed, the methods are expected to have comparable asymptotic dependence on $n$, differing mainly in constants and approximation bias; for example, a similar $\mathcal{O}(n^{-1/2})$ rate was established for \emph{Glasso} in \citet{lyu2019tensor} and for \emph{Robust} in \citet{zhang2023covariance}.

\subsection{Parametric estimation}\label{sec:sim-par}

For parametric estimation, we evaluate the estimator that maximizes the KSIC-projected likelihood described in Section~\ref{subsec:parametric_estimation}. We benchmark the proposed estimator against two primary baselines: the exact Maximum Likelihood Estimator (MLE) and the global Sparse Inverse Cholesky (SIC) projected likelihood estimator. We deliberately focus on these two reference points to establish rigorous statistical and computational bounds. 

The exact MLE serves as the theoretical oracle benchmark, allowing us to quantify any statistical efficiency loss or bias induced by the nested KSIC approximation. Conversely, the global SIC projection represents the state-of-the-art unstructured precision approximation; comparing against it isolates the specific statistical and computational gains achieved strictly by exploiting the multi-way Kronecker structure. 

To ensure a fair comparison, the conditioning-set size for the global SIC projection is calibrated to match the marginal conditioning-set sizes $m$ used by \emph{KSIC-$m$}. Specifically, we evaluate both \emph{SIC-$m$} and \emph{SIC-$m^2$}, mirroring the setting in Section~\ref{sec:sim-KLD}. For global SIC, conditioning sets are chosen as nearest-neighbor sets in the product space of the marginal domains with random ordering.

Similar to the nonparametric experiments, we simulate multi-way data $\{\cX_{j}\}_{j = 1}^{n}$ with $K = 2$ modes under the covariance model in \eqref{eq-sim-cov} from Section~\ref{sec:sim-KLD}. We fix $\sigma^2 = 1$ and $\delta = 0.1$ as known parameters and vary the unknown parameter $\phi \in \{1, 2, 3, 4, 5\}$. We evaluate estimation performance for $\phi$ across 10 independent simulated datasets using two metrics: estimation error relative to the MLE $\bftheta_{\text{MLE}}$ and the Kullback--Leibler divergence $\mathrm{KL}(\bfSigma_{\bftheta_{\text{true}}} \Vert{} \widehat{\bfSigma}_{\text{KSIC}})$.

\begin{figure}[htbp]
\centering
	\begin{subfigure}{.48\textwidth}
	\centering
  	\includegraphics[width =.98\linewidth]{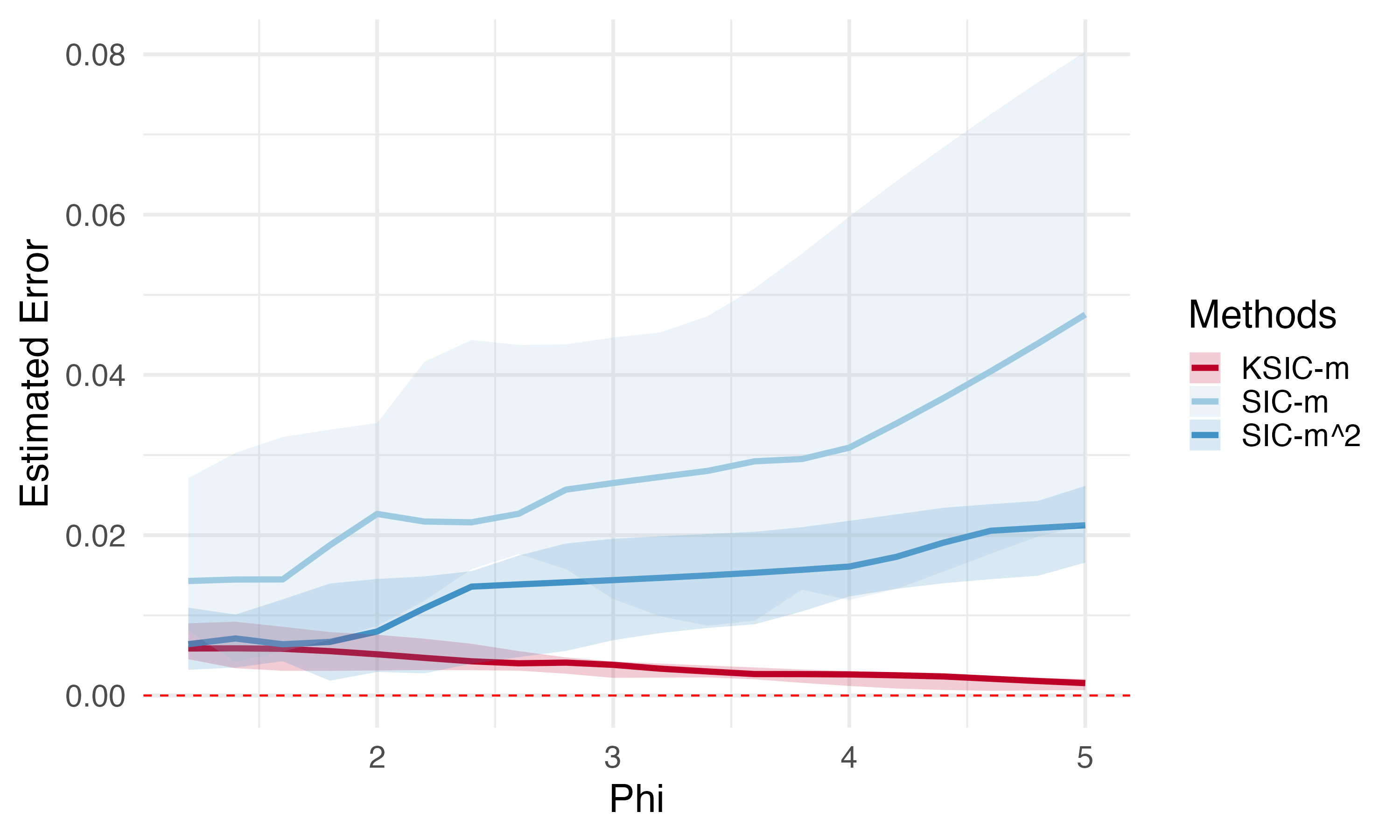}
	\caption{Estimation error}
	\end{subfigure}%
\hfill
	\begin{subfigure}{.48\textwidth}
	\centering
  	\includegraphics[width =.98\linewidth]{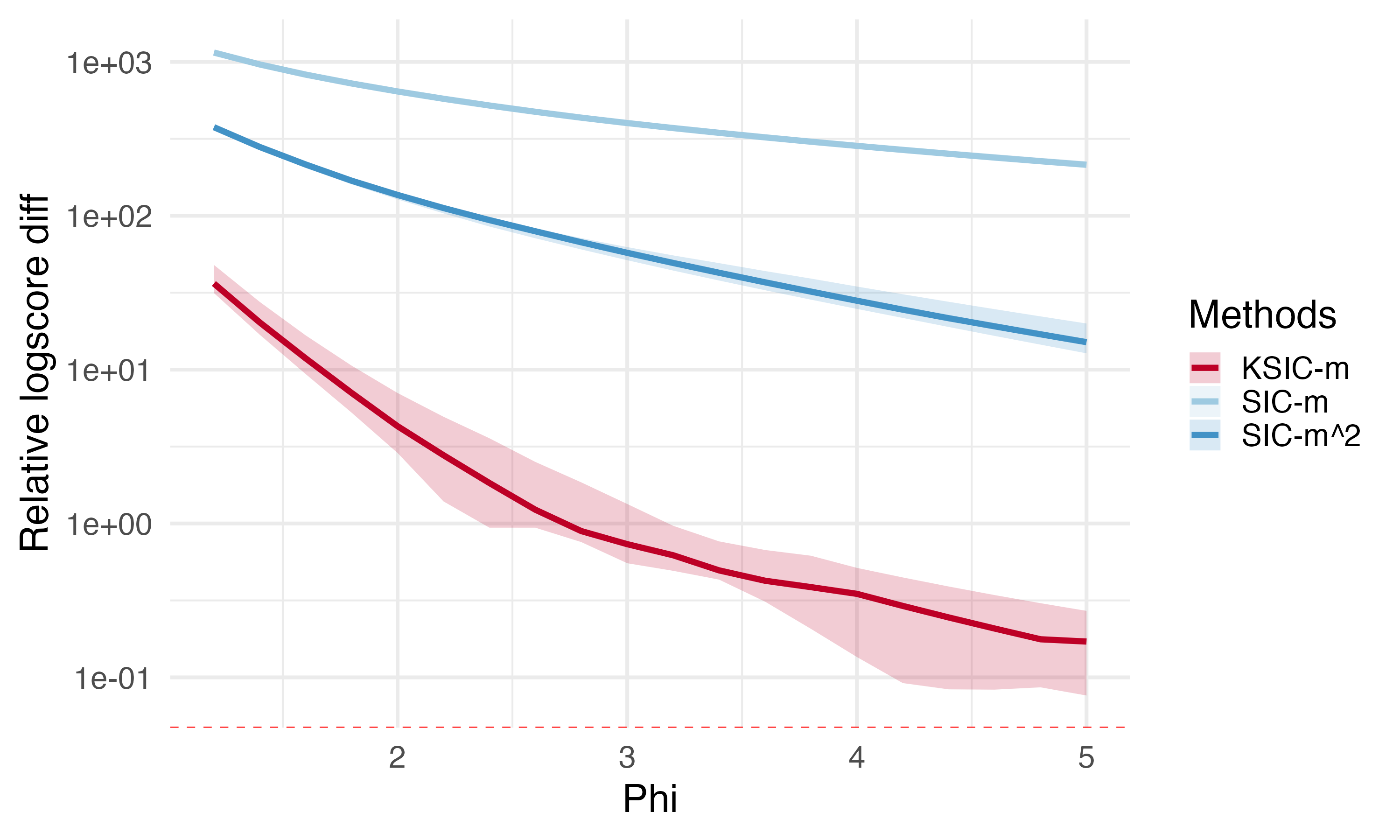}
	\caption{Relative KL divergence}
	\end{subfigure}%
\caption{Estimation error and KL divergence across different estimation approaches for parametric estimation.}
\label{fig:parametric}
\end{figure}

Figure \ref{fig:parametric} presents the comparative performance across varying values of $\phi$. Crucially, the true underlying covariance model in this simulation is non-separable. This demonstrates that the KSIC projection is highly effective at capturing complex, non-separable dependencies, as the \emph{KSIC-$m$} estimator consistently tracks closer to the exact MLE than the unstructured \emph{SIC} estimators in both parameter estimation error and KL divergence. This demonstrates that the KSIC projection $\Pi_{\text{KS}}(\bfSigma, \cS_{\text{KS}})$ provides a tighter approximation to $\bfSigma$ than the standard global SIC projection $\Pi(\bfSigma, \cS)$ when constrained to comparable conditioning sizes. Crucially, these findings indicate that conditioning on separable sparsity patterns in \emph{KSIC} preserves significantly more joint covariance information than unstructured sparsity patterns in the full product space. Combined with its dramatic reduction in computational complexity, these results establish \emph{KSIC} as a highly effective approach for multi-way parametric covariance inference.

\section{Real data applications}\label{sec:real}

\subsection{Spatial-temporal analysis of climate data}\label{sec:real-temp}

We obtain 2-meter air temperature data from the \textit{WeatherBench2 ERA5} dataset \citep{era5_cds, hersbach2020era5, rasp2020weatherbench, rasp2024weatherbench2}, which provides hourly global atmospheric fields from 1959--2023. Focusing on a consistent boreal summer window, we extracted hourly observations for June across all 64 available years. The spatial domain spans 110°W--70°W and 40°N--30°S, covering the eastern Pacific Ocean, Central America, and parts of South America. We analyzed the data on a coarsened $70 \times 40$ spatial grid, yielding $p_1 = 2{,}800$ vectorized spatial locations. The temporal observations are structured by hour ($p_2=24$) and day ($p_3 = 30$) within June, with the $n=64$ years treated as independent training replicates. We centered the data by subtracting the empirical mean for each location-hour-day combination, producing temperature anomaly fields that remove systematic geospatial and diurnal effects prior to statistical modeling.

Figure~\ref{fig:real-temp}(a) shows a snapshot of the evolving temperature anomalies for June 1959 at 12:00 pm, and Figure~\ref{fig:real-temp}(b) shows hourly anomalies on June 1, 1959. The temperature anomaly fields exhibit clear spatial coherence, with large-scale structures persisting across neighboring regions and broadly aligning with continental geography. Stronger positive and negative anomalies frequently occur in the northern part of the domain, around $20^\circ$N--$40^\circ$N, whereas the southern portion appears smoother and exhibits weaker variability. Temporally, the hourly anomalies evolve gradually rather than abruptly; spatial patterns propagate and shift across successive panels, indicating strong short-term temporal dependence. Across days, this dependence is weaker, with abrupt changes visible, for example, between June 1 and 2. The anomaly fields also show intermittent intensification, in which localized regions of large anomalies emerge and dissipate, suggesting episodic variability superimposed on a relatively stable large-scale background.
\begin{figure}[!h]
\centering
	\begin{subfigure}{.495\textwidth}
	\centering
  	\includegraphics[width =.995\linewidth]{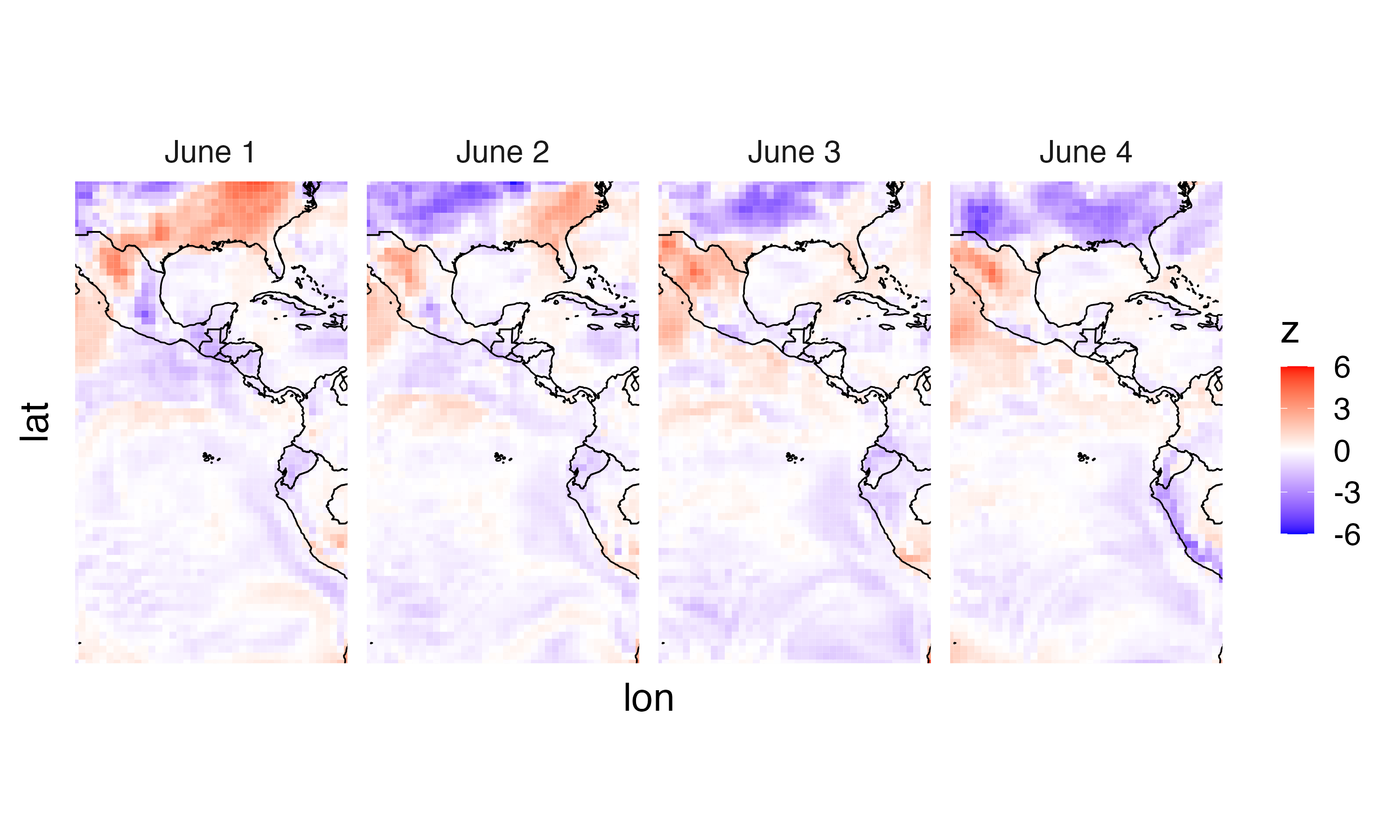}
	\caption{Temperature anomalies on June 1--6 at noon}
	\end{subfigure}%
\hfill
	\begin{subfigure}{.495\textwidth}
	\centering
 	\includegraphics[width =.995\linewidth]{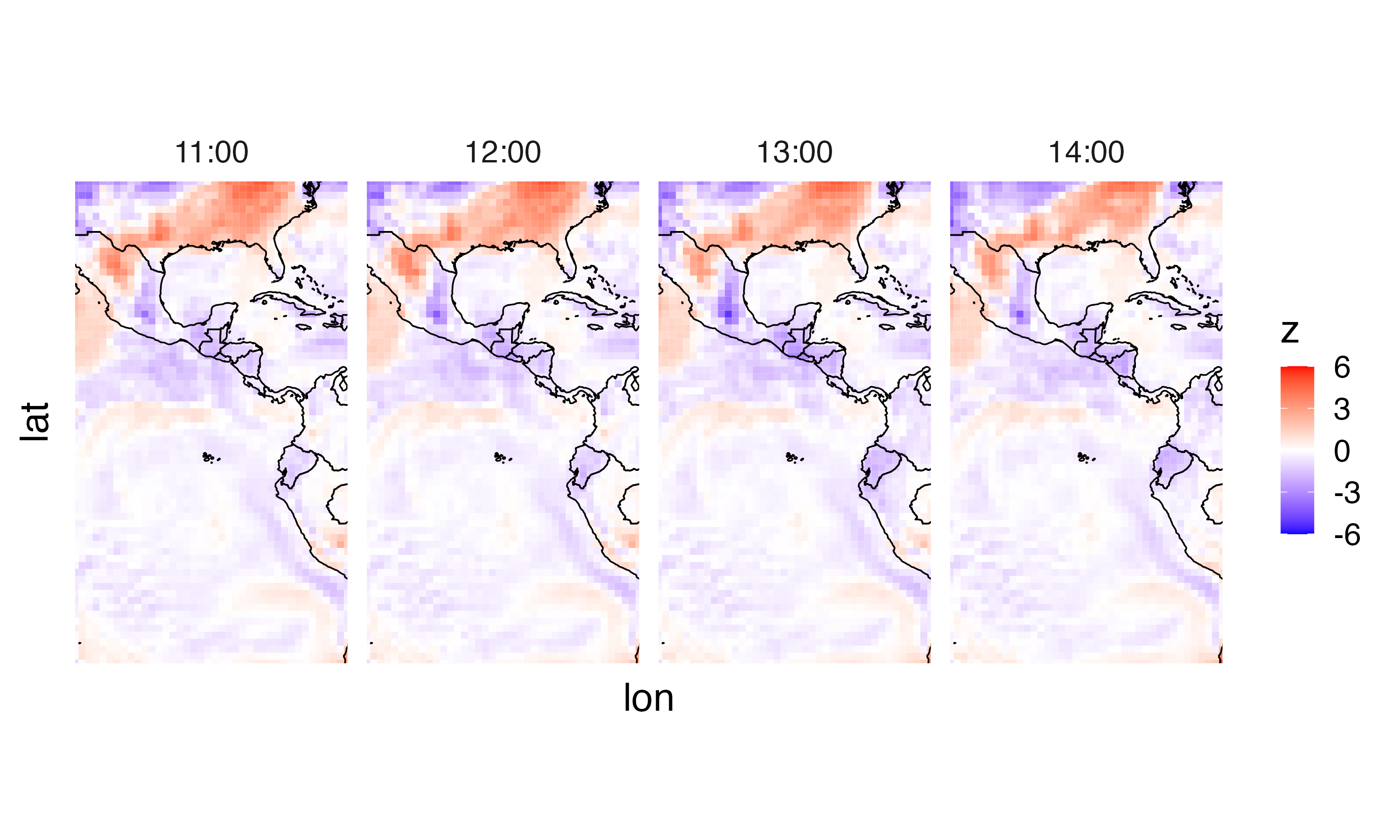}
	\caption{Hourly temperature anomalies on June 1}	
	\end{subfigure}%
\vspace{0.5em}
\centering
	\begin{subfigure}{.48\textwidth}
	\centering
  	\includegraphics[width =.98\linewidth]{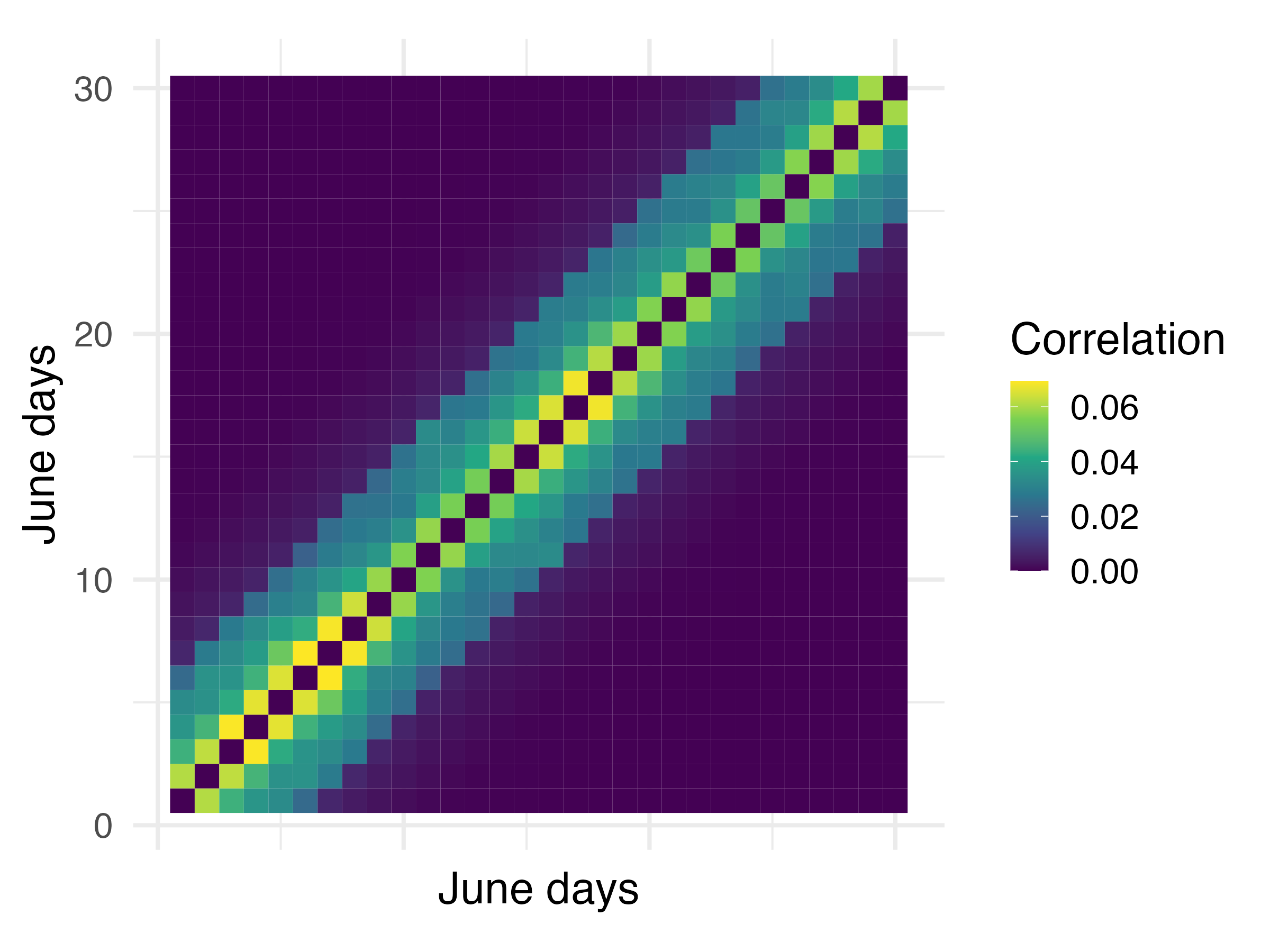}
	\caption{Covariance across days}
	\end{subfigure}%
\hfill
	\begin{subfigure}{.48\textwidth}
	\centering
 	\includegraphics[width =.98\linewidth]{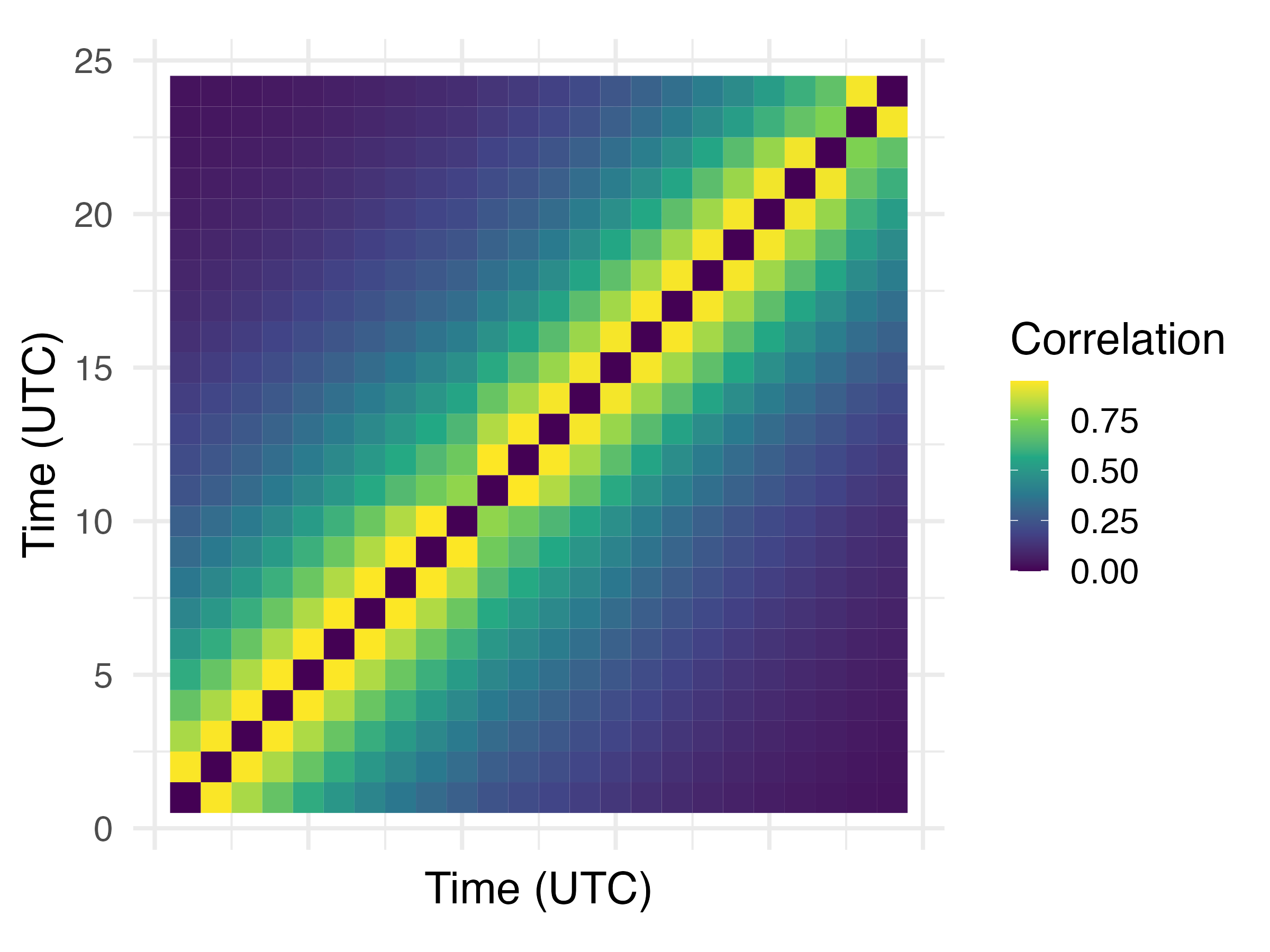}
	\caption{Covariance across hours of the day}	
	\end{subfigure}%
\vspace{0.5em}
\centering
	\begin{subfigure}{.48\textwidth}
	\centering
  	\includegraphics[width =.98\linewidth]{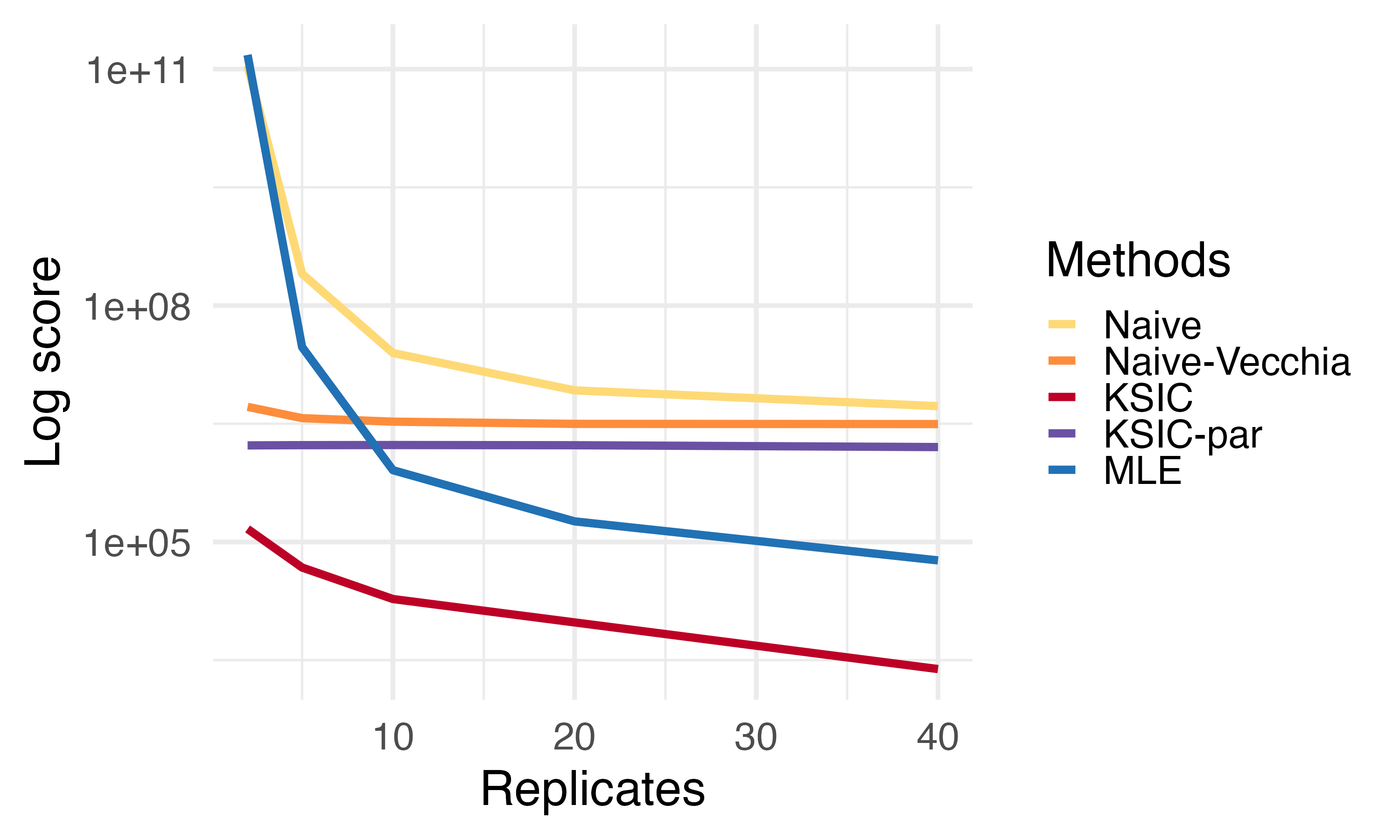}
	\caption{Log-score on test set}
	\end{subfigure}%
    \begin{subfigure}{.48\textwidth}
	\centering
  	\includegraphics[width =.98\linewidth]{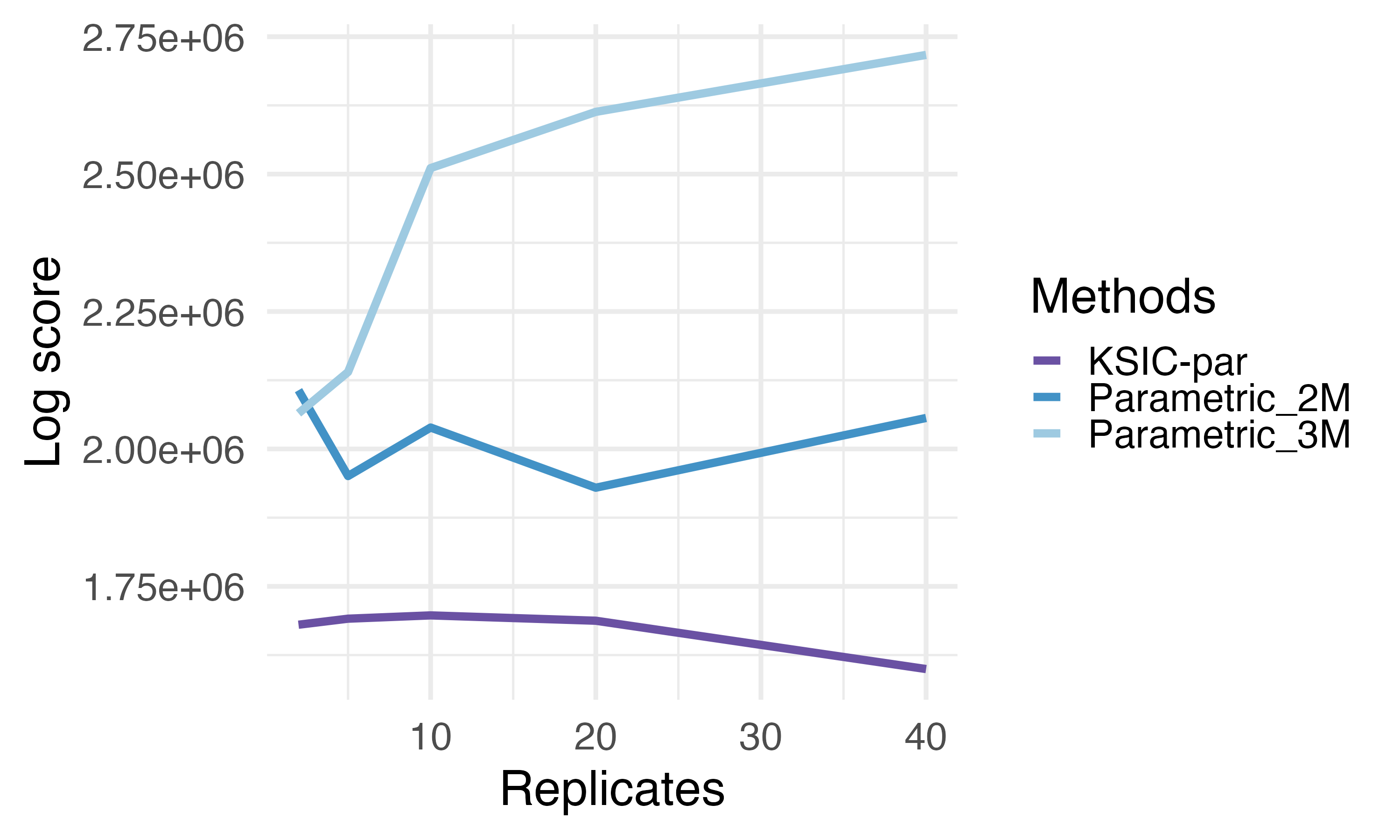}
	\caption{Log-score only for parametric methods}
	\end{subfigure}%
\caption{Temperature anomalies analysis}
\label{fig:real-temp}
\end{figure}

When applying KSIC, we naturally treat space, days, and hours as three separate modes. Although the two temporal modes could be merged because the anomaly process is continuous across adjacent days, the order-3 representation separately captures temporal dependence at daily and hourly resolutions, providing multi-scale scientific interpretability. Because the total dimension of each sample is $2{,}800 \times 24 \times 30 > 2$ million, most direct estimators of $\bfSigma$ are infeasible on a standard machine. In the following analysis, we conduct the nonparametric covariance estimation described in Section~\ref{subsec:non-parametric}, and validate the proposed approach both quantitatively and qualitatively. We randomly select training samples from the first 40 years and reserve the remaining 24 years as testing samples. In Figures~\ref{fig:real-temp}(c) and~\ref{fig:real-temp}(d), we plot the within-month and within-day correlation matrices estimated by KSIC using 10 random training samples. The magnitudes of the estimated correlations clearly reflect the temporal dependence observed in the anomaly maps. 
Moreover, the block structure in the within-day correlation matrix corresponds to summer temperature cycles from early morning to late afternoon, illustrating the validity and interpretability of the KSIC estimates.

To quantitatively evaluate performance, we compare \emph{KSIC} with the scalable nonparametric approaches considered in Section~\ref{sec:sim-np}, including \emph{Naive}, \emph{Naive-Vecchia}, and \emph{MLE}. We exclude \emph{MMCD}, \emph{Glasso}, and \emph{Robust} from this experiment because their computational cost becomes prohibitive even when the mode dimensions are only on the order of thousands. For all included nonparametric methods, nugget stabilization is applied when needed. 
For approaches involving SIC, we apply maximin ordering on spatial locations and natural ordering on time points.%

We also include parametric models for comparison. \emph{KSIC-par} denotes the KSIC parametric estimation detailed in Section \ref{subsec:parametric_estimation}, where we assume the kernel is separable between spatial and temporal modes. In contrast, \emph{Parametric} methods separately fit parametric models for each mode and normalize the marginal variances. These methods treat observations along the remaining modes as replicates, similar in spirit to the naive family approaches. \emph{Parametric\_2M} treats space as one mode and time in hours as the other mode, while \emph{Parametric\_3M} treats space, hours within a day, and days within a month as three separate modes. In all three parametric approaches, we assume an isotropic Mat\'ern kernel on any possible mode. The estimated Mat\'ern parameters include the variance, range, smoothness, and nugget. The implementation details for these approaches can be found in Appendix \ref{app:sim-detail}. 

In Figure~\ref{fig:real-temp}(e), we show the relative log-score on the testing set, shifted by a constant for visualization purposes, against the number of training replicates $n$. The nonparametric \emph{KSIC} consistently outperforms the competing methods across different values of $n$. Both \emph{MLE} and the \emph{Naive} approach become unstable when $n \leq 5$, which is close to the algorithmic threshold of MLE (i.e., $\max_k\{p_k/p_{-k}\} \approx 5$ for $K = 3$) \citep{Manceur2013TensorNormal}. In contrast, KSIC remains stable even when $n$ approaches $1$ since the algorithmic threshold is $\max_k\{(m_k+1)/p_{-k}\} < 1$. In this low-sample regime, \emph{KSIC}, \emph{Naive-Vecchia}, and the parametric approaches remain feasible because they introduce regularization through sparsity constraints or parametric structure. As $n$ increases, both \emph{KSIC} and \emph{MLE} improve in parallel, whereas the \emph{Naive} and parametric approaches plateau. These patterns reflect the bias induced by ignoring dependence across the remaining modes or by imposing restrictive parametric assumptions.

In Figure~\ref{fig:real-temp}(f), we zoom into the log-score curve of \emph{KSIC-par} and compare it with \emph{Parametric-2D} and \emph{Parametric-3D}. Within the parametric paradigm, \emph{KSIC} still outperforms the other approaches across different sample sizes $n$. This superior performance stems from \emph{KSIC}'s reliance on a joint parametric model that does not fully separate the modes. For example, the variance term $\sigma^2$ is estimated using information from all modes, which cannot be correctly characterized by the separate-fitting paradigm.

\subsection{fMRI data}\label{sec:real-mri}

fMRI data provide another important source of multi-way data, where high dimensionality and data scarcity are common challenges. Here, we apply KSIC to a dataset from the Autism Brain Imaging Data Exchange (ABIDE) release \citep{klein2012101, di2014autism} for nonparametric covariance estimation. 

After appropriate preprocessing, each fMRI multi-way sample has a dimension of $11{,}545$ voxels by on the order of 100 time points. However, the resulting dimensions are highly imbalanced and challenging for methods such as the Naive estimator and MLE. To include a broader range of competing methods, we summarize the data into $392$ regions of interest (ROIs) defined by the brain parcellation of \cite{craddock2012whole}. ROI-level covariance estimation reflects large-scale functional connections among brain regions and is widely used in statistical applications \citep{zhao2025estimating}. Appendix~\ref{app:real-mri} provides details regarding the data source, preprocessing steps, and additional results for voxel-level analysis.

To demonstrate the general effectiveness of KSIC for marginal covariance estimation, we apply Algorithm~\ref{alg-np} separately to data from each site, treating ROIs and time points as two modes and individuals as independent samples. Within each site, we reserve $1/3$ of the samples for testing and randomly draw training samples from the remaining $2/3$ at various sizes. For this split, we retain only sites with more than 5 individuals and aligned time courses across all individuals. After filtering, 11 sites remain, each represented by a tensor of size $p_{roi}\times p_t\times n$. The ROI dimension is fixed at $p_{roi} = 392$ due to standardized registration; $p_t$ ranges from $116$ (UCLA) to $246$ (Leuven); and $n$ ranges from $12$ (SDSU) to $73$ (NYU). Within each site, the data are centered by subtracting the mean BOLD signal at each ROI and time point to remove global fluctuations.

We randomly order the observations and determine the conditioning sets for KSIC along both modes in a data-driven manner. Specifically, we first apply the adaptive choice for $m$, and then include the top-$m$ most correlated indices in the conditioning set of each index. This approach is related to the correlation-based Vecchia approximation (CVecchia) proposed by \cite{Kang2021}, which is useful when standard geometric information, such as Euclidean distance, is unavailable or insufficient. In fMRI, dependence among ROIs is not determined solely by physical distance, making correlation a natural metric for constructing SIC conditioning sets. In addition, the conditioning-set sizes $m_k$ are selected adaptively using the procedure described in Section~\ref{sec:sim-np}.

\begin{figure}[htbp]
\centering
	\begin{subfigure}{.99\textwidth}
    \centering
 	\includegraphics[width =.98\linewidth]{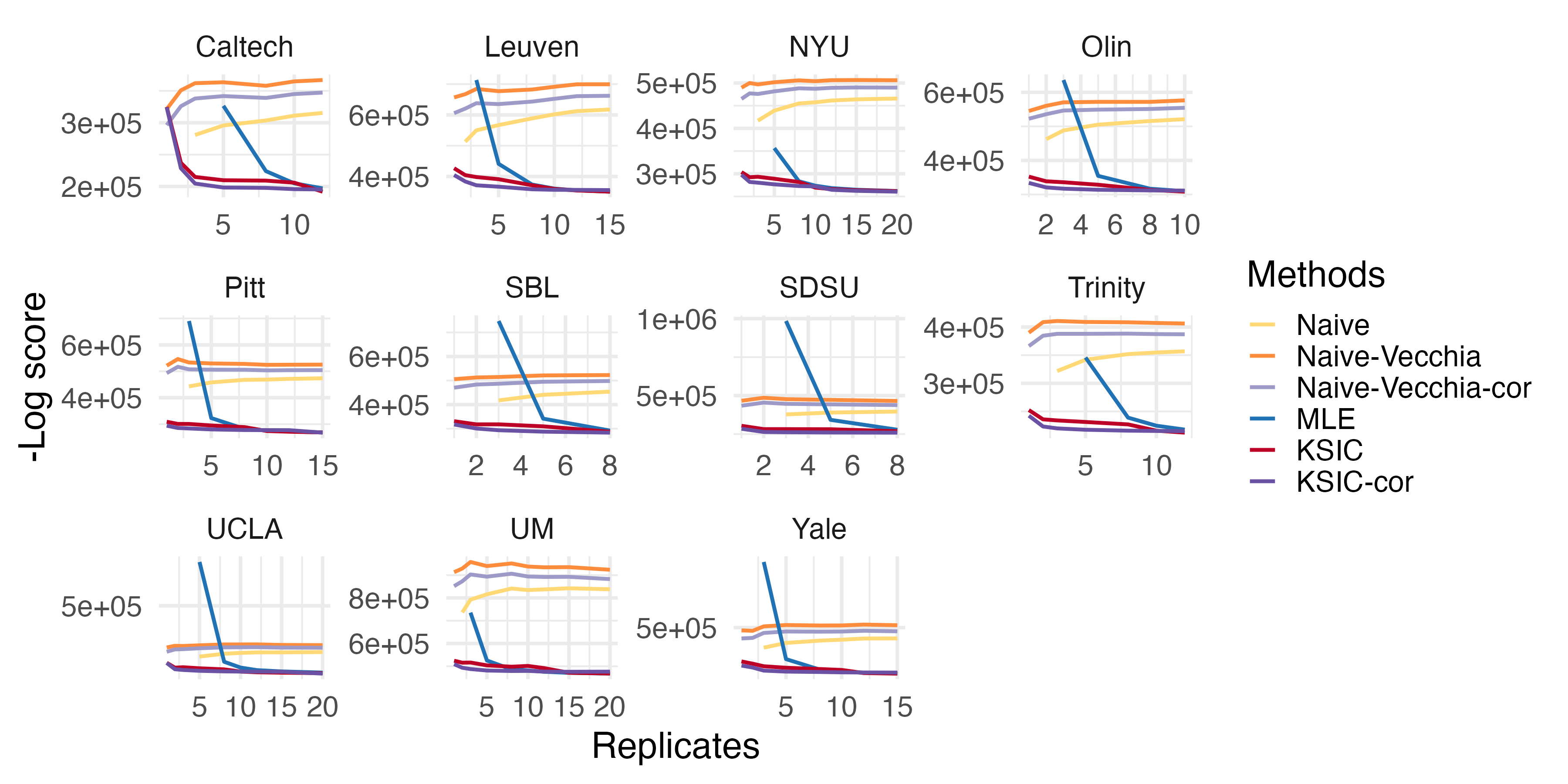}
	\caption{Log-score of different approaches v.s. number of replicates.}
	\end{subfigure}%
\vspace{0.5em}
\centering
	\begin{subfigure}{.48\textwidth}
	\centering
  	\includegraphics[width =.98\linewidth]{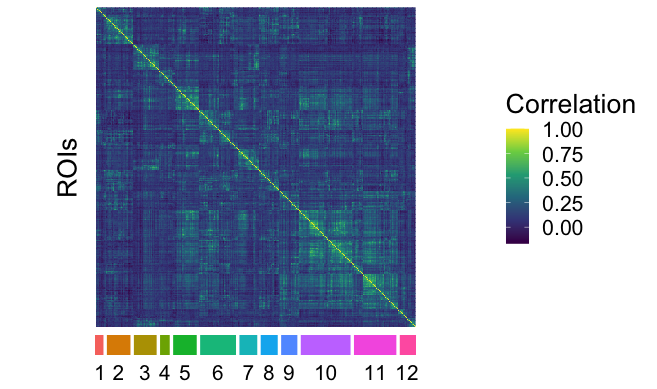}
	\caption{Estimated correlation matrix}
	\end{subfigure}%
\hfill
	\begin{subfigure}{.48\textwidth}
	\centering
 	\includegraphics[width =.98\linewidth]{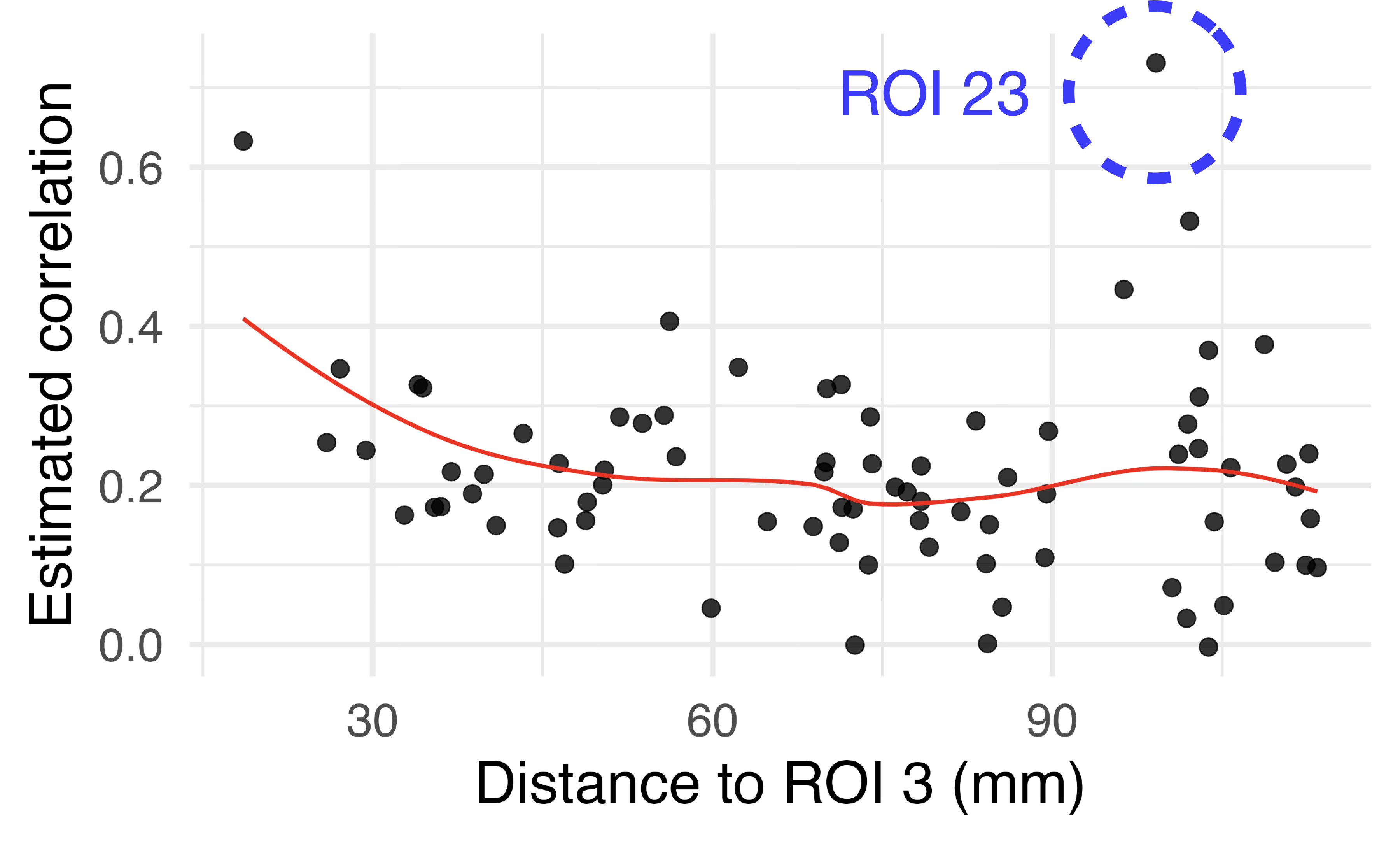}
	\caption{Estimated correlation with ROI $3$.}	
	\end{subfigure}%
\caption{ROI-level FMRI analysis}
\label{fig:real-mri}
\end{figure}

In Figure~\ref{fig:real-mri}(a), we present the testing log-score across 11 sites as a function of the training sample size, with results summarized by the median over five random trials. We compare six approaches for marginal covariance estimation: \emph{Naive} and its variants \emph{Naive-Vecchia} and \emph{Naive-Vecchia-cor}; \emph{MLE}; and \emph{KSIC} alongside its variant \emph{KSIC-cor}. The methods \emph{Naive}, \emph{Naive-Vecchia}, \emph{MLE}, and \emph{KSIC} have been introduced in previous sections. To illustrate the benefit of CVecchia, we use randomly selected conditioning sets for \emph{Naive-Vecchia} and \emph{KSIC}, and additionally consider \emph{Naive-Vecchia-cor} and \emph{KSIC-cor}, where the conditioning sets are determined using correlations estimated from the training data. To mitigate rank-deficiency issues, we apply nugget regularization of an appropriate magnitude when needed. 

Across all sites, \emph{KSIC-cor} consistently achieves the best performance and substantially outperforms methods in the \emph{Naive} family. Moreover, correlation-based conditioning sets improve performance relative to randomly selected conditioning sets. The difference appears modest in the figure mainly because the y-axis scale must accommodate several poorly performing methods. Nevertheless, the improvement is practically meaningful: \emph{KSIC-cor} requires substantially fewer samples than \emph{KSIC} to achieve a given log-score. This advantage is more pronounced in the voxel-level analysis presented in Appendix~\ref{app:real-mri}. 

Meanwhile, \emph{MLE} begins to deteriorate as $n$ approaches the theoretical sample-size threshold $p_1/p_2 + p_2/p_1$, whereas the KSIC-based approaches remain stable and accurate. This empirical behavior again reflects KSIC's robustness under data scarcity discussed in Section~\ref{subsec:non-parametric}. From a practical perspective, the strong small-sample performance of \emph{KSIC} is particularly valuable. By using the multi-way structure to share information across spatial, temporal, and subject-level modes, \emph{KSIC} can yield well-conditioned and sparse connectivity representations without requiring additional clinical subjects and substantially higher data-collection costs.

To illustrate the interpretability of KSIC, we apply \emph{KSIC-cor} to 10 randomly selected training samples from the Caltech site and derive ROI-level correlations from the estimated covariance matrix. Figure~\ref{fig:real-mri}(b) visualizes the resulting correlation matrix after rearranging ROIs according to cluster labels. Several clusters corresponding to known resting-state networks are clearly identified. For example, ROIs in Cluster 2 closely align with the classical core default mode network (DMN), with prominent medial frontal, paracingulate, posterior cingulate, and precuneus regions. Cluster 5 is dominated by visual-network regions, including lingual, cuneal, calcarine, intracalcarine, and occipital regions, consistent with the visual system identified in intrinsic connectivity parcellations \citep{yeo2011organization}. Cluster 11 mainly comprises peri-Sylvian sensorimotor, opercular, insular, and superior temporal regions, suggesting a sensorimotor/opercular-auditory cluster that overlaps with the salience network.

Figure~\ref{fig:real-mri}(c) reports the correlations between ROI 3, labeled as the left postcentral gyrus, and all other ROIs. The overall pattern exhibits a distance-related trend, consistent with the spatial smoothness of neural activity. At the same time, several spatially distant ROIs show strong correlations with ROI 3. For example, ROI 23, labeled as the right precentral gyrus, is strongly correlated with ROI 3, aligning with classical findings on bilateral sensorimotor resting-state connectivity \citep{biswal1995functional}. This example illustrates that KSIC can capture, in a data-driven manner, both broad geometric structure and nontrivial functional connections with biological relevance. The full correlation distribution and additional examples are provided in Appendix~\ref{app:real-mri}.

\section{Conclusions}

In this work, we propose KSIC, a statistical framework for fast, stable, and flexible covariance estimation in multi-way data analysis. The core of our approach is the KSIC projection, defined as a moment-matching projection onto a sparse Kronecker-structured inverse Cholesky manifold. We further develop a block coordinate descent (BCD) algorithm that efficiently computes this projection in extremely high-dimensional settings.

Although KSIC uses a Kronecker-structured representation internally, the KSIC projection does not require the target covariance matrix to be separable. Consequently, it is applicable to general multi-way covariance families. We theoretically analyze the KSIC projection and establish conditions under which it is well defined, both at the level of individual algorithmic updates and at the level of global existence.

KSIC naturally facilitates covariance estimation under both nonparametric and parametric regimes. In the nonparametric setting, we establish finite-sample error bounds for the KSIC estimator. These results generalize existing theory by explicitly characterizing the role of implicit data augmentation and the bias-variance trade-off induced by sparsity regularization. In particular, the convergence rate improves when the working covariance estimates along the remaining modes are close to the corresponding true covariance structures. In the parametric setting, we introduce a projected-likelihood approach that exploits the sparsity and Kronecker structure of the KSIC manifold to make likelihood-based estimation feasible for high-dimensional multi-way data. A full theoretical analysis of the parametric KSIC estimator would require studying a bilevel optimization problem and remains beyond the scope of this paper. Extending the KSIC projection to higher-order optimization approaches, such as Fisher scoring \citep{Guinness2019}, is an interesting direction for future work.

Beyond its theoretical contributions, KSIC provides substantial practical advantages. It inherits the computational benefits of SIC and integrates them into a scalable BCD algorithm. As a result, KSIC remains computationally efficient even in extremely high-dimensional settings. To our knowledge, the nonparametric KSIC estimator is the first nontrivial multi-way covariance estimation procedure that scales linearly with the dimension $p$ while retaining mode-specific covariance structures.

Importantly, the sparsity constraints in KSIC also improve robustness under data scarcity. When the sample size falls below the existence threshold required by classical methods such as the Kronecker MLE, KSIC can still provide stable and accurate covariance estimates through sparsity-induced regularization. This advantage is particularly valuable in applications where the number of statistically independent replicates is the primary bottleneck. Our simulation studies and real-data examples suggest that stronger existence guarantees for KSIC may exist between the algorithmic threshold and the more conservative existence threshold derived in this paper. Bridging this gap remains an important open problem. More broadly, KSIC belongs to the family of regularized covariance estimation methods, alongside Gaussian graphical models \citep{banerjee2008model}, as well as banding and tapering methods \citep{Levina2008}. A general theory of regularized multi-way covariance estimation under data scarcity is a promising direction for future research.

While we primarily focus on settings where the tensor-valued observations are directly observed and modeled through a multi-way covariance structure, the framework could be extended to latent tensor models in which the underlying tensor is observed through a measurement operator and corrupted by noise. Such extensions would broaden the class of covariance structures represented by KSIC and reduce model misspecification. They would also better reflect many practical settings, including spatial modeling with nugget effects, missing-data problems, and remote-sensing applications.

\footnotesize
\appendix
\section*{Acknowledgments}

Support for this research was provided by the Office of the Vice Chancellor for Research at the University of Wisconsin--Madison with funding from the Wisconsin Alumni Research Foundation. MK's research was also partially supported by National Science Foundation (NSF) Grant DMS--1953005/2433548.
We would like to thank Chris Geoga and Garvesh Raskutti for helpful comments and discussions.

\bibliographystyle{apalike}
\bibliography{mendeley,additionalrefs}

\section{Additional Theoretical Results for Nonparametric Estimation}\label{app:np}

\subsection{Nonparametric Algorithm}

For reader's reference, we first present the modified Algorithm \ref{alg-main} under a nonparametric paradigm as Algorithm \ref{alg-np}. The only difference lies in the method of computing the pseudo-covariance $\tilde{\bfSigma}_k$.
\begin{algorithm}[htbp]
\caption{KSIC-BCD-nonparametric}
\KwInput{Multi-way samples $\mathcal{X}_1,\ldots,\mathcal{X}_n$, sparse pattern sets $\{\cS_1, \ldots, \cS_K\}$, maximum iteration $T$, stopping threshold $\epsilon$.}
\begin{algorithmic}[1]
\STATE Initialize the individual covariance matrices $\{\bfSigma^{(0)}_k, \ k = 1, \ldots, K\}$, $t \gets 0$.
\STATE Compute Vecchia's approximation on Cholesky factors $\{\bL^{(0)}_k, \ k = 1, \ldots, K\}$.
\WHILE{$t < T$ \AND $\max_k\|\bL_k^{(t)} - \bL_k^{(t-1)}\|_F < \epsilon$}
\STATE $t \gets t+1$.
\FOR{$k=1,2,\ldots, K$}
     \STATE Compute $\tilde{\bfSigma}_k$ according to Proposition \ref{prop-low_rank}. 
     \STATE Compute the SIC projection $\hat{\bL}_k = \Pi(\tilde\bfSigma_k,\mathcal{S}_k)$.
\ENDFOR
\ENDWHILE
\RETURN $\{\hat{\bL}_k, \ k = 1, \ldots, K\}$
\end{algorithmic}
\label{alg-np}
\end{algorithm}

\subsection{Asymptotic Efficiency}\label{app:np-asym}
This section complements Section~\ref{subsec:np-asym}. In the following theorem, we assume a separable multi-way covariance $\bfSigma = \bigotimes_{k=1}^K\bfSigma_k$ and show that the population covariance is proportional to the true marginal covariance. On a certain mode $k$, we adapt the setting of \citet{Schafer2017}, where the covariance is induced by a boundary-conditioned Mat\'ern-type model. Under an in-fill asymptotic regime, where mode dimensions and conditioning set potentially grow with sample size $n$, we provide a global error rate of KSIC marginal covariance estimator on that mode. 
First, we introduce notation and assumptions adapted from \citet{Schafer2017}.
\begin{assumption}[Homogeneous design]
\label{asmp:sic-1}
The locations \(s_1,\ldots,s_p\in\Omega \subseteq \mathbb{R}^d\) satisfy the homogeneity condition
in \citet[Theorem~3.4]{Schafer2020}. Equivalently, the homogeneity ratio
\[
    \delta_p
    :=
    \frac{
        \min_i \operatorname{dist}
        \bigl(s_i,\{s_j:j\neq i\}\cup\partial\Omega\bigr)
    }{
        \max_{x\in\Omega}
        \operatorname{dist}
        \bigl(x,\{s_i\}_{i=1}^p\cup\partial\Omega\bigr)
    }
\]
is bounded below by a positive constant independent of \(p\).
\end{assumption}
\begin{assumption}[Green's-function covariance]
\label{asmp:sic-2}
The covariance matrix \(\bfSigma\in\mathbb R^{p\times p}\) is generated by
the Green's function \(G\) of a local, symmetric, positive elliptic boundary-value
problem on \(\Omega\), namely
\[
    \bfSigma_{ij}=G(s_i,s_j),
    \qquad i,j=1,\ldots,p.
\]
In particular, this includes the boundary-conditioned Mat\'ern-type model whose
precision operator is
\[
    (\kappa^2-\Delta)^\alpha,
    \qquad
    \alpha=\nu+d/2\in\mathbb N,
\]
with homogeneous Dirichlet boundary condition.
\end{assumption}

The following lemma states that when the covariance is induced by the settings above, the error of SIC approximation can be well-controlled by adjusting the conditioning sets.
\begin{lemma}[Theorem 3.4 in \citet{Schafer2020}]
\label{lem:sic-shared}
Let $\Omega\subset\mathbb R^d$ be a bounded Lipschitz domain. Let $\cS_\rho$ be the reverse-maximin sparsity pattern with radius parameter
    $\rho$ used in \citet[Theorem~3.4]{Schafer2020}, and define
    \[
        \bL_\rho=\Pi(\bfSigma,\cS_\rho),
    \]
    where $\Pi(\bfSigma,\cS_\rho)$ denotes the SIC
    projection onto the sparse matrix subspace $\cS_\rho$.
Then $\bL_\rho\bL_\rho^\top$ is the SIC approximation to the precision matrix
$\bfSigma^{-1}$, and $(\bL_\rho\bL_\rho^\top)^{-1}$ is the corresponding SIC-induced
covariance matrix.
Under  Assumptions \ref{asmp:sic-1} and \ref{asmp:sic-2}, there exist constants $C, c > 0$ such that, if
\[
    \rho\ge \frac{1}{c}\log(p/\epsilon),
\]
then
\[
    \mathrm{KL}
    \left(
        N(0,\bfSigma)\,\|\,
        N\big(0,(\bL_\rho\bL_\rho^\top)^{-1}\big)
    \right)
    +
    \left\|
        \bfSigma-(\bL_\rho\bL_\rho^\top)^{-1}
    \right\|_F
    \le
    \epsilon.
\]
In particular, choosing $\epsilon$ proportional to $p \exp(- c\rho)$ gives
\[
    \left\|
        \bfSigma-(\bL_\rho\bL_\rho^\top)^{-1}
    \right\|_F
    \le
    C p\exp(-c\rho).
\]
The constants $C$ and $c$ depend only on $d$, $\Omega$, the ellipticity and regularity
constants of the underlying differential operator, the operator norms of $G$ appearing in
\citet[Theorem~3.4]{Schafer2020}, and the lower bound on the homogeneity ratio
$\delta_p$, but not on $p$, $\epsilon$, or $\rho$.
\end{lemma}

Beyond controlling the SIC approximation error, obtaining an explicit convergence rate in terms of $n$, $p_k$, and $m_k$ requires specifying their joint asymptotic regime. The following assumption requires the conditioning-set size $m_k$ to grow with the mode dimension $p_k$, so that the SIC approximation error remains controlled. It also imposes a lower bound on the sample size $n$ to ensure that the stochastic estimation error vanishes.
\begin{assumption}[Growth rate]
\label{asmp:sic-3}
Let $m_k$ be the maximum conditioning set size of the sparsity pattern $S_{\rho}$ defined in Lemma \ref{lem:sic-shared}. Define
\[
\Lambda_K := \lambda_{\max}(\bL_k^\top\bfSigma_k\bL_k), \quad \Lambda_{-k} = \prod_{l\neq k}\Lambda_l.
\]
As $n \to \infty$, the dimension and conditioning set sizes $p_k, m_k \to\infty$ . 
For mode $k$, the sequences satisfy the scalings:
\[
m_k = \omega(\log^{d} p_k),
\]
and 
\[
n = \omega\left(\frac{\Lambda_{-k}}{\Lambda_k}\sqrt{\frac{\log p_k}{pp_k}}\exp(cm_k^{1/d})\right)
\]
where $d$ is the dimension of space introduced in Assumption \ref{asmp:sic-1} and $c$ is the constant introduced in Lemma \ref{lem:sic-shared}.
\end{assumption}

\begin{theorem}\label{col-main}
Under the setting of Theorem~\ref{thm-rate}, suppose that the true covariance matrix is separable, \(\bfSigma=\bigotimes_{l=1}^K\bfSigma_l\). For  mode \(k\), let $\Omega\subset\mathbb R^d$ be a bounded Lipschitz domain, and assume that \(\bfSigma_k\) satisfies Assumptions~\ref{asmp:sic-1}--\ref{asmp:sic-3} under the in-fill asymptotic regime for \(p_k\). Further assume on the ambient modes that for all $l\neq k$,  
\[
\tr(\bL_l^\top\bL_l)\asymp p_l \quad \text{ and } \quad\lambda_{\min}(\bfSigma_l) \text{ has a uniform lower bound.}
\]
With the notation \(p_{-k}:=\prod_{l\neq k}p_l\), \(\bL_{-k}:=\bigotimes_{l\neq k}\bL_l\), $\bfOmega_k = \bfSigma_k^{-1}$, and \(\bfSigma_{-k}:=\bigotimes_{l\neq k}\bfSigma_l\), the induced marginal target for mode \(k\) is
\[
    \widetilde{\bfSigma}_k
    =
    \frac{\operatorname{tr}(\bL_{-k}^\top\bfSigma_{-k}\bL_{-k})}{p_{-k}}\bfSigma_k
    =
    \frac{\prod_{l\neq k}\operatorname{tr}(\bL_l^\top\bfSigma_l\bL_l)}{p_{-k}}\bfSigma_k .
\]
Moreover, the discrepancy term in Theorem~\ref{thm-rate} reduces to \(\max_{i^{(k)},j^{(k)}}\|\bE_{i^{(k)},j^{(k)}}\|_2=\|\bL_{-k}^\top\bfSigma_{-k}\bL_{-k}\|_2\). Consequently, let $m^*_k := \min\{m_k, p^{1/2}_k\}$, the normalized KSIC estimator satisfies
\[
\left\|
\frac{\hat{\bL}_{k}\hat{\bL}_{k}^\top}
{\|\hat{\bL}_{k}\hat{\bL}_{k}^\top\|_F}
-
\frac{\bfOmega_k}{\|\bfOmega_k\|_F}
\right\|_F
=
\mathcal{O}_p\left(
C_1
\|\bL^\top_{-k}\bfSigma_{-k}\bL_{-k}\|_2
\sqrt{\frac{m^{*2}_k\log p_k}{np_{-k}}}
+
(p_k^{1/2+2\nu/d} + m^*_kp_k)
\exp\{-cm_k^{1/d}\}
\right).
\]
\end{theorem}
The first implication of Theorem \ref{col-main} is that, under a separable covariance model, the discrepancy term reduces to $\prod_{l\neq k}\operatorname{tr}(\bL_l^\top\bfSigma_l\bL_l)$, which can be interpreted as the cumulative approximation error contributed by the mode-specific factors along the ambient modes. In the ideal case of perfect mode-wise recovery, where $\bL_l\bL_l^\top = \bfSigma_l^{-1}$ for all $l\neq k$, we have $\|\bL_{-k}^\top\bfSigma_{-k}\bL_{-k}\|_2 \equiv 1$. In contrast, in a poorly specified case such as $\bL_{-k}=\bI_{p_{-k}}$, the same quantity may grow as $\|\bL_{-k}^\top\bfSigma_{-k}\bL_{-k}\|_2 = \mathcal{O}(p_{-k}^{1/2})$, thereby substantially worsening the convergence rate. This deterioration becomes more pronounced when $p_{-k}$ is large, precisely in settings where the multi-way structure offers greater potential for implicit data augmentation across the remaining modes.

Second, the key role of the conditioning-set size $m_k$ is now explicit. Enlarging the conditioning set reduces the SIC approximation error, corresponding to the second term in the rate, but also increases the stochastic estimation error, corresponding to the first term, as well as the computational cost. In practice, especially in nonparametric estimation, this trade-off can be managed by selecting $m_k$ using a validation set, which is precisely the purpose of the \emph{adaptive choice} described in Section~\ref{subsec:np-scarce}. In principle, when $m_k$ increases, $m_k^*$ can be as large as $p_k^{1/2}$. In this dense-conditioning limit, if the misspecification along the ambient modes is held fixed, the SIC approximation error vanishes and the stochastic estimation term reduces to $\sqrt{p_k\log p_k/(np_{-k})}$, which is consistent with optimal rates in related literature \citep{cai2016estimating, lyu2019tensor, franks2026near}.

\section{Proofs of Main Results}\label{app:proof}

\subsection{Proofs of Theorems \ref{thm-step-existence} and \ref{thm-existence}}

\begin{proof}[Proof of Theorem \ref{thm-step-existence}]
Define the full index set $\bs_i^* = \{i\} \cup \bs_i$. By definition of the SIC projection for a covariance matrix $\bfSigma$, $\Pi(\bfSigma, \cS_k) \in \mathbf{GL}(p_k)$ if and only if the submatrix $\bfSigma_{\bs_i^*, \bs_i^*}$ is non-singular (full rank) for all $i = 1, \dots, p_k$. In our setting, it suffices to prove that the pseudo-covariance block
\[
\sum_{j=1}^n \bX_{j\bs_i^*}^{(k)} \bL_{-k} \bL_{-k}^\top \bX_{j\bs_i^*}^{(k)\top}
\]
is positive definite for all $i = 1, \ldots, p_k$. 

This matrix is a sum of positive semidefinite matrices. A vector $\bm{y}$ belongs to its kernel if and only if $\bm{y}$ lies in the kernel of $(\bX_{j\bs_i^*}^{(k)} \bL_{-k})^\top$ for all $j = 1, \dots, n$. Thus, we need only show that for any $i$, no non-zero vector lies in the common kernel of the matrices $\bL_{-k}^\top \bX_{j\bs_i^*}^{(k)\top}$ across all $j$.

Fix $i \in \{1, \dots, p_k\}$. Under Assumption \ref{asmp-exist-2}, the sum $\sum_{j=1}^n \bX_{j\bs_i^*}^{(k)} \bX_{j\bs_i^*}^{(k)\top}$ is almost surely full rank as a sum of $n p_{-k} \ge m_k + 1 \ge |\bs_i^*|$ rank-one matrices of dimension $|\bs_i^*| \times |\bs_i^*|$. Consequently, no non-zero vector lies in the intersection of the kernels of $\bX_{j\bs_i^*}^{(k)\top}$ across all $j$. Since $\bL_l$ is non-singular with positive diagonal entries for all $l \neq k$, the Kronecker product $\bL_{-k} = \bigotimes_{l \neq k} \bL_l$ is non-singular. Therefore, no non-zero vector can lie in the common kernel of $\bL_{-k}^\top \bX_{j\bs_i^*}^{(k)\top}$ across $j$, establishing that $\sum_{j=1}^n \bX_{j\bs_i^*}^{(k)} \bL_{-k} \bL_{-k}^\top \bX_{j\bs_i^*}^{(k)\top}$ is strictly positive definite.
\end{proof}

\begin{lemma}\label{lemma-const-trace}
Let $\bfSigma \in \mathbf{PD}(p)$ be a $p \times p$ positive definite covariance matrix, and let $\cS$ be a valid sparse matrix space. If $\bL = \Pi(\bfSigma, \cS)$ is the SIC projection of $\bfSigma$ onto $\cS$, then 
\[
\mathrm{tr}(\bL \bL^\top \bfSigma) = p.
\]
\end{lemma}

\begin{proof}[Proof of Lemma \ref{lemma-const-trace}]
Let $\bfepsilon \sim \mathcal{N}(\mathbf{0}, \bfSigma)$. Then
\[
\mathrm{tr}(\bL \bL^\top \bfSigma) = \mathbb{E}\left[ \bfepsilon^\top \bL \bL^\top \bfepsilon \right].
\]
According to Section \ref{subsec:sic_review}, each column $i$ of $\bL$ has the closed-form expression $\bL_{\bs_i^*, i} = \bfbeta_i (\mathbf{e}_1^\top \bfbeta_i)^{-1/2}$, where $\bfbeta_i = (\bfSigma_{\bs_i^*, \bs_i^*})^{-1} \mathbf{e}_1$. Expanding the expectation yields:
\[
\begin{aligned}
\mathrm{tr}(\bL \bL^\top \bfSigma) 
&= \sum_{i=1}^p (\mathbf{e}_1^\top \bfbeta_i)^{-1} \mathbb{E}\left[ (\bfepsilon_{\bs_i^*}^\top \bfbeta_i)^2 \right] \\
&= \sum_{i=1}^p \left(\mathbf{e}_1^\top (\bfSigma_{\bs_i^*, \bs_i^*})^{-1} \mathbf{e}_1\right)^{-1} \mathbb{E}\left[ \frac{(\bfepsilon_i - \bfSigma_{\bs_i, i}^\top \bfSigma_{\bs_i, \bs_i}^{-1} \bfepsilon_{\bs_i})^2}{(\bfSigma_{i,i} - \bfSigma_{\bs_i, i}^\top \bfSigma_{\bs_i, \bs_i}^{-1} \bfSigma_{\bs_i, i})^2} \right] \\
&= \sum_{i=1}^p \mathbb{E}\left[ \frac{(\bfepsilon_i - \bfSigma_{\bs_i, i}^\top \bfSigma_{\bs_i, \bs_i}^{-1} \bfepsilon_{\bs_i})^2}{\bfSigma_{i,i} - \bfSigma_{\bs_i, i}^\top \bfSigma_{\bs_i, \bs_i}^{-1} \bfSigma_{\bs_i, i}} \right] \\
&= \sum_{i=1}^p \frac{\bfSigma_{i,i} - \bfSigma_{\bs_i, i}^\top \bfSigma_{\bs_i, \bs_i}^{-1} \bfSigma_{\bs_i, i}}{\bfSigma_{i,i} - \bfSigma_{\bs_i, i}^\top \bfSigma_{\bs_i, \bs_i}^{-1} \bfSigma_{\bs_i, i}} = p.
\end{aligned}
\]
Hence, the trace value is identically equal to $p$ and independent of the specific sparsity constraints in $\cS$.
\end{proof}

\begin{proof}[Proof of Theorem \ref{thm-existence}]

\quad

\noindent \textbf{Part 1: $K=2$ case.} 

\quad

We first consider the $K=2$ mode case and subsequently discuss the extension to $K > 2$. Expanding the forward-KL divergence objective \eqref{eq:forward_kl_obj}, we write:
\[
\begin{aligned}
g(\bL_1, \bL_2) &\coloneqq 2\mathrm{KL}\left( \mathcal{N}(\mathbf{0}, \mathbf{\Sigma}) \;\|\; \mathcal{N}\big(\mathbf{0}, (\mathbf{L}\mathbf{L}^\top)^{-1}\big) \right) \\
&= \mathrm{tr}(\bL \bL^\top \bfSigma) - \log\det(\bL \bL^\top) - \log\det(\bfSigma) -p_1p_2 \\
&= p_2 \mathrm{tr}\big(\bL_1 \bL_1^\top \tilde{\bfSigma}_1(\bL_2)\big) - p_2 \log\det(\bL_1 \bL_1^\top) - p_1 \log\det(\bL_2 \bL_2^\top) - \log\det(\bfSigma) -p_1p_2,
\end{aligned}
\]
where $\tilde{\bfSigma}_1(\bL_2)$ is the pseudo-covariance matrix defined in Proposition \ref{prop-low_rank}. By Theorem \ref{thm-step-existence}, for any fixed non-singular factor $\bL_2 \in \mathbf{GL}(p_2)$, the conditional objective $g_2(\bL_1) \coloneqq g(\bL_1, \bL_2)$ achieves its unique minimum at
\[
\bL_1(\bL_2) = \Pi\big(\tilde{\bfSigma}_1(\bL_2), \cS_1\big) \in \mathbf{GL}(p_1).
\]
To establish the existence of a positive definite minimizer of $g(\bL_1, \bL_2)$ over $\cS_{\text{KS}}$, it suffices to show that the profile function
\[
g\left(\Pi\big(\tilde{\bfSigma}_1(\bL_2), \cS_1\big), \bL_2\right)
\]
attains its minimum on $\cS_2 \cap \mathbf{GL}(p_2)$. Because $\bL_1(\bL_2)$ is the SIC projection of $\tilde{\bfSigma}_1(\bL_2)$, Lemma \ref{lemma-const-trace} implies $\mathrm{tr}\big(\bL_1 \bL_1^\top \tilde{\bfSigma}_1(\bL_2)\big) = p_1$ identically.
Minimizing $g(\bL_1, \bL_2)$ is therefore equivalent to maximizing
\begin{equation}\label{eq:proof-thm-main-obj}
\begin{aligned}
f(\bL_2) &\coloneqq p_2 \log\det\big(\bL_1(\bL_2) \bL_1^\top(\bL_2)\big) + p_1 \log\det(\bL_2 \bL_2^\top) \\
&= 2 p_2 \log\det\big(\bL_1(\bL_2)\big) + 2 p_1 \log\det(\bL_2).
\end{aligned}
\end{equation}
In addition, since $g(\bL_1, \bL_2)$ is always positive as KL divergence, whenever the KSIC projection $\bL$ exists, we have
\begin{equation}\label{eq:proof-thm-main-K2-det}
\log\det(\bL\bL^\top) + \log\det\bfSigma \leq 0.
\end{equation}

Suppose for contradiction that $\sup f(\bL_2)$ is not attained on $\cS_2\cap \mathbf{GL}(p_2)$. There exists a sequence $\{\bL_2^{(t)}\}_{t=1}^\infty \subset \mathbf{GL}(p_2)$ such that $f(\bL_2^{(t)}) \to \sup f(\bL_2)$ as $t \to \infty$. Define $\bfPsi^{(t)} \coloneqq \bL_2^{(t)} \bL_2^{(t)\top} = \bM^{(t)} \bD^{(t)} \bM^{(t)\top}$, where $\bM^{(t)}$ is obtained by column-normalizing $\bL_2^{(t)}$ via $\bM_{\cdot, i} = \bL_{2, \cdot i} / \|\bL_{2, \cdot i}\|_2$. Since the space of normalized matrices in $\cS_2$ is compact, we may pass to a convergent subsequence if necessary and assume $\bM^{(t)} \to \bM$. By passing to further subsequences, we may assume each diagonal element $D_{ii}^{(t)}$ converges to $0$, $+\infty$, or a finite positive limit. 

Applying Lemma A.1 of \citet{drton2021existence}, we decompose $\bfPsi^{(t)}$ as
\begin{equation}
\bfPsi^{(t)} = \epsilon_1^{(t)} \bfPsi_1^{(t)} + \epsilon_2^{(t)} \bfPsi_2^{(t)} + \cdots + \epsilon_D^{(t)} \bfPsi_D^{(t)},
\end{equation}
such that:
\begin{itemize}
    \item $\bfPsi_d^{(t)}$ is a sequence of positive semidefinite matrices converging to $\bfPsi_d$ for each $d = 1, \dots, D$;
    \item $\bigoplus_{d=1}^D \operatorname{Im}(\bfPsi_d^{(t)}) = \mathbb{R}^{p_2}$;
    \item $\operatorname{rank}(\bfPsi_d^{(t)}) = \operatorname{rank}(\bfPsi_d)$ for all $d$;
    \item $\epsilon_d^{(t)} > 0$ with ordering $\epsilon_{d+1}^{(t)} / \epsilon_d^{(t)} \to 0$ as $t \to \infty$.
\end{itemize}

To analyze the determinant of $\bfPsi^{(t)}$, we reparameterize $\bfPsi_d^{(t)}$ and $\epsilon_d^{(t)}$ as
\[
\begin{aligned}
&\epsilon_d^{(t)} \coloneqq \epsilon_d^{(t)} - \epsilon_{d+1}^{(t)} \quad (d = 1, \dots, D-1), \qquad \epsilon_D^{(t)} \coloneqq \epsilon_D^{(t)}; \\
&\bfPsi_d^{(t)} \coloneqq \bfPsi_1^{(t)} + \cdots + \bfPsi_d^{(t)}, \qquad \bfPsi_d \coloneqq \bfPsi_1 + \cdots + \bfPsi_d \quad (d = 1, \dots, D).
\end{aligned}
\]
The properties above hold under this reparameterization, and the column spaces of $\bfPsi_d^{(t)}$ now satisfy the nested containment
\[
\operatorname{Im}(\bfPsi_1^{(t)}) \subset \cdots \subset \operatorname{Im}(\bfPsi_D^{(t)}) = \mathbb{R}^{p_2}.
\]
Let $h_d = \operatorname{rank}(\bfPsi_d)$. By Lemma 3.2 of \citet{drton2021existence},
\begin{equation}\label{eq:proof-thm-main-1}
\det(\bL_2^{(t)})^2 = \det(\bfPsi^{(t)}) \asymp (\epsilon_D^{(t)})^{p_2} (\gamma_1^{(t)})^{h_1} (\gamma_2^{(t)})^{h_2} \cdots (\gamma_{D-1}^{(t)})^{h_{D-1}},
\end{equation}
where $\gamma_d^{(t)} \coloneqq \epsilon_d^{(t)} / \epsilon_{d+1}^{(t)} \to \infty$ as $t \to \infty$.

We next analyze the first determinant term in $f(\bL_2)$. Using the closed-form representation of the SIC projection,
\begin{equation}\label{proof-thm_L1}
\begin{aligned}
\det\left(\bL_1\big(\bL_2^{(t)}\big)\right) &= \prod_{i=1}^{p_1} \bL_{1, ii} = \prod_{i=1}^{p_1} \left( \tilde{\bfSigma}_{1, ii} - \tilde{\bfSigma}_{1, i\bs_i} \tilde{\bfSigma}_{1, \bs_i \bs_i}^{-1} \tilde{\bfSigma}_{1, \bs_i i} \right)^{-1/2} \\
&= \left( \prod_{i=1}^{p_1} \frac{\det(\tilde{\bfSigma}_{1, \bs_i^* \bs_i^*})}{\det(\tilde{\bfSigma}_{1, \bs_i \bs_i})} \right)^{-1/2} 
= \left( \prod_{i=1}^{p_1} \frac{\det\left( \sum_{j=1}^n \bX_{j \bs_i^*, \cdot}^{(1)} \bfPsi^{(t)} \bX_{j \cdot, \bs_i^*}^{(1)} \right)}{\det\left( \sum_{j=1}^n \bX_{j \bs_i, \cdot}^{(1)} \bfPsi^{(t)} \bX_{j \cdot, \bs_i}^{(1)} \right)} \right)^{-1/2}.
\end{aligned}
\end{equation}
For $i = 1, \dots, p_1$ and $d = 1, \dots, D$, define
\[
\begin{aligned}
\bM_i^{(t)} &\coloneqq \sum_{r=1}^n \bX_{r \bs_i^*, \cdot}^{(1)} \bfPsi^{(t)} \bX_{r \cdot, \bs_i^*}^{(1)}, \qquad &\bM_{i,d}^{(t)} &\coloneqq \sum_{r=1}^n \bX_{r \bs_i^*, \cdot}^{(1)} \bfPsi_d^{(t)} \bX_{r \cdot, \bs_i^*}^{(1)}, \\
\bN_i^{(t)} &\coloneqq \sum_{r=1}^n \bX_{r \bs_i, \cdot}^{(1)} \bfPsi^{(t)} \bX_{r \cdot, \bs_i}^{(1)}, \qquad &\bN_{i,d}^{(t)} &\coloneqq \sum_{r=1}^n \bX_{r \bs_i, \cdot}^{(1)} \bfPsi_d^{(t)} \bX_{r \cdot, \bs_i}^{(1)}.
\end{aligned}
\]
The limits $\bM_{i,d} = \lim_{t\to\infty} \bM_{i,d}^{(t)}$ and $\bN_{i,d} = \lim_{t\to\infty} \bN_{i,d}^{(t)}$ are well-defined. Defining $d_{M,i} \coloneqq \min \{d : \operatorname{Im}(\bM_{i,d}^{(t)}) = \mathbb{R}^{|\bs_i^*|}\}$ and $d_{N,i} \coloneqq \min \{d : \operatorname{Im}(\bN_{i,d}^{(t)}) = \mathbb{R}^{|\bs_i|}\}$, we have the nested rank structures:
\[
\begin{aligned}
\operatorname{Im}(\bM_{i,1}^{(t)}) \subset \cdots \subset \operatorname{Im}(\bM_{i, d_{M,i}}^{(t)}) = \cdots = \operatorname{Im}(\bM_{i, D}^{(t)}) = \mathbb{R}^{|\bs_i^*|}, \\
\operatorname{Im}(\bN_{i,1}^{(t)}) \subset \cdots \subset \operatorname{Im}(\bN_{i, d_{N,i}}^{(t)}) = \cdots = \operatorname{Im}(\bN_{i, D}^{(t)}) = \mathbb{R}^{|\bs_i|}.
\end{aligned}
\]
Letting $h_d^{M,i} = \operatorname{rank}(\bM_{i,d}^{(t)})$ and $h_d^{N,i} = \operatorname{rank}(\bN_{i,d}^{(t)})$, we obtain the asymptotic expansions
\[
\begin{aligned}
\det(\bM_i^{(t)}) &\asymp (\epsilon_D^{(t)})^{|\bs_i^*|} (\gamma_1^{(t)})^{h_1^{M,i}} (\gamma_2^{(t)})^{h_2^{M,i}} \cdots (\gamma_{d_{M,i}-1}^{(t)})^{h_{d_{M,i}-1}^{M,i}} (\gamma_{d_{M,i}}^{(t)})^{|\bs_i^*|} \cdots (\gamma_{D-1}^{(t)})^{|\bs_i^*|}, \\
\det(\bN_i^{(t)}) &\asymp (\epsilon_D^{(t)})^{|\bs_i|} (\gamma_1^{(t)})^{h_1^{N,i}} (\gamma_2^{(t)})^{h_2^{N,i}} \cdots (\gamma_{d_{N,i}-1}^{(t)})^{h_{d_{N,i}-1}^{N,i}} (\gamma_{d_{N,i}}^{(t)})^{|\bs_i|} \cdots (\gamma_{D-1}^{(t)})^{|\bs_i|},
\end{aligned}
\]
where  we align the index offsets such that $d_{N,i} = d_{M,i}$ without loss of generality. Defining the total rank difference $\delta_d \coloneqq \sum_{i=1}^{p_1} (h_d^{M,i} - h_d^{N,i})$ and substituting these expansions into \eqref{proof-thm_L1} gives
\begin{equation}\label{eq:proof-thm-main-2}
\det\left(\bL_1\big(\bL_2^{(t)}\big) \bL_1^\top\big(\bL_2^{(t)}\big)\right) \asymp (\epsilon_D^{(t)})^{-p_1} (\gamma_1^{(t)})^{-\delta_1} (\gamma_2^{(t)})^{-\delta_2} \cdots (\gamma_{D-1}^{(t)})^{-\delta_{D-1}}.
\end{equation}

Note that $\delta_d$ can be expressed as
\[
\delta_d = \sum_{i=1}^{p_1} \left[ \operatorname{rank}\left( [\bX_{1\bs_i^*, \cdot}^{(1)}, \dots, \bX_{n\bs_i^*, \cdot}^{(1)}] (\bI_n \otimes \bfPsi_d^{(t)1/2}) \right) - \operatorname{rank}\left( [\bX_{1\bs_i, \cdot}^{(1)}, \dots, \bX_{n\bs_i, \cdot}^{(1)}] (\bI_n \otimes \bfPsi_d^{(t)1/2}) \right) \right].
\]
Using the shorthand notation
\[
\Delta_i(\bZ) \coloneqq \operatorname{rank}\left( [\bX_{1\bs_i^*, \cdot}^{(1)}, \dots, \bX_{n\bs_i^*, \cdot}^{(1)}] (\bI_n \otimes \bZ) \right) - \operatorname{rank}\left( [\bX_{1\bs_i, \cdot}^{(1)}, \dots, \bX_{n\bs_i, \cdot}^{(1)}] (\bI_n \otimes \bZ) \right),
\]
we lower bound $\delta_d$ by
\begin{equation}\label{eq:proof-thm-main-3}
\delta_d = \sum_{i=1}^{p_1} \Delta_i(\bfPsi_d^{(t)1/2}) \ge \min_{\operatorname{rank}(\bZ) = h_d} \sum_{i=1}^{p_1} \Delta_i(\bZ).
\end{equation}

Define the nested index sets $\bt_i \coloneqq \{1, \dots, i-1\}$ and $\bt_i^* \coloneqq \bt_i \cup \{i\}$. Since $\bs_i \subseteq \bt_i$, 
\[
\Delta_i(\bZ) \ge \operatorname{rank}\left( [\bX_{1\bt_i^*, \cdot}^{(1)}, \dots, \bX_{n\bt_i^*, \cdot}^{(1)}] (\bI_n \otimes \bZ) \right) - \operatorname{rank}\left( [\bX_{1\bt_i, \cdot}^{(1)}, \dots, \bX_{n\bt_i, \cdot}^{(1)}] (\bI_n \otimes \bZ) \right).
\]
Summing over $i = 1, \dots, p_1$ yields a telescoping sum:
\[
\begin{aligned}
\delta_d &\ge \min_{\operatorname{rank}(\bZ) = h_d} \sum_{i=1}^{p_1} \left[ \operatorname{rank}\left( [\bX_{1\bt_i^*, \cdot}^{(1)}, \dots, \bX_{n\bt_i^*, \cdot}^{(1)}] (\bI_n \otimes \bZ) \right) - \operatorname{rank}\left( [\bX_{1\bt_i, \cdot}^{(1)}, \dots, \bX_{n\bt_i, \cdot}^{(1)}] (\bI_n \otimes \bZ) \right) \right] \\
&= \min_{\operatorname{rank}(\bZ) = h_d} \operatorname{rank}\left( [\bX_1^{(1)}, \dots, \bX_n^{(1)}] (\bI_n \otimes \bZ) \right).
\end{aligned}
\]
By Corollary 2 of \citet{soloveychik2016gaussian}, when $n > p_1/p_2 + p_2/p_1$,
\[
\min_{\operatorname{rank}(\bZ) = h_d} \operatorname{rank}\left( [\bX_1^{(1)}, \dots, \bX_n^{(1)}] (\bI_n \otimes \bZ) \right) > \frac{p_1 h_d}{p_2} \quad \text{a.s.}
\]
Hence $\delta_d \ge p_1 h_d / p_2$, which implies $p_1 h_d - p_2 \delta_d \le 0$. Substituting \eqref{eq:proof-thm-main-1} and \eqref{eq:proof-thm-main-2} into $f(\bL_2^{(t)})$, the leading scalar factors $(\epsilon_D^{(t)})^{p_1 p_2}$ cancel, yielding
\[
f(\bL_2^{(t)}) + C = 2 \sum_{d=1}^{D-1} (p_1 h_d - p_2 \delta_d) \log(\gamma_d^{(t)}).
\]
Since $\gamma_d^{(t)} \to \infty$ and $p_1 h_d - p_2 \delta_d \le 0$, we have $f(\bL_2^{(t)}) \to -\infty$ as $t \to \infty$. This contradicts the assumption that $\lim_t f(\bL^{(t)}_2) \to \sup f(\bL_2)$. Therefore, $f$ attains its maximum at an interior point $\bL_2 \in \cS_2 \cap \mathbf{GL}(p_2)$. In particular, this result holds when $\bL_1(\bL_2^{(t)}) = \tilde{\bfSigma}^{-1/2}_1(\bL_2)$, where $\cS_1$ corresponds to the lower triangular matrices without any sparsity constraint and $\bs_i = \bt_i$ for all $i = 1, \cdots, p_1$. 

Combining this result with Theorem \ref{thm-step-existence}, the global minimum of the forward-KL divergence is attained at non-singular factors $\bL_k \in \cS_k \cap \mathbf{GL}(p_k)$ for $k = 1, 2$. 

\quad

\noindent We conclude Part 1 by summarizing the mode-$2$ results:
\begin{itemize}
\item[\textbf{A}] Under Assumption \ref{asmp-exist-2}, when \[n > p_1/p_2 + p2/p_1,\] the KSIC projection $\bL = \bigotimes_{k=1}^2\bL_k = \Pi_{KS}(\bfSigma, \cS_{KS})$ is attained at non-singular factors $\bL_k \in \cS_k \cap \mathbf{GL}(p_k)$ for $k = 1, 2$. 
\item[\textbf{B}] According to \eqref{eq:proof-thm-main-K2-det}, in the meantime,
\[
\log\det(\bL\bL^\top) + \log\det\bfSigma \leq 0.
\]
\end{itemize}
\quad

\noindent \textbf{Part 2: Mathematical induction.} 

\quad

\noindent For $K > 2$, we proceed by induction. 
Suppose the following results hold in order-$(K-1)$ case:
\begin{itemize}
\item[\textbf{A}] Under Assumption \ref{asmp-exist-2}, when \[n > \max_{k_1\neq k_2}\left\{\frac{1}{p^2_{k_1}}+ \frac{1}{p^2_{k_2}}\right\}p,\] the KSIC projection $\bL = \bigotimes_{k=1}^2\bL_k = \Pi_{KS}(\bfSigma, \cS_{KS})$ is attained at non-singular factors $\bL_k \in \cS_k \cap \mathbf{GL}(p_k)$ for $k = 1, \cdots, K-1$. 
\item[\textbf{B}] Under the same setting as \textbf{result A}, 
\[
\log\det(\bL\bL^\top) + \log\det\bfSigma \leq 0.
\]
\end{itemize}
We want to show that these results also hold in order-$K$ case.

\quad

\noindent We write the order-$K$ forward-KL divergence objective as
\begin{equation}
    \begin{aligned}
g(\bL_1,\cdots, \bL_K) &\coloneqq 2\mathrm{KL}\left( \mathcal{N}(\mathbf{0}, \mathbf{\Sigma}) \;\|\; \mathcal{N}\big(\mathbf{0}, (\mathbf{L}\mathbf{L}^\top)^{-1}\big) \right) \\
&= \mathrm{tr}(\bL \bL^\top \bfSigma) - \log\det(\bL \bL^\top) -\log\det\bfSigma- p \\
&= p_{K} \mathrm{tr}\big(\bL_{-K} \bL_{-K}^\top \tilde{\bfSigma}_{-K}(\bL_{K})\big) - p_K \log\det(\bL_{-K} \bL_{-K}^\top) - p_{-K} \log\det(\bL_{K} \bL_{K}^\top)-\log\det\bfSigma -p.
\end{aligned}
\end{equation}
Given the above equation, by Lemma \ref{lemma-const-trace}, it's already clear that \textbf{order-$K$ result B} holds whenever the KSIC projection is attained at strictly positive definite factors. We only need to verify \textbf{order-$K$ result A}.

By the \textbf{order-$(K-1)$ result A}, if 
\begin{equation}\label{eq:proof-thm-K-1}
\operatorname{rank}(\tilde{\bfSigma}_{-K}(\bL_{K})) \geq \max_{k_1, k_2 \in 1\cdots K-1}\left\{\frac{1}{p^2_{k_1}} + \frac{1}{p^2_{k_2}}\right\}\prod_{k=1}^{K-1}p_k,
\end{equation}
the conditional objective $g_K(\bL_{-K}) = g(\bL_1,\cdots, \bL_K)$ attains its minimum at some $\bL_{-K}(\bL_K) \in \cS_{-K} \cap \mathbf{GL}(p_{-K})$. 

Suppose \eqref{eq:proof-thm-K-1} holds, minimizing $g(\bL_1,\cdots, \bL_K)$ over $\bL_1,\ldots,\bL_K$ is equivalent to maximizing
\begin{equation}\label{eq:proof-thm-K-2}
  f(\bL_K) = p_K \log\det\big(\bL_{-K}(\bL_K)\big) + p_{-K} \log\det(\bL_{K}).
\end{equation}
According to the proof in part 1, if we can show that \eqref{eq:proof-thm-K-2} is bounded above, \textbf{order-$K$ result A} will then hold.

Our task turns to showing that under the target order-$K$ assumption
\begin{equation}\label{eq:proof-thm-K-main}
    n \geq \max_{k_1, k_2 \in 1\cdots K}\left\{\frac{1}{p^2_{k_1}} + \frac{1}{p^2_{k_2}}\right\}\prod_{k=1}^{K}p_k, \quad \text{and} \quad 
\end{equation}
the following two conditions hold:
\begin{itemize}
    \item[\textbf{I:}] \eqref{eq:proof-thm-K-1} holds whenever $\bL_{K} \in \mathbf{GL}(p_{K})$;
    \item[\textbf{II:}] \eqref{eq:proof-thm-K-2} is bounded above.
\end{itemize}
We first verify condition \textbf{I} by analyzing the rank of $\tilde{\bfSigma}_{-K}(\bL_{K})$. By Proposition~\ref{prop-low_rank},
\[
\tilde{\bfSigma}_{-K}(\bL_{K}) = \frac{1}{p_{K}}\sum_{r=1}^n\bX^{(-K)}_{r}\bL_{K}\bL^\top_{K}\bX^{(-K)\top}_{r},
\]
where $\bX_r^{(-K)}$ denotes the $p_{-K} \times p_{K}$ matrix obtained by unfolding $\bx_r$ along all modes except mode $K$. Define
\[
\bW_{K} = [\bX^{(-K)}_{1}\bL_{K}, \cdots, \bX^{(-K)}_{n}\bL_{K}].
\]
Since
\[
\tilde{\bfSigma}_{-K}(\bL_{K}) = \frac{1}{p_K}\bW_{K}\bW^\top_{K},
\]
and $\bL_{K} \in \mathbf{GL}(p_{K})$, we have
\[
\operatorname{rank}\big(\tilde{\bfSigma}_{-K}(\bL_{K})\big)
=
\operatorname{rank}(\bW_{K})
=
\operatorname{rank}([\bX^{(-K)}_{1}, \cdots, \bX^{(-K)}_{n}]).
\]
Suppose that $\operatorname{rank}([\bX^{(-K)}_{1}, \cdots, \bX^{(-K)}_{n}]) = r_0$. Then all columns of every $\bX^{(-K)}_{r}$ lie in the same $r_0$-dimensional subspace $\cS \subseteq \mathbb{R}^{p_{-K}}$. Hence each $\bX^{(-K)}_{r}$ belongs to $\cS\otimes \mathbb{R}^{p_{K}}$, which has dimension $r_0p_{K}$. On the other hand, by Assumption~\ref{asmp-exist-2}, the vectorized observations $\bx_1, \cdots, \bx_n$ are linearly independent. Therefore,
\[
n \leq r_0p_{K} \quad \text{or equivalently} \quad r_0\geq n/p_K.
\]
When \eqref{eq:proof-thm-K-main} holds,
\[
\begin{split}
 r_0\geq n/p_K 
 &\geq \max_{k_1, k_2 \in 1\cdots K}\left\{\frac{1}{p^2_{k_1}} + \frac{1}{p^2_{k_2}}\right\}\prod_{k=1}^{K-1}p_k \geq \max_{k_1, k_2 \in 1\cdots K-1}\left\{\frac{1}{p^2_{k_1}} + \frac{1}{p^2_{k_2}}\right\}\prod_{k=1}^{K-1}p_k .
\end{split}
\]
Consequently, \eqref{eq:proof-thm-K-1} holds, and condition \textbf{I} is satisfied.

We next verify condition \textbf{II}. Notice that by \textbf{order-$(K-1)$ result 2}, 
\[
\det(\bL_{-K}\bL^\top_{-K}) \leq \det(\tilde{\bfSigma}_{-K}^{-1}),
\]
and consequently:
\[
f(\bL_K) \leq -\frac{1}{2}p_K \log\det\big(\bfSigma_{-K}\big) + p_{-K} \log\det(\bL_{K}).
\]
We can show that when assuming $f(\bL_K)$ is not bounded above, one can not find sequence $\{\bL_K^{(t)}\}_{t=1}^\infty \subset \mathbf{GL}(p_K)$ such that 
\begin{equation}
    -\frac{1}{2}p_K \log\det\big(\bfSigma_{-K}(\bL_{K}^{(t)})\big) + p_{-K} \log\det(\bL_{K}^{(t)}) \to \infty \quad \text{as} \quad t\to\infty,
\end{equation}
which lead to a contradiction. The argument can be made following  the same steps as in the proof of part 1, where $p_{K}$ plays the role of $p_2$ and $p_{-K}$ plays the role of $p_1$. Since \textbf{order-$2$ result 1} required $n > p_1/p_2 + p_2/p_1$, 
\[
n \geq \frac{p_{-K}}{p_K} +\frac{p_{K}}{p_{-K}}
\]
here ensures that \eqref{eq:proof-thm-K-2} is bounded above. This condition is guaranteed by the target order-$K$ assumption \eqref{eq:proof-thm-K-main}. The \textbf{order-$K$ result 1} is now proved.

In conclusion, we have shown that if the \textbf{order-$(K-1)$} results hold, then the \textbf{order-$K$ results} also hold. Since we have established $K=2$ case in the first part of the proof, by \textbf{order-$K$ result 1}, the general order-$K$ existence threshold is
\[
n \geq \max_{k_1, k_2 \in 1\cdots K}\left\{\frac{1}{p^2_{k_1}} + \frac{1}{p^2_{k_2}}\right\}\prod_{k=1}^{K}p_k.
\]
\end{proof}

\subsection{Proofs of Theorem \ref{thm-rate} and Theorem \ref{col-main}}
\begin{lemma}\label{lemma-asym-eigen-max}
Let $\tilde{\bL}_{k} = \Pi(\tilde{\bfSigma}_{k}, \cS_k)$ be the population one-step SIC estimator, where $\tilde{\bfSigma}_{k}$ is the pseudo covariance defined by Proposition \ref{prop:pseudocov}. Then
\[
\lambda_{max}(\tilde{\bL}_k\tilde{\bL}_k^\top) \leq \frac{p_k}{\tr(\bL_{-k}\bL_{-k}^\top)\lambda_{min}(\bfSigma)}
\]
\end{lemma}
\begin{proof}
First, by Lemma \ref{lemma-const-trace}, we have,
\[
\begin{split}
    p_k = \tr(\tilde{\bL}_k\tilde{\bL}_k^\top\tilde{\bfSigma}_k)  \geq \tr(\tilde{\bL}_k\tilde{\bL}_k^\top)\lambda_{min}(\tilde{\bfSigma}_k).
\end{split}
\]
Then, since $\lambda_{max}(\tilde{\bL}_k\tilde{\bL}_k^\top) \leq \tr(\tilde{\bL}_k\tilde{\bL}_k^\top)$,
\[
\lambda_{max}(\tilde{\bL}_k\tilde{\bL}_k^\top) \leq \tr(\tilde{\bL}_k\tilde{\bL}_k^\top) \leq \frac{p_k}{\lambda_{min}(\tilde{\bfSigma}_k)}.
\]
The lower bound for eigenvalues of $\tilde{\bfSigma}_k$ can be obtained from the proof of Proposition \ref{prop:pseudocov}. Specifically, if using the notation $\bfOmega_{-k} = \bL_{-k}\bL_{-k}^\top$:
\[
\begin{split}
\bx^\top\tilde{\bfSigma}_k\bx
&=\sum_{i,j=1}^{p_{-k}}(\bfOmega_{-k})_{i,j}\bx^\top\bfSigma_{(i,j)}^{(k)}\bx
=\tr\left((\bx\bx^\top\otimes\bfOmega_{-k})\bfSigma\right)\\
&\ge \lambda_{min}(\bfSigma)\tr(\bx\bx^\top\otimes\bfOmega_{-k}) = \lambda_{min}(\bfSigma)\tr(\bx\bx^\top)\tr(\bfOmega_{-k})\\
& = \lambda_{min}(\bfSigma)\tr(\bL_{-k}\bL_{-k}^\top)\bx^\top\bx.
\end{split}
\]
As a result, $\lambda_{min}(\tilde{\bfSigma}) \geq \lambda_{min}(\bfSigma)\tr(\bL_{-k}\bL_{-k}^\top)$, and we conclude that 
\[
\lambda_{max}(\tilde{\bL}_k\tilde{\bL}_k^\top) \leq \frac{p_k}{\tr(\bL_{-k}\bL_{-k}^\top)\lambda_{min}(\bfSigma)}
\]
\end{proof}

\begin{lemma}\label{lemma-asym-1}
Given $n$ multi-way replicates $\{\bX_1, \ldots,  \bX_n \in \mathbb{R}^{p\times q}\}$ with $\text{vec}(\bX) \sim \mathcal{N}(0, \bfSigma)$,
Define $\bE_{s,t}$ be discrepancy matrix between a positive definite matrix $\bfOmega$ and the $(s, t)$th $p\times p$ sub-covariance-matrix of $\bfSigma$ corresponding to the $s$th and $t$th columns of $\bX$, 
\[
\bE_{s,t} := \bfOmega^{1/2}\bfSigma_{(s,t)}\bfOmega^{\top/2}.
\]
If $\log(q) = o(np)$, then for any symmetric and positive definite $p\times p$ matrix $\bfOmega$, we have
\begin{equation}
\|\frac{1}{np}\sum_{r=1}^n\bX_r^\top\bfOmega\bX_r - \frac{1}{p}\mathbb{E}\big\{\bX^\top\bfOmega\bX\big\}\|_{\max} = \mathcal{O}_p\left(\max\limits_{s,t}\|\bE_{(s,t)}\|_2\sqrt{\frac{\log q}{np}}\right)
\end{equation}
\end{lemma}
\begin{proof}[Proof of Lemma \ref{lemma-asym-1}]
The proof is similar to the one of Lemma S.1 in \cite{lyu2019tensor}, but in a more general setting where the covariance is not necessarily separable. First, consider the $(s,t)$th entry of matrix $\bX^\top\bfOmega\bX$, i.e. $\bX_{\cdot s}^\top\bfOmega\bX_{\cdot t}$. The $s$th and $t$th columns of $\bX$ follow a normal distribution:
\[
\begin{bmatrix}
\bX_{\cdot s} \\
\bX_{\cdot t}
\end{bmatrix} 
\sim \mathcal{N}\left(0, 
\begin{bmatrix}
\bfSigma_{(s,s)} & \bfSigma_{(s,t)} \\
\bfSigma_{(t,s)} & \bfSigma_{(t,t)} 
\end{bmatrix} 
\right) \in \mathbb{R}^{2p\times 2p}
\]
We now introduce two $p \times 2p$ auxiliary matrices $\bM_s$ and $\bM_t$ representing the Cholesky factor of the joint covariance matrix, such that 
\[
\begin{bmatrix}
\bM_s \\
\bM_t
\end{bmatrix} =
\begin{bmatrix}
\bfSigma_{(s,s)} & \bfSigma_{(s,t)} \\
\bfSigma_{(t,s)} & \bfSigma_{(t,t)} 
\end{bmatrix}^{\frac{1}{2}}.
\]
By definition, $\bM_s\bM_s^\top = \bfSigma_{(s,s)}$, $\bM_s\bM_t^\top = \bfSigma_{(s,t)}$, and $\bM_t\bM_t^\top = \bfSigma_{(t,t)}$.
Now $\bX_{\cdot s}$ and $\bX_{\cdot t}$ can be rewritten using i.i.d standard Gaussian $2p \times 1$ random vector $\by$, such that $\bX_{\cdot s} = \bM_s\by$ and $\bX_{\cdot t} = \bM_t\by$. In order to reveal the concentration property of the empirical average, for given $s$ and $t$, we introduce length $p$ vectors $\ba_{r} = \bfOmega^{1/2}\bM_{s}\by_r$ and $\bb_{r} = \bfOmega^{1/2}\bM_{t}\by_r$:
\begin{equation}\label{eq-lemma1-1}
\begin{split}
\Big\{\frac{1}{np}\sum_{r=1}^n\bX_{r}^\top\bfOmega\bX_{r}\Big\}_{s,t} 
&= \frac{1}{np}\sum_{r=1}^n\bX_{r\cdot s}^\top\bfOmega\bX_{r\cdot t} = \frac{1}{np}\sum_{r=1}^n(\bfOmega^{1/2}\bM_{s}\by_r)^\top(\bfOmega^{1/2}\bM_{t}\by_r)\\ 
&=\frac{1}{np}\sum_{r=1}^n\ba_r^\top\bb_r = \frac{1}{4np}\sum_{r=1}^n\Big(\|\ba_r+\bb_r\|_2^2 - \|\ba_r-\bb_r\|_2^2\Big)\\
&:= \frac{1}{4np}\Big(\|\bA+\bB\|_2^2 - \|\bA-\bB\|_2^2\Big),
\end{split}
\end{equation}
where $\bA = [\ba^\top_1, \ldots, \ba^\top_n]^\top$ and $\bB = [\bb^\top_1, \ldots, \bb^\top_n]^\top$ are length-$np$ vectors. If we use $\bI_n$ to denote $n \times n$ identity matrix, and $\bY = [\by^\top_1, \ldots, \by^\top_n]^\top$ to denote the concatenated i.i.d standard Gaussian vector, A and B can be rewritten as 
\[
\bA = (\bI_n\otimes\bfOmega^{1/2}\bM_{s})\bY := \bU_s\bY,\ \text{ and } \  \bB = (\bI_n\otimes\bfOmega^{1/2}\bM_{t})\bY := \bU_t\bY.
\]
Since $\bY$ is a length-$2np$ i.i.d standard Gaussian random vector, the distributions of $\bA+\bB$ and $\bA-\bB$ are:
\begin{equation}\label{eq-lemma1-2}
\begin{split}
\bA+\bB &\sim \mathcal{N}\big(0, (\bU_s + \bU_t)(\bU_s + \bU_t)^\top\big):= \mathcal{N}\big(0, \bQ_s\big)\\ 
\bA - \bB &\sim \mathcal{N}\big(0, (\bU_s - \bU_t)(\bU_s - \bU_t)^\top\big):= \mathcal{N}\big(0, \bQ_t\big).
\end{split}
\end{equation}
Before using a concentration inequality on $\|\bA+\bB\|_2$ and $\|\bA-\bB\|_2$, we still need upper-bounds for the spectral norm of $\bQ_s$ and $\bQ_t$, which can be easily obtained by matrix norm inequalities. From the definition:
\begin{equation}\label{eq-lemma1-3-1}
\max\{\|\bQ_s\|_2, \|\bQ_t\|_2\}\leq \|\bU_s\bU^\top_s\|_2+2\|\bU_s\bU^\top_t\|_2+\|\bU_t\bU^\top_t\|_2\leq 4\max_{s,t}\|\bU_s\bU^\top_t\|_2
\end{equation}
and 
\begin{equation}\label{eq-lemma1-3-2}
\begin{split}
\|\bU_s\bU^\top_t\|_2 &= \|\bI_n\otimes(\bfOmega^{1/2}\bM_s\bM_t^\top\bfOmega^{\top/2})\|_2 = \|\bI_n\otimes(\bfOmega^{1/2}\bfSigma_{(s,t)}\bfOmega^{\top/2})\|_2\\
&\leq \|\bI_n\|_2\|\bfOmega^{1/2}\bfSigma_{(s,t)}\bfOmega^{\top/2}\|_2 = \|\bfOmega^{1/2}\bfSigma_{(s,t)}\bfOmega^{\top/2}\|_2
\end{split}
\end{equation}
Additionally, the decomposition in \eqref{eq-lemma1-1} is invariant under expectation. If we generally define $\ba = \bfOmega^{1/2}\bM_{s}\by$ and $\bb = \bfOmega^{1/2}\bM_{s}\by$ for standard Gaussian vector $\by$, we have
\begin{equation}\label{eq-lemma1-4}
\begin{split}
\frac{1}{p}\mathbb{E}\{\bX^\top\bfOmega\bX\}_{s,t} &= \frac{1}{4p}\mathbb{E}\big(\|\ba+\bb\|_2^2 - \|\ba-\bb\|_2^2\big) = \frac{1}{4np}\sum_{j=1}^n\mathbb{E}\big(\|\ba_j+\bb_j\|_2^2 - \|\ba_j-\bb_j\|_2^2\big) \\
& = \frac{1}{4np}\Big(\mathbb{E}\{\|\bA+\bB\|_2^2\} - \mathbb{E}\{\|\bA-\bB\|_2^2\}\Big),
\end{split}
\end{equation}
This suggests a tail probability bound for any $\delta>0$:
\[
\begin{split}
\mathbb{P}\left\{\Bigg|\bigg(\frac{1}{np}\sum_{j=1}^n\bX_{j}^\top\bfOmega\bX_{j} - \frac{1}{p}\mathbb{E}\{\bX^\top\bfOmega\bX\}\bigg)_{s,t}\Bigg| \geq \delta\right\} &\leq \\
\mathbb{P}\left\{\Bigg|\frac{1}{np}\Big(\|\bA+\bB\|_2^2 - \mathbb{E}\{\|\bA+\bB\|_2^2\}\Big)\Bigg| \geq 4\delta\right\} &+ \mathbb{P}\left\{\Bigg|\frac{1}{np}\Big(\|\bA-\bB\|_2^2 - \mathbb{E}\{\|\bA-\bB\|_2^2\}\Big)\Bigg| \geq 4\delta\right\}
\end{split}
\]
Combining the definition of $\bQ_s$ and $\bQ_t$ in Equation \eqref{eq-lemma1-2} with Lemma S.12 from \citet{lyu2019tensor} and Lemma I.2 from \citet{negahban2011estimation}, %
the two terms can be bounded by
\begin{equation}
\begin{split}
&\mathbb{P}\left\{\Bigg|\frac{1}{np}\Big(\|\bA+\bB\|_2^2 - \mathbb{E}\{\|\bA+\bB\|_2^2\}\Big)\Bigg| \geq 4\delta\right\} \leq
2\exp\left\{-\frac{np}{2}\bigg(\frac{\delta}{\|\bQ_s\|_2} - \frac{2}{\sqrt{np}}\bigg)^2\right\} + 2\exp\Big(-\frac{np}{2}\Big);\\
&\mathbb{P}\left\{\Bigg|\frac{1}{np}\Big(\|\bA-\bB\|_2^2 - \mathbb{E}\{\|\bA-\bB\|_2^2\}\Big)\Bigg| \geq 4\delta\right\} \leq
2\exp\left\{-\frac{np}{2}\bigg(\frac{\delta}{\|\bQ_t\|_2} - \frac{2}{\sqrt{np}}\bigg)^2\right\} + 2\exp\Big(-\frac{np}{2}\Big).
\end{split}
\end{equation}
Combining the inequalities \eqref{eq-lemma1-3-1} and \eqref{eq-lemma1-3-2} with a maximal sum inequality, i.e. for generic random variables $x_1, \ldots, x_n$, 
\[
\mathbb{P}(\max_{r}x_r>\delta) \leq n\max_r\mathbb{P}(x_r \geq \delta),
\] 
we obtain
\begin{equation}
\begin{split}
&\mathbb{P}\left\{\max\limits_{s,t}\bigg(\frac{1}{np}\sum_{j=1}^n\bX_{j}^\top\bfOmega\bX_{j} - \frac{1}{p}\mathbb{E}\{\bX^\top\bfOmega\bX\}\bigg)_{s,t} \geq \delta\right\} \leq \\
&4q^2\exp\left\{-\frac{np}{2}\bigg(\frac{\delta}{4\max\limits_{s,t}\|\bfOmega^{1/2}\bfSigma_{(s,t)}\bfOmega^{\top/2}\|_2} - \frac{2}{\sqrt{np}}\bigg)^2\right\} + 4q^2\exp\Big(-\frac{np}{2}\Big).
\end{split}
\end{equation}
As $\log(q) = o(np)$, the second term converges to $0$. For a sufficiently large $C$, when $\delta \geq C\max\limits_{s,t}\|\bE_{(s,t)}\|_2\sqrt{\frac{\log q}{np}}$, the above probability $\sim o(1)$, and the proof can be concluded.
\end{proof}

\begin{proof}[Proof of Theorem \ref{thm-rate}]
For the simplicity of notation, we use the notation $\bfOmega_k = \bL_k\bL_k^\top$ for precision matrix. Consider the objective function
\[
L(\bL_{k}) = \tr\big(\bL_{k}^\top\bL_{k}\bar{\bfSigma}_{k}\big) - \log\det(\bL_{k}\bL_{k})
\]
Recall that for the individual minimization problem, to bound the difference between $\hat{\bL}_{k}$ and $\tilde{\bL}_{k}$, we only need to show that the minimum of $L(\bL_{k})$ can not be achieved outside a neighborhood of $\tilde{\bL}_{k}$. Then the minimizer $\hat{\bL}_{k}$ must fall in this neighborhood.

We now introduce $\bfDelta = \bL_{k}\bL^\top_{k} - \tilde{\bL}_{k}\tilde{\bL}_{k}^\top$ for simplicity and describe the neighborhood by introducing a boundary set
\begin{equation}
    \mathbb{B}(\delta):= \{\bL_{k} \in \mathcal{S}, \|\bL_{k}\bL^\top_{k} - \tilde{\bL}_{k}\tilde{\bL}_{k}^\top\|_F := \|\bfDelta\|_F := \delta\}
\end{equation}
Note that the function $L(\bL_k) - L(\tilde{\bL}_k)$ is convex  in $\Delta$. If for a certain $\delta$,
\begin{equation}\label{eq-rate-0}
\inf_{\bL_{k} \in \mathbb{B}(\delta)}\{L(\bL_{k}) - L(\tilde{\bL}_{k})\} > 0,
\end{equation}
then $\hat{\bL}_{k}$ must fall in the  interior
of $\mathbb{B}(\delta)$ because $\hat{\bL}_{k}$ minimizes $L$ and  $L(\hat{\bL}_{k}) - L(\tilde{\bL}_{k}) \leq 0$. 

The problem reduces to finding the minimum $\delta$ such that \eqref{eq-rate-0} holds. We parametrize the path from $\tilde{\bL}_{k}\tilde{\bL}_{k}^\top$ to $\hat{\bL}_{k}\hat{\bL}^\top_{k}$ with $\bfOmega_{k}(t) = \tilde{\bL}_{k}\tilde{\bL}_{k}^\top + t\bfDelta$:
\begin{align}\label{eq-rate-1}
L(\bL_{k}) - L(\tilde{\bL}_{k}) &= \tr(\bfDelta\bar{\bfSigma}_{k}) - \left[\log\det(\tilde{\bL}_{k}\tilde{\bL}_{k}^\top + \bfDelta) - \log\det(\tilde{\bL}_{k}\tilde{\bL}_{k}^\top)\right]\\
&:= \tr(\bfDelta\bar{\bfSigma}_{k}) - \Big[\log\det\big(\bfOmega_{k}(1)\big) - \log\det\big(\bfOmega_{k}(0)\big)\Big]\\
&= \tr\big(\bfDelta\bar{\bfSigma}_{k}\big) - \tr\big(\bfDelta(\tilde{\bL}_{k}\tilde{\bL}_{k}^\top)^{-1}\big) + vec^\top(\bfDelta)\left[\int_{0}^1(1-t)\bfOmega_{k}^{-1}(t)\otimes \bfOmega_{k}^{-1}(t)dt\right]vec(\bfDelta)
\end{align}
The third equation uses Taylor's expansion of $\bfOmega_{k}(t)$ at $t=0$. The first two terms in \eqref{eq-rate-1} can be bounded by $\|\bfDelta\|_F$:
\begin{align}\label{eq-rate-term1-1}
\Big|tr\big(\bfDelta\bar{\bfSigma}_{k}\big) - \tr\big(\bfDelta(\tilde{\bL}_{k}\tilde{\bL}_{k}^\top)^{-1}\big)\Big| &\leq \|\bar{\bfSigma}_{k} - (\tilde{\bL}_{k}\tilde{\bL}_{k}^\top)^{-1}\|_{\max}\sum_{i,j}|\bfDelta_{i,j}|\\
&\leq \|\bar{\bfSigma}_{k} - (\tilde{\bL}_{k}\tilde{\bL}_{k}^\top)^{-1}\|_{\max}|S(\bfDelta)|^{1/2}\|\bfDelta\|_F
\end{align}
We decompose the first two terms in \eqref{eq-rate-1} into the estimation error $\bar{\bfSigma}_{k} - \tilde{\bfSigma}_{k}$ and the approximation error $\tilde{\bfSigma}_{k}- (\tilde{\bL}_{k}\tilde{\bL}_{k}^\top)^{-1}$. According to %
Lemma \ref{lemma-asym-1}, 
the estimation error can be upper-bounded:
\[
\|\bar{\bfSigma}_{k} - \tilde{\bfSigma}_{k}\|_{\max} \leq C_1\max\limits_{i^{(k)},j^{(k)}}\|\bE_{i^{(k)},j^{(k)}}\|_2\sqrt{\frac{\log p_{k}}{np_{-k}}}.
\]
and
\begin{align}\label{eq-rate-term1-2}
\|\bar{\bfSigma}_{k} - (\tilde{\bL}_{k}\tilde{\bL}_{k}^\top)^{-1}\|_{\max} 
&\leq \|\bar{\bfSigma}_{k} - \tilde{\bfSigma}_{k}\|_{\max} + \|\tilde{\bfSigma}_{k}- (\tilde{\bL}_{k}\tilde{\bL}_{k}^\top)^{-1}\|_{\max}\\
&\leq C_1\max\limits_{i^{(k)},j^{(k)}}\|\bE_{i^{(k)},j^{(k)}}\|_2\sqrt{\frac{\log p_{k}}{np_{-k}}} + \|\tilde{\bfSigma}_{k}- (\tilde{\bL}_{k}\tilde{\bL}_{k}^\top)^{-1}\|_{\max}
\end{align}
For the third term in \eqref{eq-rate-1}, it's always positive and can be lower-bounded, 
\begin{align}\label{eq-rate-term3}
vec^\top(\bfDelta)\left[\int_{0}^1(1-t)\bfOmega_{k}^{-1}(t)\otimes \bfOmega_{k}^{-1}(t)dt\right]vec(\bfDelta) 
&\geq 
\|vec(\bfDelta)\|_2^2\int_{0}^1(1-t)\lambda_{\min}\big(\bfOmega_{k}^{-1}(t)\otimes \bfOmega_{k}^{-1}(t)\big)dt\\
&\geq
\frac{1}{2}\|\bfDelta\|_F^2\min_{t\in [0,1]}\lambda^{-2}_{\max}\big(\bfOmega_{k}(t)\big)\\
&\geq \frac{1}{2}\|\bfDelta\|_F^2\lambda^{-2}_{\max}\big(\bfOmega_{k}(1)\big) \\
&\geq \frac{1}{2}\|\bfDelta\|_F^2\big[\lambda_{\max}(\tilde{\bL}_{k}\tilde{\bL}_{k}^\top) + \lambda_{\max}(\bfDelta)\big]^{-2}
\end{align}

We now claim that there exists a constant $C_2$ such that when 
\[
\delta(C_2) = {C_2}\left(C_1\max\limits_{i^{(k)},j^{(k)}}\|\bE_{i^{(k)},j^{(k)}}\|_2\sqrt{\frac{\log p_{k}}{np_{-k}}} + \|\tilde{\bfSigma}_{k}- (\tilde{\bL}_{k}\tilde{\bL}_{k}^\top)^{-1}\|_{\max}\right)|S(\bfDelta)|^{1/2},
\]
$L(\bL_{k}) - L(\tilde{\bL}_{k}) > 0$ always holds on $\mathbb{B}(\delta(C_2))$.

For $\bfDelta \in \mathbb{B}(\delta(C_2))$ defined above. Firstly, $\lambda_{\max}(\bfDelta) \leq \|\bfDelta\|_F = \delta(C_2)$, and 
\[
\lim_{n \to \infty}\lambda_{\max}(\bfDelta) \leq \lim_{n \to\infty}\delta(C_2) = C_2\lim_{n \to \infty}\|\tilde{\bfSigma}_{k}- (\tilde{\bL}_{k}\tilde{\bL}_{k}^\top)^{-1}\|_{\max}|S(\bfDelta)|^{1/2} := C_2\delta_{\infty}.
\]
Then the first two terms in \eqref{eq-rate-1} has an upper bound:
\begin{equation}\label{eq-rate-term1-res}
\begin{split}
\Big|tr\big(\bfDelta\bar{\bfSigma}_{k}\big) - \tr\big(\bfDelta(\tilde{\bL}_{k}\tilde{\bL}_{k}^\top)^{-1}\big)\Big|  &\leq
\lim_{n\to\infty} \|\bar{\bfSigma}_{k} - (\tilde{\bL}_{k}\tilde{\bL}_{k}^\top)^{-1}\|_{\max}|S(\bfDelta)|^{1/2}\|\bfDelta\|_F \\
&\leq \lim_{n\to\infty}\Big(C_1\max\limits_{i^{(k)},j^{(k)}}\|\bE_{i^{(k)},j^{(k)}}\|_2\sqrt{\frac{\log p_{k}}{np_{-k}}} + \|\tilde{\bfSigma}_{k}- (\tilde{\bL}_{k}\tilde{\bL}_{k}^\top)^{-1}\|_{\max}\Big)|S(\bfDelta)|^{1/2}\|\bfDelta\|_F \\ 
& \leq C_2\delta_{\infty} \lim_{n\to\infty}\|\tilde{\bfSigma}_{k}- (\tilde{\bL}_{k}\tilde{\bL}_{k}^\top)^{-1}\|_{\max}|S(\bfDelta)|^{1/2}.
\end{split}
\end{equation}
Meanwhile, by lemma \ref{lemma-asym-eigen-max}, $\lambda_{\max}(\tilde{\bL}_{k}\tilde{\bL}_{k}^\top) \leq p_k/\lambda_{min}(\bfSigma)\tr(\bL_{-k}\bL^\top_{-k})$. 
Consequently, as $n\to\infty$, the limit of the third term in \eqref{eq-rate-1} can be lower-bounded with:
\begin{equation}\label{eq-rate-term3-res}
\begin{split}
    \lim_{n\to\infty}vec^\top(\bfDelta)\left[\int_{0}^1(1-t)\bfOmega_{k}^{-1}(t)\otimes \bfOmega_{k}^{-1}(t)dt\right]vec(\bfDelta) &\geq \lim_{n\to\infty}\frac{1}{2}\|\bfDelta\|_F^2\big[\lambda_{\max}(\tilde{\bL}_{k}\tilde{\bL}_{k}^\top) + \lambda_{\max}(\bfDelta)\big]^{-2}\\
    &\geq \frac{C_2^2\delta_{\infty}^2}{2(C_2\delta_{\infty}+p_k/\lambda_{min}(\bfSigma)\tr(\bL_{-k}\bL^\top_{-k}))^2}
\end{split}
\end{equation}
To ensure the existence of $C_2 >0$ such that the third term \eqref{eq-rate-term3-res} dominates \eqref{eq-rate-term1-res} and $L(\bL_{k}) - L(\tilde{\bL}_{k}) > 0$ always holds in \eqref{eq-rate-1}, we need
\[
\lim_{n\to\infty}\|\tilde{\bfSigma}_{k}- (\tilde{\bL}_{k}\tilde{\bL}_{k}^\top)^{-1}\|_{\max}|S(\bfDelta)|^{1/2} \leq \frac{\lambda_{min}(\bfSigma)\tr(\bL_{-k}\bL^\top_{-k})}{8p_{k}}.
\]
This requirement is met by the condition in the theorem.

In conclusion, we find $\delta(C_2) > 0$ such that \eqref{eq-rate-0} holds on $\mathbb{B}(\delta(C_2))$ for large enough $n$. As a result, $\hat{\bL}_{k}$ must fall in the  interior
of $\mathbb{B}(\delta(C_2))$, i.e, 
\[
\|\hat{\bL}_{k}\hat{\bL}^\top_{k} - \tilde{\bL}_{k}\tilde{\bL}^\top_{k}\|_F = \|\bfDelta\|_F \leq \delta(C_2).
\]
In other words,
\begin{equation}
    \begin{split}
            \|\hat{\bL}_{k}\hat{\bL}^\top_{k} - \tilde{\bL}_{k}\tilde{\bL}^\top_{k}\|_F =    \mathcal{O}_p\left( \left( C_1\max\limits_{i^{(k)},j^{(k)}}\|\bE_{(i^{(k)},j^{(k)})}\|_2\sqrt{\frac{\log p_{k}}{np_{-k}}} + 
    \|\tilde{\bfSigma}_{k}- (\tilde{\bL}_{k}\tilde{\bL}_{k}^\top)^{-1}\|_{\max} \right) \big|\cS_{k}\big|^{1/2} \right).
    \end{split}
\end{equation}
\end{proof}

\begin{lemma}[Max-entry error of the SIC-induced covariance]
\label{lem:sic-max}
Under the conditions of Lemma~\ref{lem:sic-shared}, if
\[
    \rho\ge \frac{1}{c}\log(p/\epsilon),
\]
then
\[
    \left\|
        \bfSigma-(\bL_\rho\bL_\rho^\top)^{-1}
    \right\|_{\max}
    \le
    \epsilon.
\]
In particular, choosing $\epsilon$ proportional to $p \exp(- c\rho)$ gives
\[
    \left\|
        \bfSigma-(\bL_\rho\bL_\rho^\top)^{-1}
    \right\|_{\max}
    \le
    C p\exp(-c\rho).
\]
If the sparsity pattern $\cS_\rho$ has at most $m$ nonzero entries per column and $m\asymp \rho^d$, then
\[
    \left\|
        \bfSigma-(\bL_\rho\bL_\rho^\top)^{-1}
    \right\|_{\max}
    \le
    C p\exp(-c m^{1/d}).
\]
\end{lemma}

\begin{proof}
By Lemma~\ref{lem:sic-shared}, $\left\|\bfSigma-(\bL_\rho\bL_\rho^\top)^{-1}\right\|_F\le\epsilon$.
Since $\|\bA\|_{\max}\le \|\bA\|_F$
for any matrix $\bA$, it follows immediately that $\left\|\bfSigma-(\bL_\rho\bL_\rho^\top)^{-1}\right\|_{\max}\le\epsilon$.
The exponential form follows by taking
\[
    \epsilon=Cp\exp(-c\rho),
\]
as implied by Lemma~\ref{lem:sic-shared}. Finally, the reverse-maximin sparsity
pattern satisfies $m\asymp \rho^d$, giving
\[
    \left\|
        \bfSigma-(\bL_\rho\bL_\rho^\top)^{-1}
    \right\|_{\max}
    \le
    Cp\exp(-c m^{1/d}).
\]
\end{proof}

\begin{lemma}[Frobenius precision error of SIC]
\label{lem:sic-fro}
Under the conditions of Lemma~\ref{lem:sic-shared}, if the sparsity pattern $\cS_\rho$ has at most $m$ nonzero entries per column such that $m\asymp \rho^d$ and $p\exp(-cm^{1/d}) = o(1)$, then the Frobenius loss can be bounded by
\[
    \left\|
        \bfSigma^{-1}-\bL_\rho\bL_\rho^\top
    \right\|_F
    \le
    C \frac{p}{\lambda_{\min}(\bfSigma)}\exp(-c m^{1/d}).
\]
\end{lemma}

\begin{proof}
According to Theorem 3.4 of \citet{Schafer2020}, under the conditions of Lemma~\ref{lem:sic-shared}, the KL divergence between $\bfSigma$ and $(\bL_{\rho}\bL_{\rho}^\top)^{-1}$ share the same upper bound with $\|\bfSigma-(\bL_\rho\bL_\rho^\top)^{-1}\|$, i.e.
\[
\mathrm{KL}\left( \mathcal{N}(\bfzero, \bfSigma) \;\|\; \mathcal{N}(\bfzero, (\bL\bL^\top)^{-1}) \right) \leq Cp\exp(-c\rho).
\]
Since $p\exp(-cm^{1/d}) = o(1)$, according to  Lemma B.8 of \citet{Schafer2020}, we have
\[
\|\bfSigma^{-1}-\bL_\rho\bL_\rho^\top\|_F
\leq \frac{1}{\lambda_{\min}(\bfSigma)}\mathrm{KL}\left( \mathcal{N}(\bfzero, \bfSigma) \;\|\; \mathcal{N}(\bfzero, (\bL_{\rho}\bL_{\rho}^\top)^{-1}) \right) \leq C \frac{p}{\lambda_{\min}(\bfSigma)}\exp(-c m^{1/d}).
\]
\end{proof}

\begin{lemma}[Endpoint eigenvalue bounds and condition number]
\label{lem:sic-cond}
Assume the conditions of Lemma~\ref{lem:sic-shared}. 
Then the boundary-conditioned Mat\'ern Green's-function covariance matrix $\bfSigma$
satisfies
\[
    \lambda_{\max}(\bfSigma)
    \le
    C\Lambda_{\max} p,
    \qquad
    \lambda_{\min}(\bfSigma)
    \ge
    \Lambda_{\min} p^{-2\nu/d},
\]
where $\Lambda_{\max}$ and $\Lambda_{\min}$ are positive constants independent of $p$. Consequently,
\[
    \operatorname{cond}(\bfSigma)
    =
    \frac{\lambda_{\max}(\bfSigma)}
         {\lambda_{\min}(\bfSigma)}
    =
    \mathcal{O}(p^{1+2\nu/d}).
\]
\end{lemma}

\begin{proof}
We first bound the largest eigenvalue. Since $\bfSigma$ is positive semidefinite,
\begin{equation}\label{eq-proof-green-1}
    \lambda_{\max}(\bfSigma)
    \le
    \operatorname{tr}(\bfSigma)
    =
    \sum_{i=1}^p G(s_i,s_i).
\end{equation}
For a boundary-conditioned Mat\'ern Green's function with $\alpha=\nu+d/2>d/2$, point
evaluation is continuous and the diagonal is uniformly bounded:
\[
    \sup_{x\in\Omega}G(x,x)\le B_G<\infty .
\]
Therefore, let $\Lambda_{\max} = B_G$
\[
    \lambda_{\max}(\bfSigma)
    \le
    B_G p = \Lambda_{\max}p.
\]

Next, we lower-bound the smallest eigenvalue. 
The homogeneity condition in \citet[Theorem~3.4]{Schafer2020} implies the boundary-aware pseudo-uniform bounds for the separation radius $q_p$ among the locations $s_i$'s.
\[
   q_p := \min_{i\neq j}\|s_i - s_j\| \asymp p^{-1/d}.
\]
The Green's-function covariance has native space equivalent to the Sobolev space $H_0^\alpha(\Omega)$ associated with the
elliptic operator. Standard stability estimates for kernel interpolation matrices with Sobolev smoothness $\alpha$ imply $\lambda_{\min}(\bfSigma) \ge \Lambda_{\min} q_p^{2\alpha-d}$;
see, for example, \citet[Ch.~12]{wendland2004scattered} and the eigenvalue stability estimates
in \citet{schaback1995error, diederichs2019improved}. Since $\alpha=\nu+d/2$,
\[
    \lambda_{\min}(\bfSigma)
    \ge
    \Lambda_{\min} q_p^{2\nu}.
\]
By pseudo-uniformity $q_p\asymp p^{-1/d}$,
and hence
\[
    \lambda_{\min}(\bfSigma)
    \ge
    \Lambda_{\min} p^{-2\nu/d}.
\]
Combining the two endpoint bounds gives
\[
    \operatorname{cond}(\bfSigma)
    =
    \frac{\lambda_{\max}(\bfSigma)}
         {\lambda_{\min}(\bfSigma)}
    \le
    \frac{\Lambda_{\max}p}{\Lambda_{\min}p^{-2\nu/d}}
    =
    \mathcal{O}(p^{1+2\nu/d}) = \mathcal{O}(p^{2\alpha/d}).
\]
\end{proof}

\begin{lemma}[Lower bound for the Frobenius norm of the precision matrix]
\label{lem:sic-eigen}
Assume the conditions of Lemma~\ref{lem:sic-shared}. 
Then,
\[
    \|\bfSigma^{-1}\|_F=\Omega(p^{1/2}).
\]
\end{lemma}
\begin{proof}
Let $\lambda_1,\ldots,\lambda_p>0$ be the eigenvalues of $\bfSigma$. Since $\bfSigma$ is positive definite, $\|\bfSigma^{-1}\|_F^2 = \sum_{j=1}^p \lambda_j^{-2}.$
By Cauchy's inequality,
\[
    \left(\sum_{j=1}^p \lambda_j\right)
    \left(\sum_{j=1}^p \lambda_j^{-1}\right)
    \ge
    p^2.
\]
Therefore,
\[
    \operatorname{tr}(\bfSigma^{-1})
    =
    \sum_{j=1}^p \lambda_j^{-1}
    \ge
    \frac{p^2}{\operatorname{tr}(\bfSigma)} 
    \geq B_G^{-1}p,
\]
where the last inequality is due to Equation \eqref{eq-proof-green-1} and the corresponding bound in the proof of Lemma \ref{lem:sic-cond}.

Since $\bfSigma^{-1}$ is positive definite,
\[
    \|\bfSigma^{-1}\|_F
    =
    \left(\sum_{j=1}^p\lambda_j^{-2}\right)^{1/2}
    \ge
    p^{-1/2}\sum_{j=1}^p\lambda_j^{-1}
    =
    p^{-1/2}\operatorname{tr}(\bfSigma^{-1}).
\]
Combining the preceding inequalities gives
\[
    \|\bfSigma^{-1}\|_F
    \ge
    B_G^{-1}p^{1/2}.
\]
\end{proof}

\begin{proof}[Proof of Theorem \ref{col-main}]
When $\bfSigma$ is separable, 
\[
\tr(\bigotimes_{k = 1}^K\bL_k\bL_k^\top\bfSigma) = \prod_{k = 1}^K\tr(\bL_k^\top\bfSigma_k\bL_k).
\]
For a fixed $k$, the joint KL divergence can be simplified to 
\[
\prod_{l\neq k}\tr(\bL_l^\top\bfSigma_l\bL_l)\cdot \tr(\bL_k^\top\bfSigma_k\bL_k) + p_{-k}\log\det(\bL_k\bL_k^\top) + C,
\]
which implies that the pseudo covariance is proportional to the true covariance: 
\[\tilde{\bfSigma}_{k} = \frac{\prod_{l\neq k}tr(\bL^\top_l\bfSigma_l\bL_l)}{p_{-k}}\bfSigma_{k}.
\] 

The total error $\|\hat{\bL}_{k}\hat{\bL}_{k}^\top - \tilde{\bfSigma}^{-1}_k\|_F$ can be decomposed as
\[
\|\hat{\bL}_{k}\hat{\bL}_{k}^\top - \tilde{\bfSigma}_k^{-1}\|_F \leq \underbrace{\|\hat{\bL}_{k}\hat{\bL}_{k}^\top - \tilde{\bL}_{k}\tilde{\bL}_{k}^\top\|_F}_{\text{Term 1}} + \underbrace{\|\tilde{\bL}_{k}\tilde{\bL}_{k}^\top - \tilde{\bfSigma}_k^{-1}\|_F}_{\text{Term 2}}.
\]
Term 1 can be bounded using Theorem \ref{thm-rate}. Under a separable covariance assumption, within the convergence rate \eqref{eq:main-rate}, the sub-covariance $\bfSigma^{(-k)}_{(i^{(k)},j^{(k)})}$ reduces to $\bfSigma_{k,i^{(k)},j^{(k)}}\cdot \bL^\top_{-k}\bfSigma_{-k}\bL_{-k}$. To use Theorem \ref{thm-rate}, we only need to check that Assumption \ref{asmp-rate-1} holds for large enough $n$. Specifically, the left hand side within Assumption \ref{asmp-rate-1} can be bounded using Lemma \ref{lem:sic-max}:
\begin{equation}
\|\tilde{\bfSigma}_{k}-(\tilde{\bL}_{k}\tilde{\bL}_{k}^\top)^{-1}\|_{\max}|S_k|^{1/2}
\leq
Cm_kp^{3/2}_k\exp(-cm_k^{1/d}).
\end{equation}
On the other hand, the left hand side within Assumption \ref{asmp-rate-1} can be bounded using Lemma \ref{lem:sic-cond}:
\begin{equation}
\begin{split}
\frac{\lambda_{min}(\bfSigma)\tr(\bL_{-k}\bL^\top_{-k})}{8p_{k}} &= \frac{\lambda_{min}(\bfSigma_k)\prod_{l\neq k}\lambda_{min}(\bfSigma_l)\tr(\bL_{-k}\bL^\top_{-k})}{8p_{k}}\\
&\geq \frac{\Lambda_{\min}\prod_{l\neq k}\lambda_{min}(\bfSigma_l)}{8}p_k^{-1-2\nu/d}\tr(\bL_{-k}\bL^\top_{-k}).
\end{split}    
\end{equation}
Since $\tr(\bL_{-k}\bL^\top_{-k}) \asymp p_{-k}$ and the eigenvalues of $\bfSigma_l$ have uniform lower bound, Assumption \ref{asmp-rate-1} holds as long as
\begin{equation}
m_kp^{3/2}_k\exp(-cm_k^{1/d}) = o\left(p_k^{-1-2\nu/d}\right),
\end{equation}
which is guaranteed by Assumption \ref{asmp:sic-3} where $m_k = \omega(\log^{d} p_k)$.
Meanwhile, the other condition in Assumption \ref{asmp:sic-3} ensures that step \eqref{eq-rate-term1-res} holds in the proof of Theorem \ref{thm-rate}. As a result, we plug $|\cS_k| = m^*_k$ and $p_k$ into the main convergence rate \eqref{eq:main-rate} and get:
\begin{equation}\label{eq:cor-main-proof-1-1}
\|\hat{\bL}_{k}\hat{\bL}_{k}^\top - \tilde{\bfSigma}_k^{-1}\|_F = \mathcal{O}_p\left(C_1\|\bL^\top_{-k}\bfSigma_{-k}\bL_{-k}\|_2\sqrt{\frac{m^{*2}_kp_{k}\log p_{k}}{np_{-k}}} + m^*_kp^{3/2}_k\exp(-cm_k^{1/d})\right).
\end{equation}
Term 2, by Lemma \ref{lem:sic-fro} and \ref{lem:sic-cond}, can be bounded with
\begin{equation}\label{eq:cor-main-proof-1-2}
\begin{split}
   \|\tilde{\bL}_{k}\tilde{\bL}_{k}^\top - \tilde{\bfSigma}_k^{-1}\|_F &\leq \frac{Cp_{-k}}{\Lambda_{\min}\prod_{l\neq k}tr(\bL^\top_l\bfSigma_l\bL_l)}p_k^{1+2\nu/d}\exp(-cm^{1/d})\\
   &= \mathcal{O}(p_k^{1+2\nu/d}\exp(-cm_k^{1/d})). 
\end{split}
\end{equation}
The last equation is due to the assumption that for $l\neq k$, $\tr(\bL_l^\top\bL_l)\asymp p_l$ and $\lambda_{\min}(\bfSigma_l)$ has a uniform lower bound.
Combining the rates for term 1 \eqref{eq:cor-main-proof-1-1} and term 2 \eqref{eq:cor-main-proof-1-2}, the overall error rate can be written as:
\begin{equation}\label{eq:cor-main-proof-2}
\|\hat{\bL}_{k}\hat{\bL}_{k}^\top - \tilde{\bL}_{k}\tilde{\bL}_{k}^\top\|_F = \mathcal{O}_p\left(C_1\|\bL^\top_{-k}\bfSigma_{-k}\bL_{-k}\|_2\sqrt{\frac{m^{*2}_kp_{k}\log p_{k}}{np_{-k}}} + (p_k^{1+2\nu/d} + m^*_kp_k^{3/2})\exp(-cm_k^{1/d})\right).
\end{equation}
Consider the normalized estimation:
\begin{equation}\label{eq:cor-main-proof-3}
\begin{split}
\left\|\frac{\hat{\bL}_{k}\hat{\bL}_{k}^\top}{\|\hat{\bL}_{k}\hat{\bL}_{k}^\top\|_F} - \frac{\tilde{\bfSigma}_k^{-1}}{\|\tilde{\bfSigma}_k^{-1}\|_F}\right\|_F &\leq \left\|\frac{\hat{\bL}_{k}\hat{\bL}_{k}^\top}{\|\hat{\bL}_{k}\hat{\bL}_{k}^\top\|_F} - \frac{\tilde{\bfSigma}_k^{-1}}{\|\hat{\bL}_{k}\hat{\bL}_{k}^\top\|_F}\right\|_F + \left\|\frac{\tilde{\bfSigma}_k^{-1}}{\|\hat{\bL}_{k}\hat{\bL}_{k}^\top\|_F} - \frac{\tilde{\bfSigma}_k^{-1}}{\|\tilde{\bfSigma}_k^{-1}\|_F}\right\|_F\\
&= \frac{1}{\|\hat{\bL}_{k}\hat{\bL}_{k}^\top\|_F}\|\hat{\bL}_{k}\hat{\bL}_{k}^\top - \tilde{\bfSigma}_k^{-1}\|_F + \frac{1}{\|\hat{\bL}_{k}\hat{\bL}_{k}^\top\|_F}\big|\|\tilde{\bfSigma}_k^{-1}\|_F-\|\hat{\bL}_{k}\hat{\bL}_{k}^\top\|_F\big|\\
&\leq \frac{2}{\|\hat{\bL}_{k}\hat{\bL}_{k}^\top\|_F}\|\hat{\bL}_{k}\hat{\bL}_{k}^\top - \tilde{\bfSigma}_k^{-1}\|_F.
\end{split}
\end{equation}
When $n \to \infty$, by Assumption \ref{asmp:sic-3}, the above result says $\|\hat{\bL}_{k}\hat{\bL}_{k}^\top - \tilde{\bfSigma}_k^{-1}\|_F \to 0$, so that we have $\|\hat{\bL}_{k}\hat{\bL}_{k}^\top\|_F \geq \|\tilde{\bfSigma}_k^{-1}\|_F/2$ for large enough $n$. By Lemma \ref{lem:sic-eigen}, $\|\tilde{\bfSigma}_k^{-1}\|_F = \Omega(p_k^{1/2})$. Plugging \eqref{eq:cor-main-proof-2} into \eqref{eq:cor-main-proof-3}, we get
\begin{equation}
\left\|\frac{\hat{\bL}_{k}\hat{\bL}_{k}^\top}{\|\hat{\bL}_{k}\hat{\bL}_{k}^\top\|_F} - \frac{\tilde{\bfSigma}_k^{-1}}{\|\tilde{\bfSigma}_k^{-1}\|_F}\right\|_F = 
\mathcal{O}_p\left(C_1\|\bL^\top_{-k}\bfSigma_{-k}\bL_{-k}\|_2\sqrt{\frac{m^{*2}_k\log p_{k}}{np_{-k}}} + (p_k^{1/2+2\nu/d} + m^*_kp_k)\exp(-cm_k^{1/d})\right).
\end{equation}
Since $\bfOmega_k$ is proportional to $\tilde{\bfSigma}_k^{-1}$, we get the rate in Theorem \ref{col-main}.
\end{proof}

\subsection{Proof of Theorem \ref{thm:parametric_stability}}
\begin{proof}[Proof of Theorem \ref{thm:parametric_stability}]
For each $\bftheta\in\Theta$, define
$$
\bB_{\bftheta}
=
\bfSigma_{\bftheta}^{1/2}
\hat\bL_{\bftheta}\hat\bL_{\bftheta}^{\top}
\bfSigma_{\bftheta}^{1/2}.
$$
Both matrices are positive semidefinite, and $\bB_{\bftheta}$ is positive definite.

We first show how the inner KL approximation controls $\bB_{\bftheta}-\bI$. Using $D_{\bftheta}$ to denote the forward KL divergence in KSIC projection, it can be expressed as .
\[
D_{\bftheta}
=
\frac{1}{2}\left[\tr\left(
\hat\bL_{\bftheta}\hat\bL_{\bftheta}^\top
\bfSigma_{\bftheta}
\right)
-
\log\det\left(
\hat\bL_{\bftheta}\hat\bL_{\bftheta}^\top
\bfSigma_{\bftheta}
\right)
-p\right],
\]
By the definition of $\bB_{\bftheta}$,
\[
2D_{\bftheta} = \tr(\bB_{\bftheta}) - \log\det(\bB_{\bftheta}) -p.
\]
In order to translate this KL bound into a matrix-norm bound, let $\lambda_1,\ldots,\lambda_p>0$ denote the eigenvalues of $\bB_{\bftheta}$. Since
\[
\tr(\bB_{\bftheta}) = \sum_{i=1}^p\lambda_i
\quad\text{and}\quad
\log\det(\bB_{\bftheta}) = \sum_{i=1}^p\log\lambda_i,
\]
we obtain the exact identity
\[2D_{\bftheta} = \sum_{i=1}^p
\left(\lambda_i-\log\lambda_i-1\right).
\]
Define the auxiliary function
$$
h(x)=x-\log x-1,
\qquad x>0.
$$
The uniform accuracy bound of KSIC projection implies
\begin{equation}\label{eq-h-upperbound}
\sum_{i=1}^p h(\lambda_i) = \sum_{i=1}^p\left(\lambda_i-\log\lambda_i-1\right)
\leq 2\delta.
\end{equation}

We note that
\begin{equation}\label{eq-h-lowerbound} 
h(x) \geq \frac{(x-1)^2}{2\max(x,1)}.
\end{equation}
Indeed, if $x\geq1$, then
$$
h(x)
=
\int_1^x\frac{t-1}{t}\,dt
\geq
\frac{1}{x}
\int_1^x(t-1)\,dt
=
\frac{(x-1)^2}{2x}.
$$
If $0<x\leq1$, then
$$
h(x)
=
\int_x^1\frac{1-t}{t}\,dt
\geq
\int_x^1(1-t)\,dt
=
\frac{(1-x)^2}{2}.
$$
This establishes the inequality in both cases.

Because each $h(\lambda_i)$ is nonnegative and
$\sum_i h(\lambda_i)\leq2\delta$, we have
$h(\lambda_i)\leq2\delta$ for every $i$. If $\lambda_i\geq1$, letting
$t_i=\lambda_i-1$ gives
\[\frac{t_i^2}{2(1+t_i)} \leq 2\delta,\]
and hence
\[t_i^2 \leq 4\delta(1+t_i).\]
Solving this quadratic inequality yields
\[
\lambda_i - 1 = t_i \leq 2\delta+2\sqrt{\delta(1+\delta)} = r(\delta).
\]
It follows that $\max(\lambda_i,1)\leq1+r(\delta)$
for every $i$. Using the preceding lower bound \eqref{eq-h-lowerbound} and condition \ref{eq-h-upperbound} for $h$ again,
\[
\begin{split}
\|\bB_{\bftheta}-\bI\|_F^2
& = \sum_{i=1}^p(\lambda_i-1)^2 \leq 2\left(1+r(\delta)\right) \sum_{i=1}^p h(\lambda_i) \leq 4\delta\left(1+r(\delta)\right).
\end{split}
\]
Since $r(\delta)$ is the solution to the equation $t^2 = 4\delta(1+t)$, $r(\delta)^2 = 4\delta\left(1+r(\delta)\right)$,
and therefore
$$
\|\bB_{\bftheta}-\bI\|_F
\leq
r(\delta).
$$

Next we bound the discrepancy between the exact and projected outer objectives. Since
$$
\hat\bL_{\bftheta}\hat\bL_{\bftheta}^\top
=
\bfSigma_{\bftheta}^{-1/2}
\bB_{\bftheta}
\bfSigma_{\bftheta}^{-1/2},
$$
we have
$$
J(\bftheta)-J^0(\bftheta)
=
\tr\left(
\bD_{\bftheta}
(\bB_{\bftheta}-\bI)
\right)
-
\log\det(\bB_{\bftheta}).
$$
Adding and subtracting $\tr(\bB_{\bftheta}-\bI)$ yields
\begin{equation}\label{eq-par-proof-1}
    J(\bftheta)-J^0(\bftheta)
=
\tr\left(
(\bD_{\bftheta}-\bI)
(\bB_{\bftheta}-\bI)
\right)
+
2D_{\bftheta}.
\end{equation}
By the Frobenius Cauchy--Schwarz inequality,
$$
\left|
\tr\left(
(\bD_{\bftheta}-\bI)
(\bB_{\bftheta}-\bI)
\right)
\right|
\leq
\|\bD_{\bftheta}-\bI\|_F
\|\bB_{\bftheta}-\bI\|_F.
$$

Finally, since $\hat{\bftheta}$ minimizes $J(\bftheta)$,
$$
J(\hat{\bftheta})
\leq
J(\hat{\bftheta}_{MLE}).
$$
Therefore,
\[
\begin{split}
    J^0(\hat{\bftheta})-J^0(\hat{\bftheta}_{MLE}) &= \big(J^0(\hat{\bftheta}) - J(\hat{\bftheta})\big)
 + \big(J(\hat{\bftheta}) - J(\hat{\bftheta}_{MLE})\big)
 +\big(J(\hat{\bftheta}_{MLE}) - J^0(\hat{\bftheta}_{MLE})\big)\\
 &\leq \underbrace{J^0(\hat{\bftheta}) - J(\hat{\bftheta})}_{\text{Term 1}}
 + 0
 +\underbrace{J(\hat{\bftheta}_{MLE}) - J^0(\hat{\bftheta}_{MLE})}_{\text{Term 2}}
\end{split}
\]
Using \eqref{eq-par-proof-1}
Term 1 satisfies
$$
J^0(\hat{\bftheta})
-
J(\hat{\bftheta}) = -\tr\left(
(\bD_{\bftheta}-\bI)
(\bB_{\bftheta}-\bI)
\right)
-
2D_{\bftheta} \leq -\tr\left(
(\bD_{\bftheta}-\bI)
(\bB_{\bftheta}-\bI)
\right)
\leq
\|\bD_{\bftheta}-\bI\|_Fr(\delta).
$$
Similary by \eqref{eq-par-proof-1}, Term 2 satisfies
$$
J(\hat{\bftheta}_{MLE})
-
J^0(\hat{\bftheta}_{MLE})
\leq
\|\bD_{\bftheta}-\bI\|_Fr(\delta)+2\delta.
$$
Combing the bounds for Term 1 and Term 2, it follows that
$$
J^0(\hat{\bftheta})-J^0(\hat{\bftheta}_{MLE})
\leq
2\|\bD_{\bftheta}-\bI\|_Fr(\delta)+2\delta.
$$
Applying the quadratic-growth condition,
$$
\frac{\mu}{2}
\|\hat{\bftheta}-\hat{\bftheta}_{MLE}\|_2^2
\leq
2\|\bD_{\bftheta}-\bI\|_Fr(\delta)+2\delta,
$$
which gives
$$
\|\hat{\bftheta}-\hat{\bftheta}_{MLE}\|_2
\leq
2\left[
\frac{
2\|\bD_{\bftheta}-\bI\|_F\big(\delta+\sqrt{\delta(1+\delta)}\big)+\delta
}{\mu}
\right]^{1/2}.
$$
This completes the proof.
\end{proof}

\subsection{Proofs of Other Propositions}

\begin{proof}[Proof of Proposition \ref{prop:pseudocov}]
Fix a mode $k$ and hold $\{\bL_l:l\neq k\}$ fixed. Let $\bL_{-k}=\bigotimes_{l\neq k}\bL_l$, $\bfOmega_{-k}=\bL_{-k}\bL_{-k}^\top$, and $p_{-k}=\prod_{l\neq k}p_l$. Up to a fixed permutation of coordinates, we may write $\bL=\bL_k\otimes\bL_{-k}$. Therefore, after removing constants independent of $\bL_k$, the forward-KL objective in \eqref{eq:ksicproj} reduces to
$$
\mathcal{L}_k(\bL_k)=\tr\left((\bL_k\bL_k^\top\otimes\bfOmega_{-k})\bfSigma\right)-p_{-k}\log\det(\bL_k\bL_k^\top)+C_{-k},
$$
where $C_{-k}$ does not depend on $\bL_k$.

It remains to rewrite the trace term as a standard SIC trace term in mode $k$. Index $\bfSigma$ by pairs $(r,a)$, where $a\in{1,\ldots,p_k}$ indexes the $k$-th mode and $i\in{1,\ldots,p_{-k}}$ indexes all remaining modes. For symmetric matrices $\bA\in\mathbb R^{p_k\times p_k}$ and $\bB\in\mathbb R^{p_{-k}\times p_{-k}}$,
$$
\tr\left((\bA\otimes\bB)\bfSigma\right)
=\sum_{a,b=1}^{p_k}\sum_{i,j=1}^{p_{-k}}\bA_{a,b}\bB_{i,j}\bfSigma_{(a,i),(b,j)}
=\tr\left(\bA\sum_{i,j=1}^{p_{-k}}\bB_{i,j}\bfSigma_{(i,j)}^{(k)}\right).
$$
Taking $\bA=\bL_k\bL_k^\top$ and $\bB=\bfOmega_{-k}$ gives
$$
\tr\left((\bL_k\bL_k^\top\otimes\bfOmega_{-k})\bfSigma\right)=\tr(\bL_k\bL_k^\top\bD_k),\qquad
\bD_k=\sum_{i,j=1}^{p_{-k}}(\bfOmega_{-k})_{i,j}\bfSigma_{(i,j)}^{(k)}.
$$
Thus, with
$$
\tilde{\bfSigma}_k=\frac{1}{p_{-k}}\bD_k=\frac{1}{p_{-k}}\sum_{i,j=1}^{p_{-k}}\left(\bL_{-k}\bL_{-k}^\top\right)_{i,j}\bfSigma_{(i,j)}^{(k)},
$$
the conditional objective can be written as
$$
\mathcal{L}_k(\bL_k)=p_{-k}\left(\tr(\bL_k\bL_k^\top\tilde{\bfSigma}_k)-\log\det(\bL_k\bL_k^\top)\right)+C_{-k}.
$$
Since $p_{-k}>0$ and $C_{-k}$ is independent of $\bL_k$, minimizing $\mathcal{L}_k(\bL_k)$ over $\cS_k$ is exactly equivalent to the standard SIC projection problem with covariance matrix $\tilde{\bfSigma}*k$. Hence
$$
\operatorname*{arg\,min}_{\bL_k:\, \bL \in \mathcal{S}_{\text{KS}}}
\mathrm{KL}\left(\mathcal N(\bfzero,\bfSigma)\,\|\,\mathcal N(\bfzero,(\bL\bL^\top)^{-1})\right)
=\Pi(\tilde{\bfSigma}_k,\mathcal S_k).
$$

Finally, $\tilde{\bfSigma}_k$ is positive semidefinite whenever $\bfSigma$ is positive semidefinite. For any $\bm{x}\in\mathbb R^{p_k}$,
$$
\bx^\top\tilde{\bfSigma}_k\bx
=\sum_{i,j=1}^{p_{-k}}(\bfOmega_{-k})_{i,j}\bx^\top\bfSigma_{(i,j)}^{(k)}\bx
=\tr\left((\bx\bx^\top\otimes\bfOmega_{-k})\bfSigma\right)\ge 0,
$$
because $\bx\bx^\top\otimes\bfOmega_{-k}$ and $\bfSigma$ are both positive semidefinite. Therefore $\bD_k$, and hence $\tilde{\bfSigma}_k$, is positive semidefinite. If $\bfSigma$ is positive definite and $\bL_{-k}$ is nonsingular, then $\bfOmega_{-k}$ is positive definite and the same argument gives $\bx^\top\bD_k\bx>0$ for every nonzero $\bx$, so $\tilde{\bfSigma}_k$ is positive definite.
\end{proof}

\begin{proof}[Proof of Proposition \ref{prop-low_rank}]
The result follows by specializing the pseudo-covariance formula in Proposition \ref{prop:pseudocov} to the low-rank representation $\bfSigma=\sum_{r=1}^n\bx_r\bx_r^\top$. Use the same mode-$k$ indexing as in Proposition \ref{prop:pseudocov}: write each index of $\bfSigma$ as a pair $(a,i)$, where $a\in{1,\ldots,p_k}$ indexes the $k$-th mode and $i\in{1,\ldots,p_{-k}}$ indexes all remaining modes. Let $\bX_r^{(k)}\in\mathbb R^{p_k\times p_{-k}}$ be the mode-$k$ unfolding of $\bx_r$, so that ${\bX_r^{(k)}}_{a,i}={\bx_r}_{(a,i)}$.

By Proposition \ref{prop:pseudocov}, under the notation $\bfOmega_{-k}=\bL_{-k}\bL_{-k}^\top$,
$$
\tilde{\bfSigma}_k=\frac{1}{p_{-k}}\bD_k,\qquad
\bD_k=\sum_{i,j=1}^{p_{-k}}(\bfOmega_{-k})_{i,j}\bfSigma_{(i,j)}^{(k)}.
$$
For the low-rank covariance $\bfSigma=\sum_{r=1}^n\bx_r\bx_r^\top$, the $(a,b)$-th entry of $\bD_k$ is
$$
\begin{aligned}
(\bD_k)_{a,b}
&=\sum_{i,j=1}^{p_{-k}}(\bfOmega_{-k})_{i,j}\bfSigma_{(a,i),(b,j)}
=\sum_{i,j=1}^{p_{-k}}(\bfOmega_{-k})_{i,j}\sum_{r=1}^n{\bx_r}_{(a,i)}{\bx_r}_{(b,j)}  \\
&=\sum_{r=1}^n\sum_{i,j=1}^{p_{-k}}{\bX_r^{(k)}}_{a,i}(\bfOmega_{-k})_{i,j}{\bX_r^{(k)}}_{b,j}
=\sum_{r=1}^n\left(\bX_r^{(k)}\bfOmega_{-k}\bX_r^{(k)\top}\right)_{a,b}.
\end{aligned}
$$
Therefore,
$$
\bD_k=\sum_{r=1}^n\bX_r^{(k)}\bfOmega_{-k}\bX_r^{(k)\top}
=\sum_{r=1}^n\bX_r^{(k)}\bL_{-k}\bL_{-k}^\top\bX_r^{(k)\top}.
$$
Dividing by $p_{-k}$ gives
$$
\tilde{\bfSigma}_k=\frac{1}{p_{-k}}\sum_{r=1}^n\bX_r^{(k)}\bL_{-k}\bL_{-k}^\top\bX_r^{(k)\top},
$$
which proves the claim.
\end{proof}

\begin{proof}[Proof of Proposition \ref{prop:cost}]
 By the definition of the Kronecker product, the fixed factor $\bL_{-k}$ has at most $m_{-k}$ non-zero elements in each column. Consequently, the symmetric matrix $\bL_{-k}\bL^\top_{-k}$ contains at most $\mathcal{O}(m_{-k}^2)$ non-zero elements per row/column. 

In Algorithm~\ref{alg-main}, computing the pseudo-covariance matrix $\tilde{\bfSigma}_k$ does not require evaluating the full $p_k \times p_k$ matrix. The SIC projection only requires the specific entries of $\tilde{\bfSigma}_k$ corresponding to the non-zero patterns of $\cS_k$. For each of the $p_k$ rows, we need at most an $m_k \times m_k$ submatrix, meaning we only compute $\mathcal{O}(m_k^2 p_k)$ entries of $\tilde{\bfSigma}_k$. 

According to \eqref{eq:pseudocov-comp}, evaluating each required entry of $\tilde{\bfSigma}_k$ involves a weighted sum over the non-zero elements of $\bL_{-k}\bL^\top_{-k}$. Because the total number of non-zero elements in a row/column of $\bL_{-k}\bL^\top_{-k}$ is bounded by $\mathcal{O}(m_{-k}^2)$, computing a single entry takes $\mathcal{O}(p_{-k} m_{-k}^2)$ flops. Computing all required $m_k^2 p_k$ entries of the pseudo-covariance therefore requires $\mathcal{O}(m_k^2 p_k \cdot p_{-k} m_{-k}^2)$ operations. Recognizing that $p_k p_{-k} = p$ and $m_k m_{-k} = m$, this cost simplifies exactly to $\mathcal{O}(m^2 p)$ flops for a single mode.

Computing the SIC projection from these submatrices requires inverting an $m_k \times m_k$ matrix $p_k$ times, adding a cost of $\mathcal{O}(m_k^3 p_k)$. Summing these costs over all $K$ modes for a single BCD iteration yields the total computational cost of:
$$
\mathcal{O}\left( K m^2 p + \sum_{k=1}^K m_k^3 p_k \right).
$$
Assuming the conditioning-set sizes are fixed constants such that $m_k \ll p_k$, the total cost scales linearly with the total dimension, $\mathcal{O}(Kp)$.
\end{proof}

\begin{proof}[Proof of Proposition \ref{prop:cost2}]
In the low-rank setting, $\bfSigma = \sum_{r=1}^n \bx_r \bx_r^\top$. Let $\bX_r^{(k)} \in \mathbb{R}^{p_k \times p_{-k}}$ denote the mode-$k$ unfolding of $\bx_r$. According to Proposition \ref{prop-low_rank}, the pseudo-covariance can be expressed as:
$$ \tilde{\bfSigma}_{k} = \frac{1}{p_{-k}}\sum_{r=1}^n \bX^{(k)}_{r}\bL_{-k}\bL^\top_{-k} (\bX^{(k)}_{r})^\top $$

To compute this efficiently, we first define the intermediate matrices $\bW_r = \bX^{(k)}_{r}\bL_{-k}$ for $r = 1 \ldots, n$. Because the sparse matrix $\bL_{-k}$ has at most $m_{-k}$ non-zeros per column, computing $\bW_r$ requires multiplying the $p_k \times p_{-k}$ dense matrix $\bX^{(k)}_{r}$ by $\bL_{-k}$, which takes $\mathcal{O}(p_k p_{-k} m_{-k}) = \mathcal{O}(p m_{-k})$ operations. Doing this for all $n$ rank components requires $\mathcal{O}(n m_{-k} p)$ flops.

Next, we extract the specific elements of $\tilde{\bfSigma}_k$ required for the SIC projection. Notice that the pseudo-covariance is now simply $\tilde{\bfSigma}_{k} = \frac{1}{p_{-k}}\sum_{r=1}^n \bW_r \bW_r^\top$. We only need to compute at most $p_k$ submatrices of maximum size $m_k \times m_k$. 

For any required entry $(i, j)$ in $\tilde{\bfSigma}_k$, we compute the inner product of the $i$-th row and $j$-th row of $\bW_r$. Since these rows have length $p_{-k}$, computing a single entry for one rank component takes $\mathcal{O}(p_{-k})$ flops. To compute the $m_k^2$ required entries for all $p_k$ rows across all $n$ components, the total inner-product cost is $\mathcal{O}(n \cdot m_k^2 p_k \cdot p_{-k}) = \mathcal{O}(n m_k^2 p)$ flops.

Finally, computing the closed-form SIC projection from these submatrices requires $\mathcal{O}(m_k^3 p_k)$ flops. Summing these three steps over all $K$ modes for a single BCD iteration gives a total computational cost of:
$$
\mathcal{O}\left( n \Big(\sum_{k=1}^K m_{-k}\Big) p + n \Big(\sum_{k=1}^K m^2_k \Big) p + \sum_{k=1}^K m^3_k p_k \right).
$$
Assuming the conditioning-set sizes $m_k$ are fixed, this reduces to $\mathcal{O}(Kn p)$.
\end{proof}

\begin{proof}[Proof of Proposition \ref{prop:separable}]
Suppose the true joint covariance matrix is exactly separable, meaning $\bfSigma = \bigotimes_{l=1}^K \bfSigma_l$. Without loss of generality, for any chosen mode $k$, we can rearrange the Kronecker factors to group all remaining modes together, writing $\bfSigma = \bfSigma_k \otimes \bfSigma_{-k}$ where $\bfSigma_{-k} = \bigotimes_{l \neq k} \bfSigma_l$. Similarly, let the full Kronecker Cholesky factor be represented as $\bL = \bL_k \otimes \bL_{-k}$ where $\bL_{-k} = \bigotimes_{l \neq k} \bL_l$.

By the definition of the fiber submatrices in Proposition 1, the $(i,j)$-th block of $\bfSigma$ corresponding to the mode-$k$ fibers is strictly determined by the entries of the remaining modes:
\[
\bfSigma_{(i,j)}^{(k)} = \{\bfSigma_{-k}\}_{i,j} \bfSigma_k,
\]
where $\{\bfSigma_{-k}\}_{i,j}$ denotes the $(i,j)$-th scalar entry of the matrix $\bfSigma_{-k}$. Substituting this structural form into the pseudo-covariance formula \eqref{eq:pseudocov-comp} from Proposition 1 yields:
\begin{align*}
\tilde{\bfSigma}_k &= \frac{1}{p_{-k}} \sum_{i,j=1}^{p_{-k}} \{\bL_{-k}\bL_{-k}^\top\}_{i,j} \bfSigma_{(i,j)}^{(k)} \\
&= \frac{1}{p_{-k}} \sum_{i,j=1}^{p_{-k}} \{\bL_{-k}\bL_{-k}^\top\}_{i,j} \{\bfSigma_{-k}\}_{i,j} \bfSigma_k \\
&= \left( \frac{1}{p_{-k}} \tr(\bL_{-k}\bL_{-k}^\top \bfSigma_{-k}) \right) \bfSigma_k \\
&= \left( \frac{1}{p_{-k}} \tr(\bL_{-k}^\top \bfSigma_{-k} \bL_{-k}) \right) \bfSigma_k.
\end{align*}
Because $\frac{1}{p_{-k}} \tr(\bL_{-k}^\top \bfSigma_{-k} \bL_{-k}) > 0$ is a scalar constant (Lemma \ref{lemma-const-trace}) that is independent of $\bL_k$, we immediately obtain the proportional relationship:
\[
\tilde{\bfSigma}_k \propto \bfSigma_k.
\]

Next, we establish the scale homogeneity of the standard column-wise SIC projection operator. Let $\bL_k^{(0)} = \Pi(\bfSigma_k, \cS_k)$ be the standard SIC projection onto the sparsity class $\cS_k$. According to the closed-form solution \eqref{eq:sic_closed_form}, the non-zero elements of the $i$-th column are given by $[\bL_k^{(0)}]_{\bs_i^k, i} = \bfbeta_i / (\mathbf{e}_1^\top \bfbeta_i)^{1/2}$, where $\bfbeta_i = (\bfSigma_{k, \bs_i^k, \bs_i^k})^{-1}\mathbf{e}_1$. Under the scaled pseudo-covariance $\tilde{\bfSigma}_k = \alpha_k \bfSigma_k$, the corresponding system becomes $\tilde{\bfbeta}_i = (\alpha_k \bfSigma_{k, \bs_i^k, \bs_i^k})^{-1}\mathbf{e}_1 = \alpha_k^{-1} \bfbeta_i$. Evaluating the column entries under this scaling yields:
\[
[\hat{\bL}_k]_{\bs_i^k, i} = \frac{\alpha_k^{-1} \bfbeta_i}{(\alpha_k^{-1} \mathbf{e}_1^\top \bfbeta_i)^{1/2}} = \alpha_k^{-1/2} \frac{\bfbeta_i}{(\mathbf{e}_1^\top \bfbeta_i)^{1/2}} = \alpha_k^{-1/2} [\bL_k^{(0)}]_{\bs_i^k, i}.
\]
Thus, the mode-wise SIC projection satisfies the scale-homogeneity property:
\[
\Pi(\tilde{\bfSigma}_k, \cS_k) = \Pi(\alpha_k \bfSigma_k, \cS_k) = \alpha_k^{-1/2} \Pi(\bfSigma_k, \cS_k).
\]
This implies that during any block update of the KSIC-BCD algorithm, the updated factor $\hat{\bL}_k$ is exactly a scalar multiple of the marginal SIC factor $\bL_k^{(0)}$, regardless of the values of the remaining modes. Consequently, the unique global minimizer of the forward-KL objective must decouple into the form $\hat{\bL} = c \left(\bigotimes_{k=1}^K \bL_k^{(0)}\right)$ for some global scaling constant $c > 0$.

To determine the optimal global scale $c$, we substitute $\bL = c \bL^{(0)}$ into the global forward-KL objective function:
\[
f(c) = \tr(c^2 \bL^{(0)}\bL^{(0)\top} \bfSigma) - \log\det(c^2 \bL^{(0)}\bL^{(0)\top}) = c^2 \tr(\bL^{(0)}\bL^{(0)\top} \bfSigma) - 2p \log c - \log\det(\bL^{(0)}\bL^{(0)\top}).
\]
Differentiating with respect to $c$ and setting the gradient to zero yields the optimal scale $c^2 = p / \tr(\bL^{(0)}\bL^{(0)\top} \bfSigma)$. Using the properties of the Kronecker product and the trace operator, the trace term expands as:
\[
\tr(\bL^{(0)}\bL^{(0)\top} \bfSigma) = \tr\left( \bigotimes_{k=1}^K \bL_k^{(0)}\bL_k^{(0)\top}\bfSigma_k \right) = \prod_{k=1}^K \tr(\bL_k^{(0)\top}\bfSigma_k\bL_k^{(0)}).
\]
By the structural property of individual M-projections, each column $i$ of a standard SIC projection satisfies the property $\bL_{k, \bs_i^k, i}^{(0)\top} \bfSigma_{k, \bs_i^k, \bs_i^k} \bL_{k, \bs_i^k, i}^{(0)} = 1$. Summing over all $p_k$ columns yields $\tr(\bL_k^{(0)\top}\bfSigma_k\bL_k^{(0)}) = p_k$ for every mode $k$. Therefore, the product of the traces simplifies exactly to:
\[
\tr(\bL^{(0)}\bL^{(0)\top} \bfSigma) = \prod_{k=1}^K p_k = p.
\]
This directly implies $c^2 = p/p = 1$, confirming that $c = 1$. The global KSIC projection thus completely decouples into the product of independent mode-wise SIC projections:
\[
\Pi_{\text{KS}}(\bfSigma,\mathcal{S}_{\text{KS}}) = \bigotimes_{k=1}^K \hat{\bL}_k, \quad \text{where } \hat{\bL}_k = \Pi(\bfSigma_k,\mathcal{S}_k).
\]
\end{proof}

\begin{proof}[Proof of Proposition \ref{prop-exist-nest}]
By the definition of KSIC projection \eqref{eq:ksicproj}, Proposition \ref{prop-exist-nest} holds naturally as the conditions imply that $\cS^a_{\text{KS}} \subseteq \cS^b_{\text{KS}}$.
\end{proof}

\section{Bi-Level Optimization via Implicit Differentiation}
\label{subsec:implicit_diff}

In Section~\ref{subsec:parametric_estimation}, numerical gradients are used to optimize the projected negative log-likelihood. In this section, we present alternative strategies for parametric estimation under the KSIC framework based on implicit differentiation. We emphasize that all experiments in this work use numerical gradients, as this approach is still the most direct and computationally efficient given our current implementation of KSIC in \texttt{R}. Implementing implicit differentiation will be an important direction for future software development.
 
 To solve the optimization problem defined by \eqref{eq:outer}--\eqref{eq:inner}, we re-parameterize the set of Cholesky factors as a vector $\bm{\psi} = \{ \text{vec}(\mathbf{L}_1), \dots, \text{vec}(\mathbf{L}_K) \}$ subject to the sparsity constraints $\mathcal{S}_{\text{KS}}$. Let $\mathbf{\Omega}_{\bm{\psi}} = \bigotimes_{k=1}^K \mathbf{L}_k \mathbf{L}_k^\top$ denote the resulting precision matrix.

The problem can be written as minimizing the outer objective $\mathcal{L}_{\text{outer}}(\bm{\theta})$ subject to $\bm{\psi}^*(\bm{\theta})$ being the minimizer of the inner objective $\mathcal{L}_{\text{inner}}(\bm{\psi}, \bm{\theta})$, where:
\begin{align}
    \mathcal{L}_{\text{inner}}(\bm{\psi}, \bm{\theta}) &\coloneqq \mathrm{KL}\left( \mathcal{N}(\mathbf{0}, \mathbf{\Sigma}_{\bm{\theta}}) \;\|\; \mathcal{N}(\mathbf{0}, \mathbf{\Omega}_{\bm{\psi}}^{-1}) \right), \label{eq:inner_loss} \\
    \mathcal{L}_{\text{outer}}(\bm{\theta}) &\coloneqq \mathrm{KL}\left( \mathcal{N}(\mathbf{0}, \bfSigma_\text{emp}) \;\|\; \mathcal{N}(\mathbf{0}, \mathbf{\Omega}_{\bm{\psi}^*(\bm{\theta})}^{-1}) \right). \label{eq:outer_loss}
\end{align}
The mapping $\bm{\theta} \mapsto \bm{\psi}^*(\bm{\theta})$ is implicit, defined only as the fixed point of the BCD algorithm described in Section \ref{subsec:KSIC}. Standard backpropagation through the unrolled BCD iterations is memory-intensive ($\mathcal{O}(T \cdot p)$ where $T$ is the number of iterations) and numerically unstable. Instead, one can employ \textit{implicit differentiation} to compute the gradient $\nabla_{\bm{\theta}} \mathcal{L}_{\text{outer}}$ exactly.

Let $G(\bm{\psi}, \bm{\theta}) = \nabla_{\bm{\psi}} \mathcal{L}_{\text{inner}}(\bm{\psi}, \bm{\theta})$ denote the gradient of the inner objective. At the inner optimum $\bm{\psi}^*$, the stationarity condition $G(\bm{\psi}^*, \bm{\theta}) = \mathbf{0}$ holds. By the Implicit Function Theorem, the Jacobian $\mathbf{J}_{\bm{\psi}} = \frac{\partial \bm{\psi}^*}{\partial \bm{\theta}}$ satisfies:
\begin{equation}
    \nabla_{\bm{\psi}} G(\bm{\psi}^*, \bm{\theta}) \mathbf{J}_{\bm{\psi}} + \nabla_{\bm{\theta}} G(\bm{\psi}^*, \bm{\theta}) = \mathbf{0} 
    \implies \mathbf{J}_{\bm{\psi}} = - [\mathbf{H}_{\bm{\psi}}]^{-1} \mathbf{B}_{\bm{\psi}\bm{\theta}},
\end{equation}
where $\mathbf{H}_{\bm{\psi}} = \nabla^2_{\bm{\psi}\bm{\psi}} \mathcal{L}_{\text{inner}}$ is the Hessian of the inner objective with respect to the Cholesky factors, and $\mathbf{B}_{\bm{\psi}\bm{\theta}} = \nabla^2_{\bm{\psi}\bm{\theta}} \mathcal{L}_{\text{inner}}$ is the mixed partial derivative matrix.

The total gradient of the outer objective required for optimization is derived via the chain rule:
\begin{equation}
    \nabla_{\bm{\theta}} \mathcal{L}_{\text{outer}} = \mathbf{J}_{\bm{\psi}}^\top \nabla_{\bm{\psi}} \mathcal{L}_{\text{outer}} + \frac{\partial \mathcal{L}_{\text{outer}}}{\partial \bm{\theta}}.
\end{equation}
Note that $\frac{\partial \mathcal{L}_{\text{outer}}}{\partial \bm{\theta}} = 0$ because $\mathcal{L}_{\text{outer}}$ depends on $\bm{\theta}$ only through $\bm{\psi}^*$. Substituting the implicit Jacobian, we obtain:
\begin{equation}
    \nabla_{\bm{\theta}} \mathcal{L}_{\text{outer}} = - \mathbf{B}_{\bm{\psi}\bm{\theta}}^\top [\mathbf{H}_{\bm{\psi}}]^{-1} \mathbf{g}_{\text{outer}},
    \label{eq:implicit_grad}
\end{equation}
where $\mathbf{g}_{\text{outer}} = \nabla_{\bm{\psi}} \mathcal{L}_{\text{outer}} |_{\bm{\psi}=\bm{\psi}^*}$ is the gradient of the outer objective with respect to the Cholesky factors.

\textbf{Efficient Linear Solves:} We do not explicitly invert the Hessian $\mathbf{H}_{\bm{\psi}}$. Instead, we solve the linear system $\mathbf{H}_{\bm{\psi}} \mathbf{v} = \mathbf{g}_{\text{outer}}$ for $\mathbf{v}$ using Conjugate Gradients. A key advantage of the KSIC framework is that the inner problem decouples over the Kronecker factors (or is solved via BCD), rendering $\mathbf{H}_{\bm{\psi}}$ block-diagonal dominant. This structure allows for efficient matrix-vector products, making the implicit gradient computation scalable.

\begin{proposition}\label{prop:implicit-grad}
Let $\bfSigma_{\bm{\theta}}$ be the parametric model and $\bfSigma_{\text{emp}}$ be the empirical covariance. Denote the Kronecker product of $\bL_{k}$'s by $\bL(\bm{\psi}) = \bigotimes_{k=1}^K\bL_{k}$. The terms in \eqref{eq:implicit_grad} have closed form:
\begin{equation}
\begin{split}
\mathbf{B}_{\bm{\psi}\bm{\theta}} &= -\text{vec}_{\cS}\big(\frac{\partial\bfSigma_{\bm{\theta}}}{\partial\bm{\theta}}\bL(\bm{\psi})\big);\\
\mathbf{H}_{\bm{\psi}} &= -\text{diag}\big\{\bD_i: i = 1, \ldots, p\big\}; \text{ where } \ \bD_i = \bL^{-2}_{ii}(\bm{\psi})\bm{e}_1\bm{e}^\top_1 + \bfSigma_{\bm{\theta}[\bs_i]};\\
\bg_{\text{outer}} &= \text{vec}_{\cS}\big(\bL^{-\top}(\bm{\psi}) - \bfSigma_{\text{emp}}\bL(\bm{\psi})\big),
\end{split}    
\end{equation}
where $\text{diag}\{\cdot\}$ denotes a block diagonal matrix, and $\text{vec}_{\cS}$ denotes the matrix-to-vector projection that only keeps the elements allowed by $\cS_{\text{KS}}$.
The gradient of outer objective $\nabla_{\bm{\theta}} \mathcal{L}_{\text{outer}}$ can be expressed as:
\begin{equation}
\frac{d\mathcal{L}_{\text{outer}}}{d\bm{\theta}}  = \sum_{i = 1}^p \langle \bD_{i}^{-1}\frac{\partial\bfSigma_{\bm{\theta}[\bs_i]}}{\partial\bm{\theta}}\bL_{\bs_i, i}, \bL_{ii}^{-1}\bm{e}_1 - \bfSigma_{\text{emp}[\bs_i]}\bL_{\bs_i, i}\rangle
\end{equation}
\end{proposition}

\textbf{Extension: Implicit Natural Gradient.} In covariance estimation, parameters are often highly correlated (e.g., variance and range parameters), leading to ill-conditioned likelihood landscapes where first-order methods (like L-BFGS) converge slowly. To address this, we can extend the implicit differentiation framework to perform \textit{Implicit Fisher Scoring}. We approximate the Fisher Information Matrix (FIM) of the outer parameters, $\mathcal{I}_{\bm{\theta}}$, by projecting the FIM of the inner structured model, $\mathcal{I}_{\bm{\psi}}$, back onto the parameter space $\Theta$:
\begin{equation}
    \mathcal{I}_{\bm{\theta}} \approx \mathbf{J}_{\bm{\psi}}^\top \mathcal{I}_{\bm{\psi}} \mathbf{J}_{\bm{\psi}}. 
\end{equation}
Here, $\mathcal{I}_{\bm{\psi}}$ is the Fisher Information of the Kronecker-structured Gaussian model with respect to its Cholesky factors. It has a formula similar to $\mathbf{H}_{\bm{\psi}}$:
\begin{equation}
    \mathcal{I}_{\bm{\psi}} = -\text{diag}\big\{\bD_{\text{emp},i}: i = 1, \ldots, p\big\}, \;\text{ where } \ \bD_{\text{emp},i} = \bL^{-2}_{ii}(\bm{\psi})\bm{e}_1\bm{e}^\top_1 + \bfSigma_{\text{emp}[\bs_i]}.
\end{equation}
Given $\bL_k$'s estimated from KSIC-BCD, $\mathcal{I}_{\bm{\psi}}$ is sparse and computationally cheap to evaluate. The update rule then becomes $\bm{\theta}_{t+1} = \bm{\theta}_t - \eta \mathcal{I}_{\bm{\theta}}^{-1} \nabla_{\bm{\theta}} \mathcal{L}_{\text{outer}}$. This preconditioning effectively ``whitens'' the gradient using the geometry of the KSIC manifold, significantly accelerating convergence for confounding parameters.

\section{Connection to Gaussian graphical models}\label{app:GGM}
Given a $p$-dimensional random vector
$\bm{x} = (x_1,\ldots,x_p)^\top \sim \mathcal{N}(\bm{0},\bfSigma)$,
with $\bfSigma \succ 0$, its density is
\[
p(\bm{x}\mid \bfSigma)
=
(2\pi)^{-p/2}|\bfSigma|^{-1/2}
\exp\left\{
-\frac{1}{2}\bm{x}^\top \bfSigma^{-1}\bm{x}
\right\}.
\]
Gaussian graphical models provide a graphical representation of the
conditional dependence structure among the components of $\bm{x}$.
Specifically, the variables $x_1,\ldots,x_p$ are represented as nodes in
an undirected graph, and edges encode conditional dependence after
adjusting for all remaining variables. This structure is characterized by
the precision matrix $\bfOmega = \bfSigma^{-1}$.
Under the Gaussian assumption,
\[
x_i \perp x_j \mid \{x_\ell:\ell\neq i,j\}
\quad \Longleftrightarrow \quad
\Omega_{ij}=0.
\]
Thus, zeros in the precision matrix correspond to missing edges in the
graph, while nonzero off-diagonal entries correspond to direct conditional dependencies.

Given independent multi-way observations
$\{\mathcal{X}_j:j=1,\ldots,n\}$, sparse Gaussian graphical
modeling estimates a sparse precision matrix by solving a penalized
Gaussian negative log-likelihood problem of the form
\begin{equation}\label{eq-GG-KL}
\widehat{\bfOmega}
=
\argmin_{\bfOmega\succ 0}
\left\{
\operatorname{tr}(\bfOmega \bfSigma_{\mathrm{emp}})
-
\log\det(\bfOmega)
+
P(\bfOmega)
\right\}.
\end{equation}
When
\[
P(\bfOmega)=\lambda\sum_{i\neq j}|\Omega_{ij}|,
\qquad \lambda>0,
\]
the objective is convex in $\bfOmega$, and the resulting estimator is the
graphical Lasso estimator
\citep{yuan2007model, banerjee2008model, friedman2008sparse}. The
$\ell_1$ penalty encourages sparsity in the off-diagonal entries of
$\bfOmega$, thereby inducing a sparse conditional independence graph.

For multi-way data, it is often of interest to learn conditional
dependence structures associated with each mode of the tensor. This
motivates tensor graphical models, in which the joint precision matrix is
parameterized through mode-specific precision matrices. Suppose the joint
precision matrix can be written as a structured composition
\[
\bQ(\bfOmega_1,\ldots,\bfOmega_K),
\]
where $\bfOmega_k$ encodes the graphical structure along the $k$th mode.
Then the corresponding penalized likelihood problem can be written as
\begin{equation}\label{eq-GG-KL2}
\argmin_{\bfOmega_1,\ldots,\bfOmega_K}
\left\{
\operatorname{tr}\left[
\bQ(\bfOmega_1,\ldots,\bfOmega_K)
\bfSigma_{\mathrm{emp}}
\right]
-
\log\det\left[
\bQ(\bfOmega_1,\ldots,\bfOmega_K)
\right]
+
P(\bfOmega_1,\ldots,\bfOmega_K)
\right\}.
\end{equation}
A common example is the Kronecker graphical Lasso, where
\[
\bQ(\bfOmega_1,\ldots,\bfOmega_K)
=
\bigotimes_{k=1}^K \bfOmega_k, \text{  and  }
P(\bfOmega_1,\ldots,\bfOmega_K)
=
\sum_{k=1}^K
\lambda_k
\sum_{i\neq j}
|\Omega_{k,ij}|.
\]
The resulting estimator is referred to as the Kronecker graphical Lasso (KGLasso)
estimator
\citep{allen2012inference, tsiligkaridis2013convergence}. Under this
Kronecker precision structure, the global conditional dependence between
two tensor entries indexed by
$\bm{i}=(i_1,\ldots,i_K)$ and $\bm{j}=(j_1,\ldots,j_K)$ is determined by $\bigotimes_{k=1}^K \bfOmega_k$.
Consequently, a global edge is present only if all corresponding
mode-specific precision entries are nonzero. Equivalently, for every mode
in which $i_k\neq j_k$, the corresponding pair of mode-specific variables
must be connected in the $k$th graphical model.

The Gaussian graphical objective in \eqref{eq-GG-KL2} is closely related
to the KSIC objective in \eqref{eq:ksicproj}. Both are derived from
Gaussian likelihoods and involve estimating or approximating sparse
precision-type structures. Moreover, both frameworks admit graphical
interpretations through conditional independence or conditional
prediction relationships.

However, Gaussian graphical models and SIC-type approaches, including
KSIC, differ in both their parameterization and their interpretation.
Gaussian graphical models impose sparsity directly on the precision matrix
$\bfOmega=\bfSigma^{-1}$. As a result, their graphical interpretation is
symmetric and order-free: a zero entry $\bfOmega_{ij}=0$ is equivalent, under
Gaussianity, to the conditional independence of $x_i$ and $x_j$ given all
remaining variables. Therefore, the main object of inference in Gaussian
graphical modeling is the precision matrix itself, together with the
undirected graph encoded by its sparsity pattern.

By contrast, SIC-type methods parameterize the precision matrix through a
SIC factor $\bL$.
The sparsity is therefore imposed on $\bL$ rather than directly on
$\bfOmega = \bL\bL^\top$. This leads to a different interpretation. The nonzero pattern
of $\bL$ describes which previously ordered variables are used in the
conditional prediction of each variable. Equivalently, it encodes a set of
ordered conditional regressions rather than an undirected conditional
independence graph. Consequently, the interpretation of an SIC
approximation is generally ordering-dependent and is closely tied to the
chosen conditioning sets, such as nearest-neighbor sets in Vecchia-type
approximations.

This distinction also reflects different methodological goals. Gaussian
graphical models primarily aim to learn a sparse precision matrix and its
associated conditional independence graph from data. SIC-type methods,
including KSIC, are more naturally viewed as structured sparse
factorization or projection methods: given a target covariance or precision
structure, they construct a computationally efficient sparse inverse
Cholesky approximation. In particular, Vecchia-based SIC methods can use
additional information, such as spatial or mode-specific geometry, to
choose meaningful conditioning sets. Thus, while Gaussian graphical models
emphasize graph-level inference through sparsity in $\bfOmega$, KSIC
emphasizes scalable likelihood computation and covariance approximation
through sparsity in $\bL$.

Despite these differences, Vecchia-based SIC methods and Gaussian
graphical models are closely connected through their use of sparse
precision representations. Highlighting this connection allows techniques
and insights from graphical modeling, sparse matrix approximation, and
large-scale Gaussian likelihood computation to inform one another.

\section{Implementation Details of Competing Methods}\label{app:sim-detail}
This section presents the details about the implementation of methods used in Sections \ref{sec:sim} and \ref{sec:real}.

For \emph{KSIC}, due to the limited sample size, we always initialize using identity matrix throughout the experiments. We set the convergence criterion to be $\|\bL^{(t)}_k - \bL^{(t-1)}_{k}\|_F/\|\bL^{(t-1)}_{k}\|_F \leq 10^{-3}$ for all $k$ after a full cycle through all $k$ modes. Typically, the number of iterations taken for KSIC to converge is around $4$. Under a nonparametric paradigm, in the adaptive choice of $m$, the candidates are $\{5, 10, 15, 20, 30, 40, 50\}$ in Section \ref{sec:sim} and \ref{sec:real-temp} and $\{5, 10, 15, 20, 40, 60, \cdots, 300\}$ in  Section \ref{sec:real-mri} due to the unknown geometry there among the ROIs. Under a nonparametric setting, nugget regularization is applied to KSIC when the sample size is below the existence threshold. Throughout the experiments, we use a nugget of magnitude $10^{-6}$.

For \emph{Naive} and \emph{Naive-Vecchia}, we emphasize that their fundamental difference from \emph{KSIC} lies not only in the number of iterations but also in the updating scheme. The KSIC-BCD algorithm sequentially computes and stores each $\bL_k$, immediately using the updated factor to update the remaining factors $\bL_l$ for $l\neq k$ within the same iteration. In contrast, the Naive-type approaches compute the $\bL_k$'s in parallel. Even if these approaches are extended to iterative versions, the updated factors are used only in the next iteration, and therefore the resulting updates are not guaranteed to decrease the objective. Apart from this difference, when implementing \emph{Naive} and \emph{Naive-Vecchia}, we use the same adaptive choice of $m$ and the same nugget regularization procedure as in \emph{KSIC}.

Both \emph{KSIC} and \emph{Naive-Vecchia} involve SIC. The output of SIC depends on the ordering of the indices. The simplest approach is the random ordering, which does not rely on additional knowledge (this is important in applications such as fMRI) and generally performs well. In most of the experiments in Section \ref{sec:sim} and \ref{sec:real}, we apply random ordering by default. Maximin ordering from \cite{Guinness2016a, katzfuss2021general} is another common choice that guarantees conditioning on an increasingly fine scale. It sharpens the GP approximation and is recommended when geometry information is clearly available, for example, in spatial analysis. In this paper, maximin ordering is applied to the spatial locations in Section \ref{sec:real-temp}. 

\emph{MLE} and \emph{MMCD} are implemented by package \texttt{robustmatrix} in \texttt{R} \citep{mayrhofer2024robustmatrix}. Specifically for \emph{MMCD}, we use $\alpha = 0.8$, $\lambda = 0$ and $nsamp = 50$. Roughly speaking, it tries to maximize the likelihoods computed based on $nsamp = 50$ different subsamples with $n_{\text{sub}} = \alpha n$. \emph{Glasso} is implemented using the state-of-the-art non-cyclic algorithm from \cite{min2022fast}. Since \emph{Glasso} requires a validation set to select the penalty parameter, it is not applicable when $n = 1$. For experiments with $n>1$ replicates, we use a $50\%-50\%$ training-validation split. The $\ell_1$ penalty is chosen from candidates $10^{\{-4, -3.5, -3, -2.5, -2, -1.5, -1\}}$.

The \emph{Robust} method uses the robust covariance estimation algorithm proposed by \cite{zhang2023covariance}. The original implementation was written in \texttt{Matlab}; for a fair comparison, we reimplemented the method in \texttt{R}. A key step in the algorithm is the rank-one unconstrained Kronecker product approximation. In the original work, this step was implemented using the \texttt{TKPSVD} package in \texttt{Matlab} \citep{batselier2017constructive}; here, we instead use the \texttt{svd} function from base \texttt{R}. For the tuning parameter $k$, which controls the bandwidth, we apply an adaptive selection procedure similar to that used for \emph{KSIC}. We implemented both the bandable and tapering covariance estimators from the original work, but report only the bandable estimator because it consistently outperformed the tapering estimator in our experiments after adaptive selection of $k$.

The Log-score for \emph{Robust} is omitted from most plots because it is substantially worse than those of the other methods. In our experiments, the raw covariance estimate produced by the rank-one unconstrained Kronecker product approximation is often not positive definite. Additional post-processing steps, including symmetrization and eigenvalue truncation, enforce positive definiteness but further distort the covariance structure. As a result, the Log-score is often unstable and nearly singular. We conjecture that this poor performance is primarily driven by the small-sample regime, since the singularity issue becomes less severe as the sample size increases.

For the parametric approaches used in Section \ref{sec:real-temp}, we fit an isotropic Mat\'ern model for each mode. Due to the isotropic nature of the model and the high dimensionality, we downsample the data to 100 spatial locations when estimating the parameters. The parameters include variance ($\sigma^2$), range ($\phi$), smoothness ($\nu$), and nugget ($\tau^2$). Denote the Mat\'ern covariance by 
\[
C_{Mat}(d|\sigma^2, \phi, \nu) = \sigma^2
\frac{2^{1-\nu}}{\Gamma(\nu)}
\left(
\frac{\sqrt{2\nu}d}{\phi}
\right)^{\nu}
K_{\nu}\left(
\frac{\sqrt{2\nu}d}{\phi}
\right).
\]
For \emph{KSIC-par}, we jointly estimate $\bftheta = \{\sigma^2, \phi_s, \phi_t, \nu_s, \nu_t, \tau_s, \tau_t\}$ from spatio-temporal fields $\{\mathcal{X}_i\}_{i=1}^n$ using KSIC projected likelihood parametric estimation algorithm. The joint covariance model is:
\[
Cov((s_i, t_i), (s_j, t_j)|\bftheta) = \sigma^2\Big(C_{Mat}(|s_i- s_j|\Big|1, \phi_s, \nu_s) + \tau_s^2\mathbf{1}\{s_i=s_j\}\Big)
\Big(C_{Mat}(|t_i-t_j|\Big|1, \phi_t, \nu_t)+ \tau_t^2\mathbf{1}\{t_i=t_j\}\Big).
\]
The key point here is that \emph{KSIC} allows precise estimation on $\sigma^2$, which is a parameter shared by all the modes.
The joint covariance can then be estimated by $Cov((s_i, t_i), (s_j, t_j)|\hat{\bftheta})$.
Instead, for \emph{Parametric-2M}, we treat spatial fields as replicates and ignore temporal dependency to estimate the parameters $\bftheta_s = \{\sigma_s^2, \phi_s, \nu_s, \tau_s\}$ with model
\[
Cov(s_i, s_j|\bftheta_s) = C_{Mat}(|s_i- s_j|\big|\sigma_s^2, \phi_s, \nu_s) + \tau_s^2\mathbf{1}\{s_i=s_j\}.
\]
A similar procedure applies to temporal mode to estimate $\bftheta_t = \{\sigma_t^2, \phi_t, \nu_t, \tau_t\}$. As the variance is involved twice when estimating for both modes, we estimate one overall sample variance $\hat{\sigma}_{st}^2$ from all observations and normalize the estimated joint covariance:
\[
\widehat{Cov}((s_i, t_i), (s_j, t_j)|\bftheta) = \frac{1}{\hat{\sigma}_{st}^2}Cov(s_i, s_j|\hat{\bftheta}_s)Cov(t_i, t_j|\hat{\bftheta}_t).
\]
Similarly, for \emph{Parametric-3M}, the joint covariance is estimated by 
\[
\widehat{Cov}((s_i, t_{1i}, t_{2i}), (s_j, t_{1j}, t_{2j})|\bftheta) = \frac{1}{\hat{\sigma}_{st}^4}Cov(s_i, s_j|\hat{\bftheta}_s)Cov(t_{1i}, t_{1j}|\hat{\bftheta}_{t1})Cov(t_{2i}, t_{2j}|\hat{\bftheta}_{t2}).
\]
For both \emph{Parametric} methods, due to the high dimensionality, we maximize the Vecchia-approximated likelihood rather than the full likelihood. The implementation of \emph{Parametric-2M} and \emph{Parametric-3M} uses \texttt{GpGp} package in \texttt{R} and its Fisher-scoring algorithm \citep{Guinness2019}. The package was originally designed for a single sample; we extended it to the multi-sample case by slightly modifying the code. We note that fitting a fully non-separable parametric model to $\bfSigma$ is theoretically possible, but repeatedly evaluating the likelihood of a million-dimensional covariance matrix is infeasible even with Vecchia approximation. Therefore, we include only separable parametric models in the comparison.

For correlation-based variants \emph{KSIC-cor} and \emph{naive-Vecchia-cor} in Section~\ref{sec:real-mri}, the correlations within one mode are based on all observations across ambient modes and all samples, other implementation details are the same as for \emph{KSIC} and \emph{naive-Vecchia}.

All experiments were conducted on AMD EPYC CPU nodes in the Center for High Throughput Computing (CHTC) %
. Standard experiments requested 8 GB of RAM, while the real-data experiments and data preprocessing in Section \ref{sec:real} requested up to 64 GB of RAM.

\section{Supplementary Experimental Results}

\subsection{Number of Iterations Until Convergence}\label{app:sim-niter}

In Propositions \ref{prop:cost} and \ref{prop:cost2}, we assume that the KSIC-BCD algorithm converges within a fixed number of iterations independent of the ambient dimension $p$. In this section, we provide empirical validation for this assumption by analyzing the iteration count of KSIC-BCD across various stress-test scenarios. To establish a baseline, we also report the convergence behavior of the analogous BCD algorithm for the Kronecker MLE \citep{dutilleul1999mle}.

We adopt the simulation framework from Section \ref{sec:sim-np}, generating data from a separable order-2 multi-way covariance matrix $\bfSigma$, and apply KSIC-BCD to the empirical covariance. Throughout these experiments, we fix $p_1$ while varying the secondary dimension $p_2$, the sample size $n$, and the shared range parameter $\phi$ of $\bfSigma_1$ and $\bfSigma_2$. The convergence criterion is defined as $\Vert{}\bL^{(t)}_{k} - \bL^{(t-1)}_k\Vert{}_F/\Vert{}\bL^{(t-1)}_k\Vert{}_F < \text{10}^{-3}$ for all $k$, capped at a maximum of 100 iterations.

Figure \ref{fig:np-niter} displays the number of iterations required for convergence against these three parameters. The results demonstrate that KSIC-BCD is highly efficient, typically converging within 5 iterations (indicated by the black dashed line). Crucially, this iteration count remains remarkably stable regardless of changes to the dimension $p_2$ or the range parameter $\phi$.

The only notable deviation occurs under extreme data scarcity ($n=1$), where KSIC-BCD requires approximately 9 iterations. This slight increase is fully expected due to the highly ill-conditioned nature of the empirical covariance matrix. Specifically, for the configuration $p_1=30$, $p_2=60$, and $n=1$, the sample size falls into the narrow gap between the algorithmic threshold $\max\{(m_1+1)/p_2, (m_2+1)/p_1\}<1$ and the existence threshold $p_2/p_1+p_1/p_2=2.5$. Remarkably, KSIC-BCD still achieves stable convergence. Conversely, the standard iterative algorithm for the Kronecker MLE entirely fails to converge in this regime, even with nugget stabilization, because $n=1$ falls strictly below its algorithmic threshold of $\max\{p_2/p_1,p_1/p_2\}=2$.

This stark contrast strongly reinforces the discussion in Section \ref{subsec:np-scarce} regarding the robustness of KSIC under severe data scarcity. Ultimately, these empirical results confirm that in operational settings where the KSIC projection exists, the KSIC-BCD algorithm reliably converges in a minimal number of iterations, effectively independent of the ambient dimension.

\begin{figure}[htbp]
\centering
	\begin{subfigure}{.3\textwidth}
	\centering
  	\includegraphics[width =.98\linewidth]{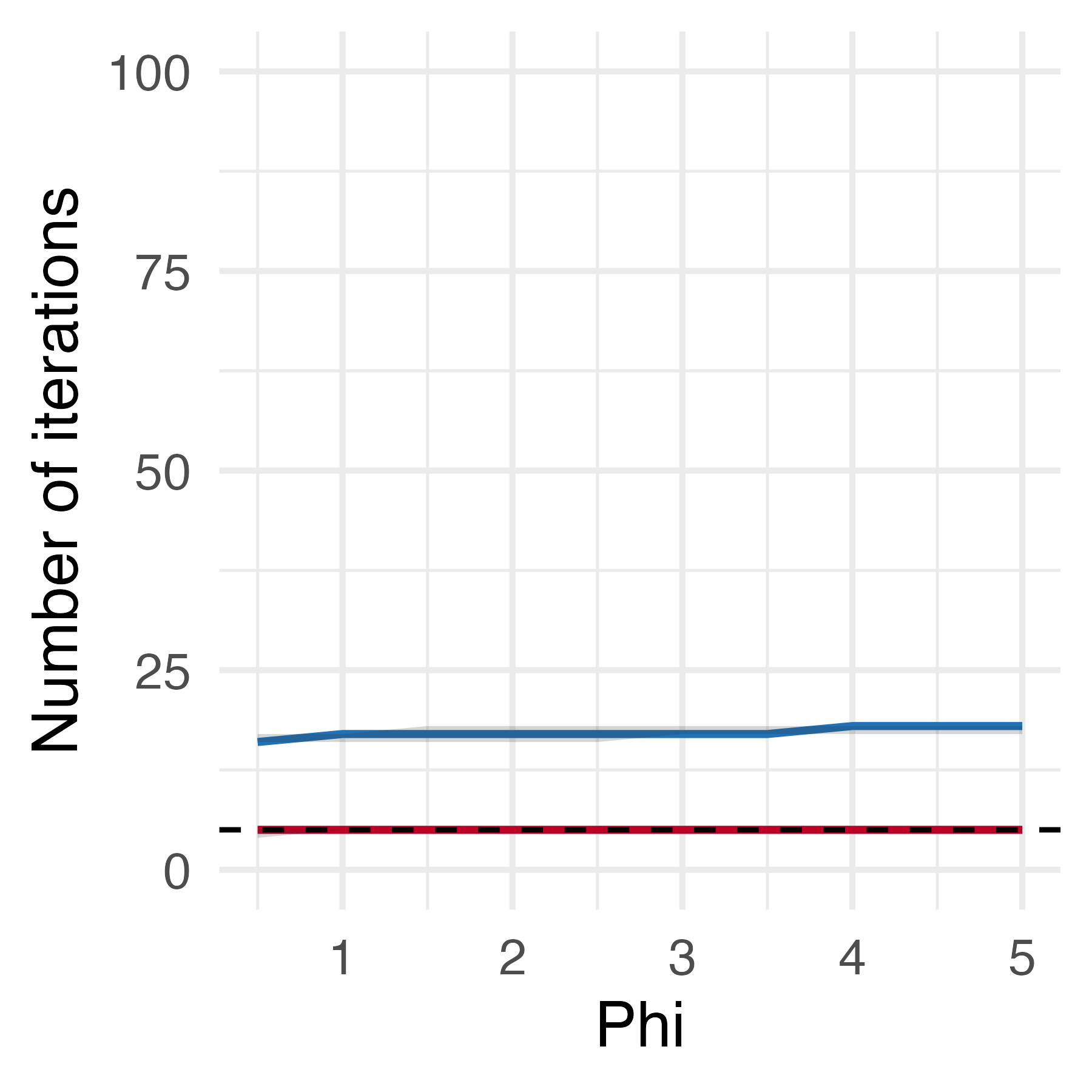}
	\caption{Iteration vs.\ range parameter.}
	\end{subfigure}%
\hfill
	\begin{subfigure}{.3\textwidth}
	\centering
 	\includegraphics[width =.98\linewidth]{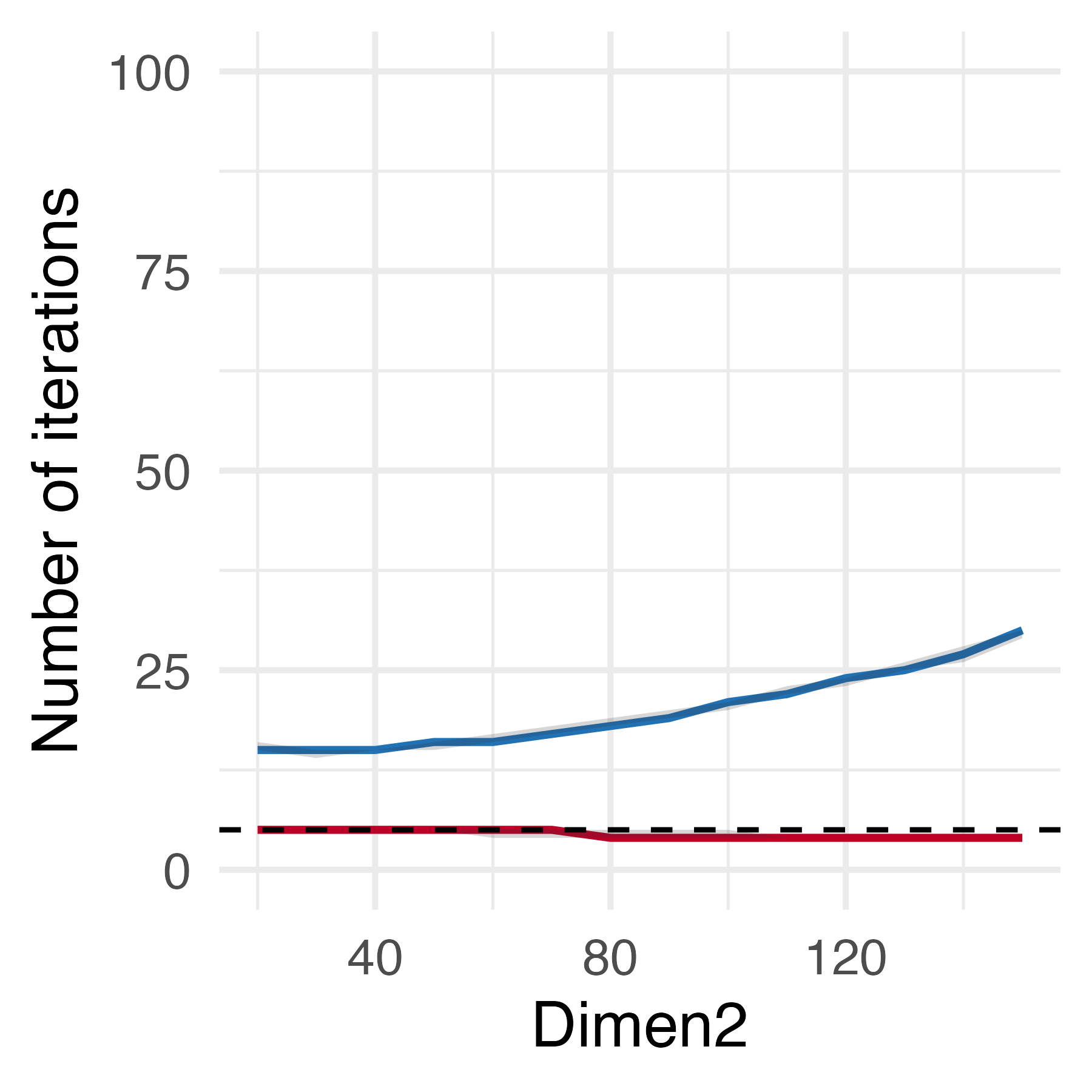}
	\caption{Iteration vs.\ $p_2$.}	
	\end{subfigure}%
\hfill
	\begin{subfigure}{.4\textwidth}
	\centering
 	\includegraphics[width =.98\linewidth]{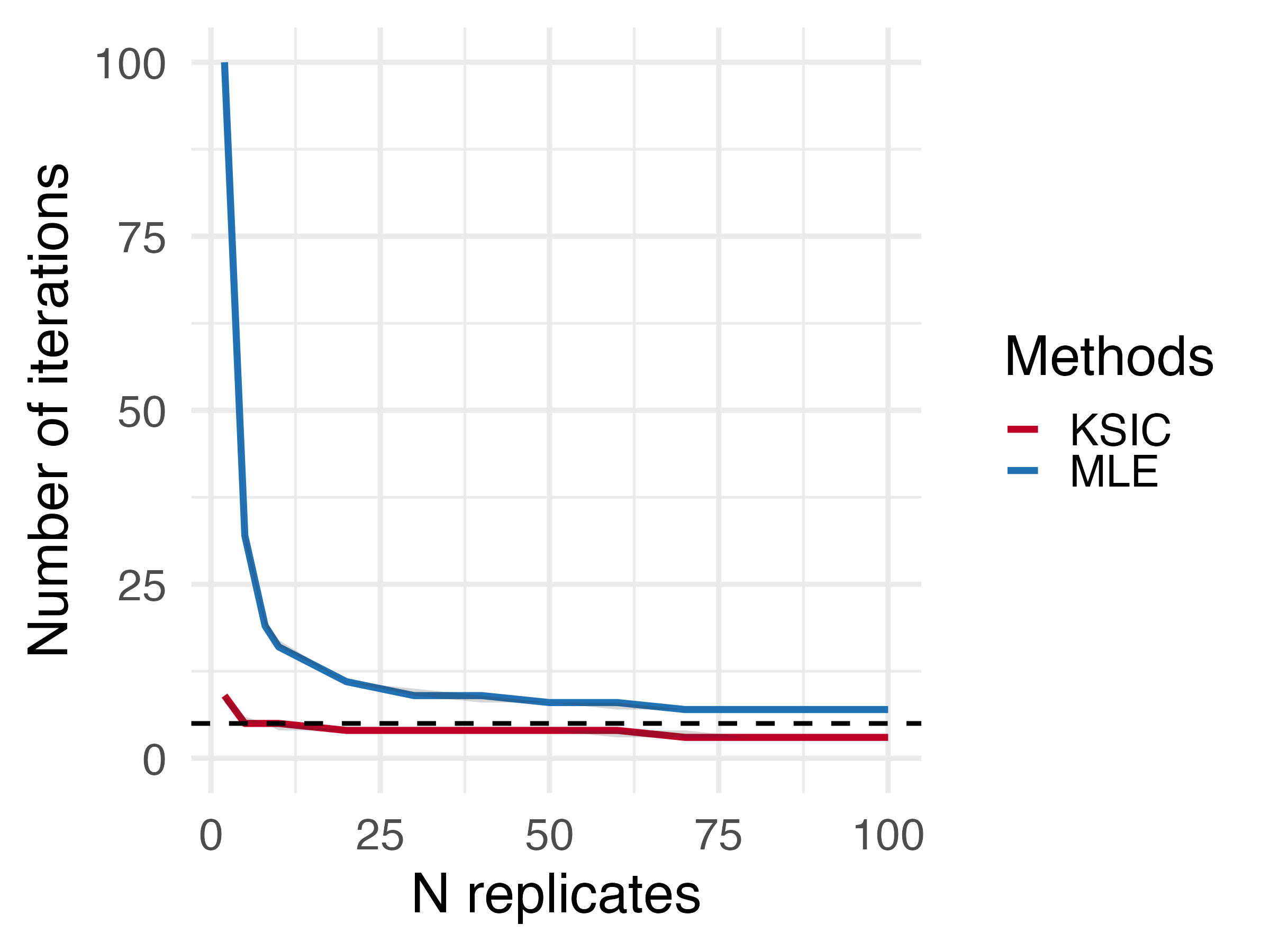}
	\caption{Iteration vs.\ $n$.}	
	\end{subfigure}%
\caption{Number of iterations required for convergence across varying range parameters, ambient mode dimensions, and sample sizes.}
\label{fig:np-niter}
\end{figure}

\subsection{Additional Simulation Results for Nonparametric Estimation}\label{app:sim-np-add}
This subsection presents additional simulation results following up on the demonstration in Section \ref{sec:sim-np}.

Figures \ref{fig:np-dimen-add} and \ref{fig:np-rep-add} include comparisons for exponential covariances with range parameter in $\{1, 2, 5\}$, which represent different geometries on both modes.
In Figure \ref{fig:np-dimen-add}, the Frobenius loss on $\bfSigma_1$ and Log-score of competing methods are compared against $p_2$, size of the second mode. Each sub-figure contains results for a specific range parameter. A smaller range parameter indicates more rapid decay of the off-diagonal elements (i.e., the covariance matrix being closer to a diagonal matrix). There's no significant difference in the trends of Frobenius loss for different range parameters. \emph{KSIC} consistently outperforms \emph{MLE} and \emph{MMCD}, and these three are significantly more accurate than naive type approaches, which aligns with the observation in the main content. In terms of log score, the \emph{MLE} starts to deviate from \emph{KSIC} earlier for small range parameter, implying that the sample size poses a greater challenge to \emph{MLE} when the effective rank of underlying covariance is large.
\begin{figure}[!t]
\centering
	\begin{subfigure}{.98\textwidth}
	\centering
  	\includegraphics[width =.98\linewidth]{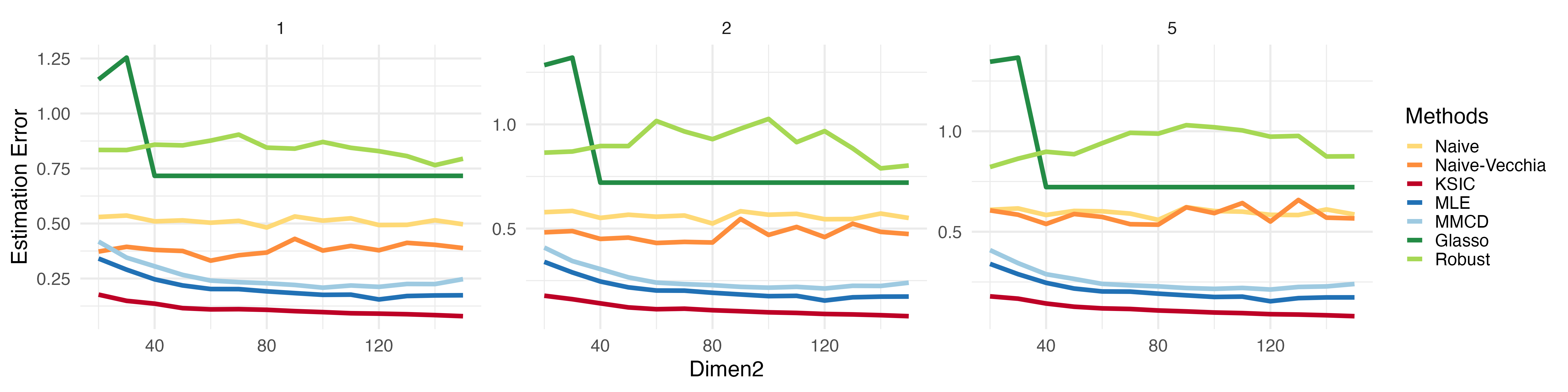}
	\caption{Frobenius loss}
	\end{subfigure}%
\vspace{0.5em}
	\begin{subfigure}{.98\textwidth}
	\centering
 	\includegraphics[width =.98\linewidth]{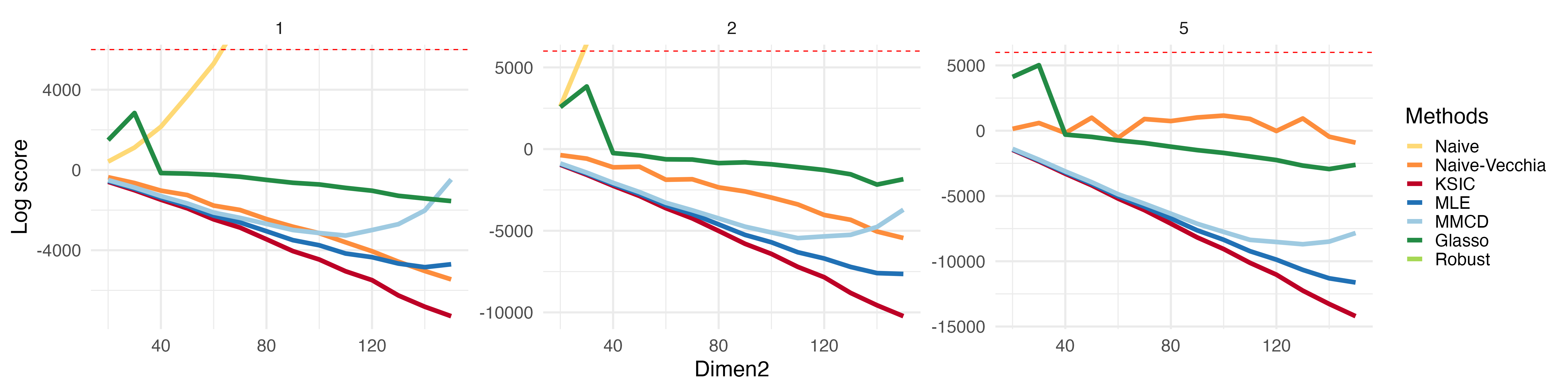}
	\caption{Log-score $^*$}
	\end{subfigure}%
\caption{Estimation error, log-score v.s. ambient mode dimension following up Figure \ref{fig:np-dimen}, varying range parameter.}
\label{fig:np-dimen-add}
\end{figure}
In Figure \ref{fig:np-rep-add}, the Frobenius loss on $\bfSigma_1$ and Log-score of competing methods are compared against the sample size $n$. Under different range parameters, KSIC once again outperforms other methods for small $n$, and has comparable performance for large $n$. A more detailed discussion on the comparison has been made in Section \ref{sec:sim-np}.
\begin{figure}[!t]
\centering
	\begin{subfigure}{.98\textwidth}
	\centering
  	\includegraphics[width =.98\linewidth]{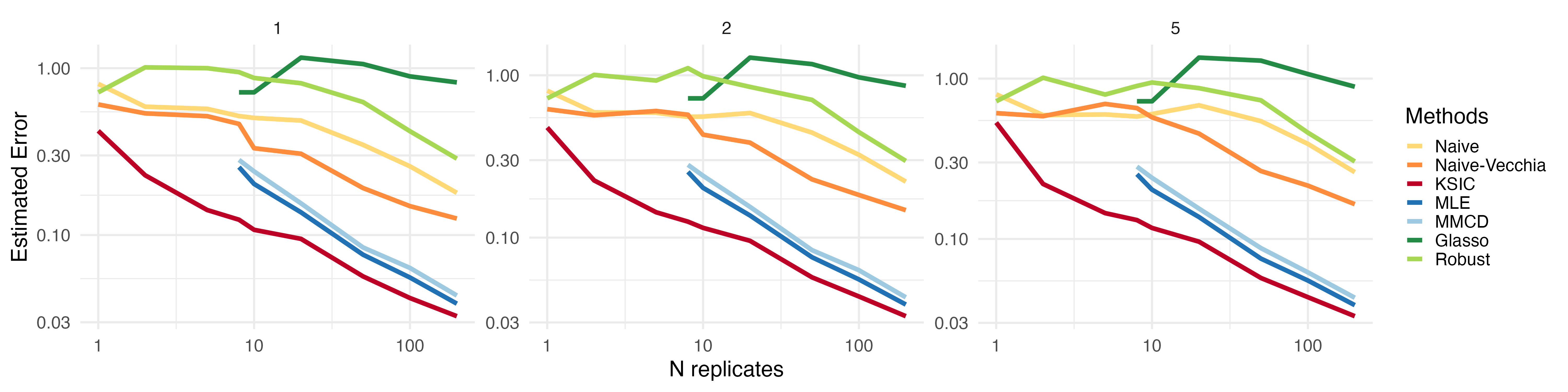}
	\caption{Frobenius loss}
	\end{subfigure}%
\vspace{0.5em}
	\begin{subfigure}{.98\textwidth}
	\centering
 	\includegraphics[width =.98\linewidth]{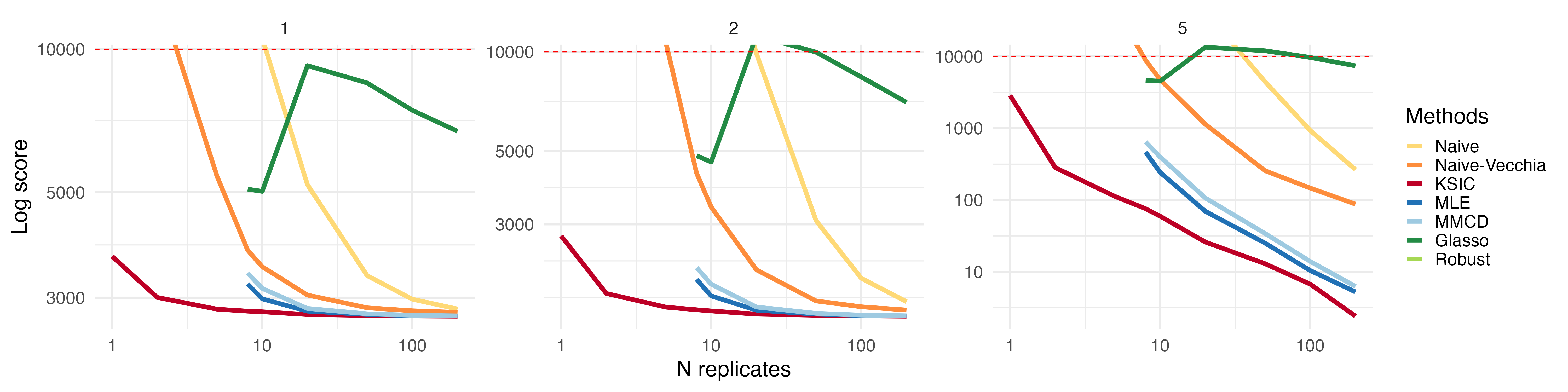}
	\caption{Log-score $^*$}
	\end{subfigure}%
\caption{Estimation error and Log-score v.s. number of replicates following up Figure \ref{fig:np-rep}, varying range parameter.}
\label{fig:np-rep-add}
\end{figure}

In Figure \ref{fig:real-matern-05-15} and \ref{fig:real-matern-15-05},  to illustrate the robustness of KSIC under different geometric patterns, we replace exponential covariance matrix $\bfSigma_1$ or $\bfSigma_2$ by Mat\'ern covariance matrix with smoothness parameter $\nu = 1.5$. We note here exponential covariance is a special case of Mat\'ern with $\nu = 0.5$. Compared to the exponential covariance, the Mat\'ern covariance with $\nu=1.5$ is smoother at the origin
and yields sample paths that are once mean-square differentiable. In this experiment, with a range parameter $0.5$ and locations in unit square, the local smoothness is dominant and Mat\'ern $\nu = 1.5$ covariance deviates from an identity matrix. Aside from the covariance, the other settings are identical to Section \ref{sec:sim-np}. Overall, we see that KSIC has the best or comparably best performance both marginally in terms of Frobenius loss and jointly in terms of log-score, which aligns with the results in Section \ref{sec:sim-np}. There is some additional information. First, compared to Figures \ref{fig:np-dimen} and \ref{fig:np-rep}, Naive-type approaches perform worse in Figure \ref{fig:real-matern-05-15}. The reason is that the Mat\'ern covariance $\bfSigma_2$ here is further from identity matrix than an exponential one, which implies a bigger loss to ignore the additional dependency, as is reflected by the discrepancy term $\|\bL^\top_{-k}\bfSigma_{-k}\bL_{-k}\|$ in Theorem \ref{col-main}. We also note that even though KSIC is most resistant to data scarcity, this additional dependency still makes estimation more difficult. This can be seen by comparing the $n=1$ cases in Figure \ref{fig:real-matern-05-15} (c), (d) to the ones in Figure \ref{fig:np-rep}. Second, KSIC seems to have a bigger bias for large $n$ in \ref{fig:real-matern-15-05} (c) compared to the one in Figure \ref{fig:np-rep} (a). This is because when $n$ grows, the approximation error of SIC becomes the dominating term in estimation error (Theorem \ref{col-main}). The approximation term cannot be reduced by increasing $n$. For a fixed $m_1$, when changing the covariance from exponential to Mat\'ern $\nu = 1.5$, the local off-diagonal correlation is stronger, and the SIC approximation has a downgraded performance, leading to larger bias in this case.
\begin{figure}[!t]
\centering
	\begin{subfigure}{.48\textwidth}
	\centering
  	\includegraphics[width =.98\linewidth]{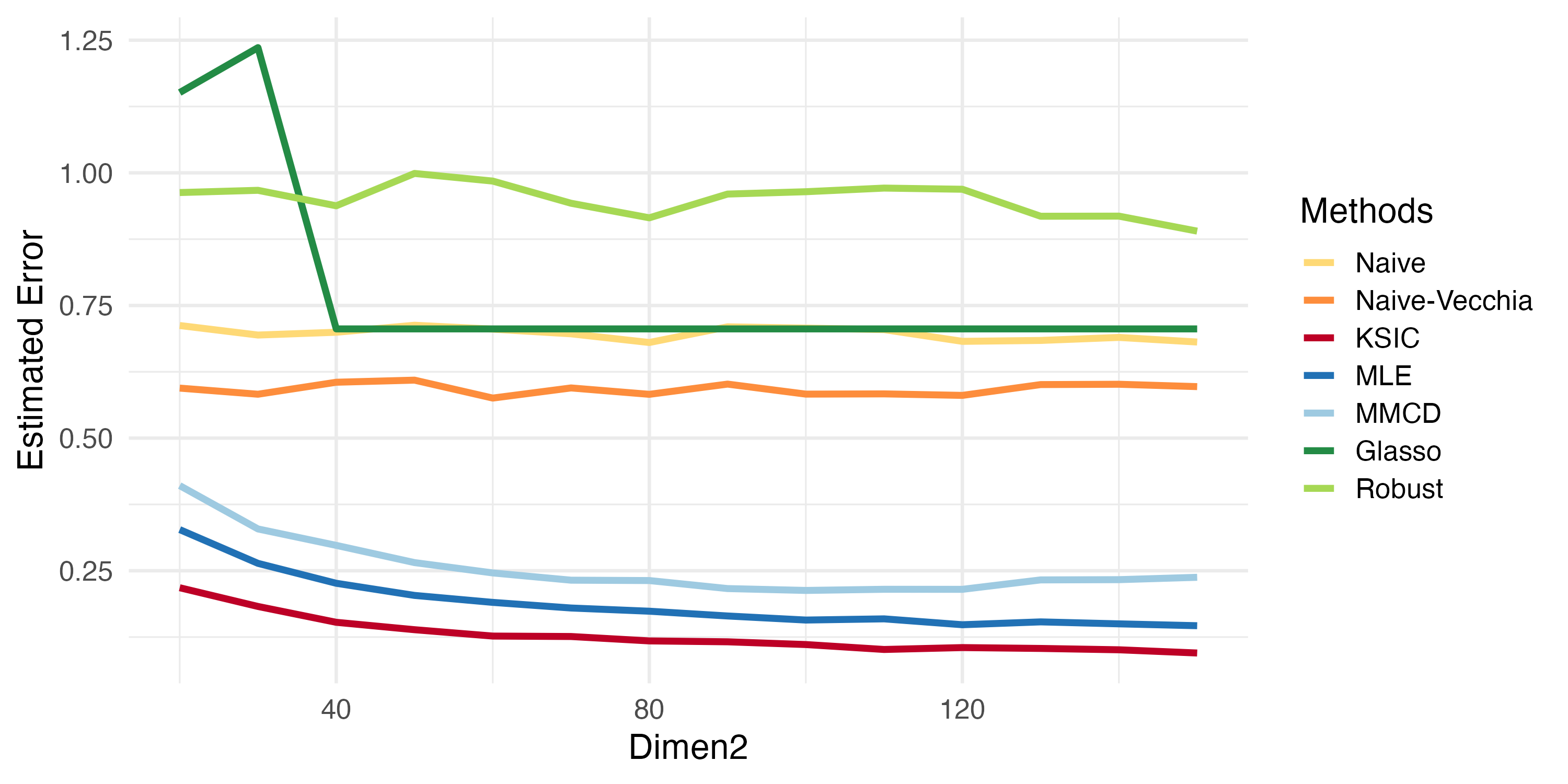}
	\caption{Frobenius loss vs $p_2$.}
	\end{subfigure}%
\hfill
	\begin{subfigure}{.48\textwidth}
	\centering
 	\includegraphics[width =.98\linewidth]{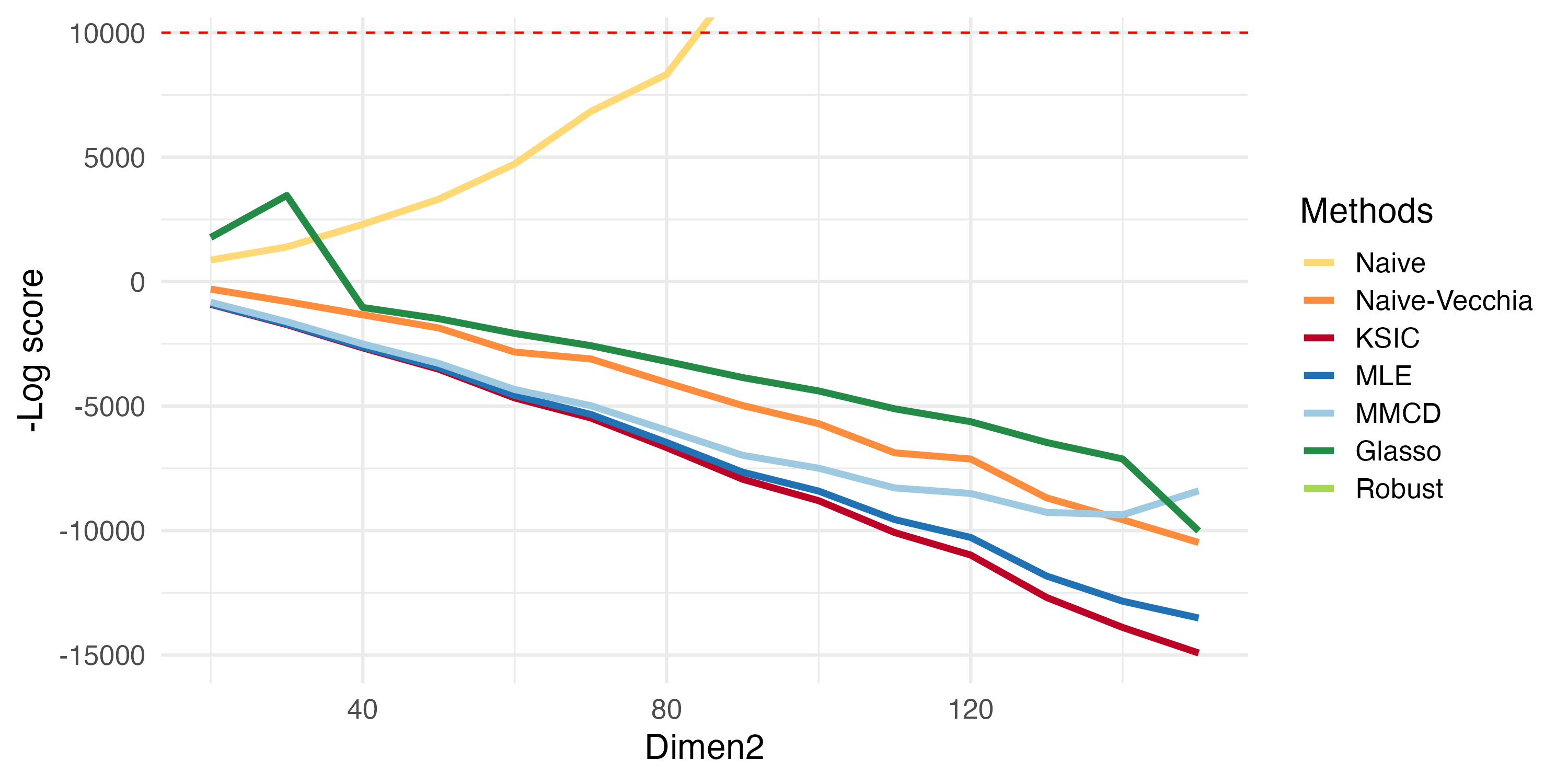}
	\caption{Log-score vs $p_2$.}	
	\end{subfigure}%
\vspace{0.5em}
\centering
	\begin{subfigure}{.48\textwidth}
	\centering
  	\includegraphics[width =.98\linewidth]{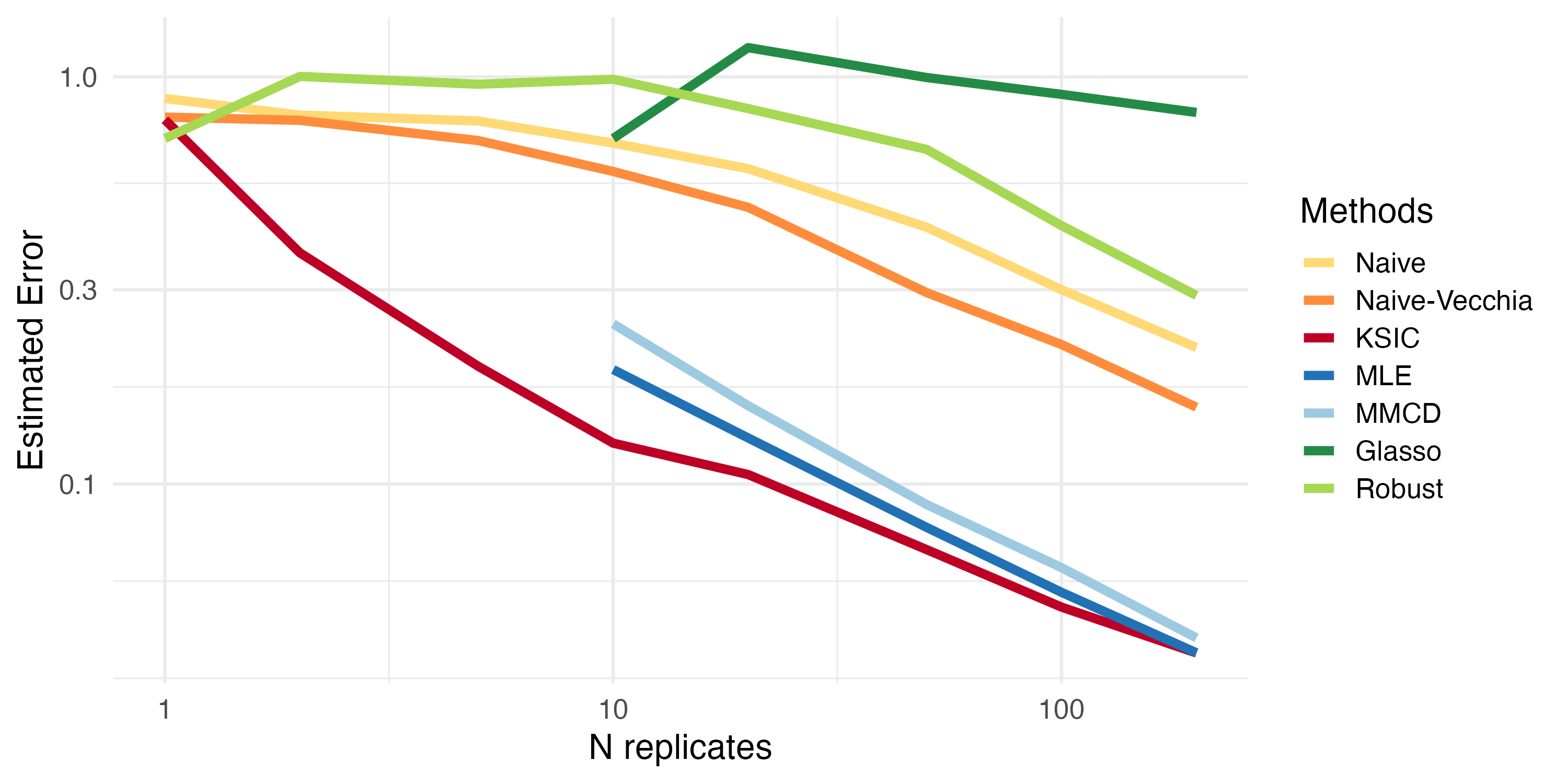}
	\caption{Frobenius loss vs $n$.}
	\end{subfigure}%
\hfill
	\begin{subfigure}{.48\textwidth}
	\centering
 	\includegraphics[width =.98\linewidth]{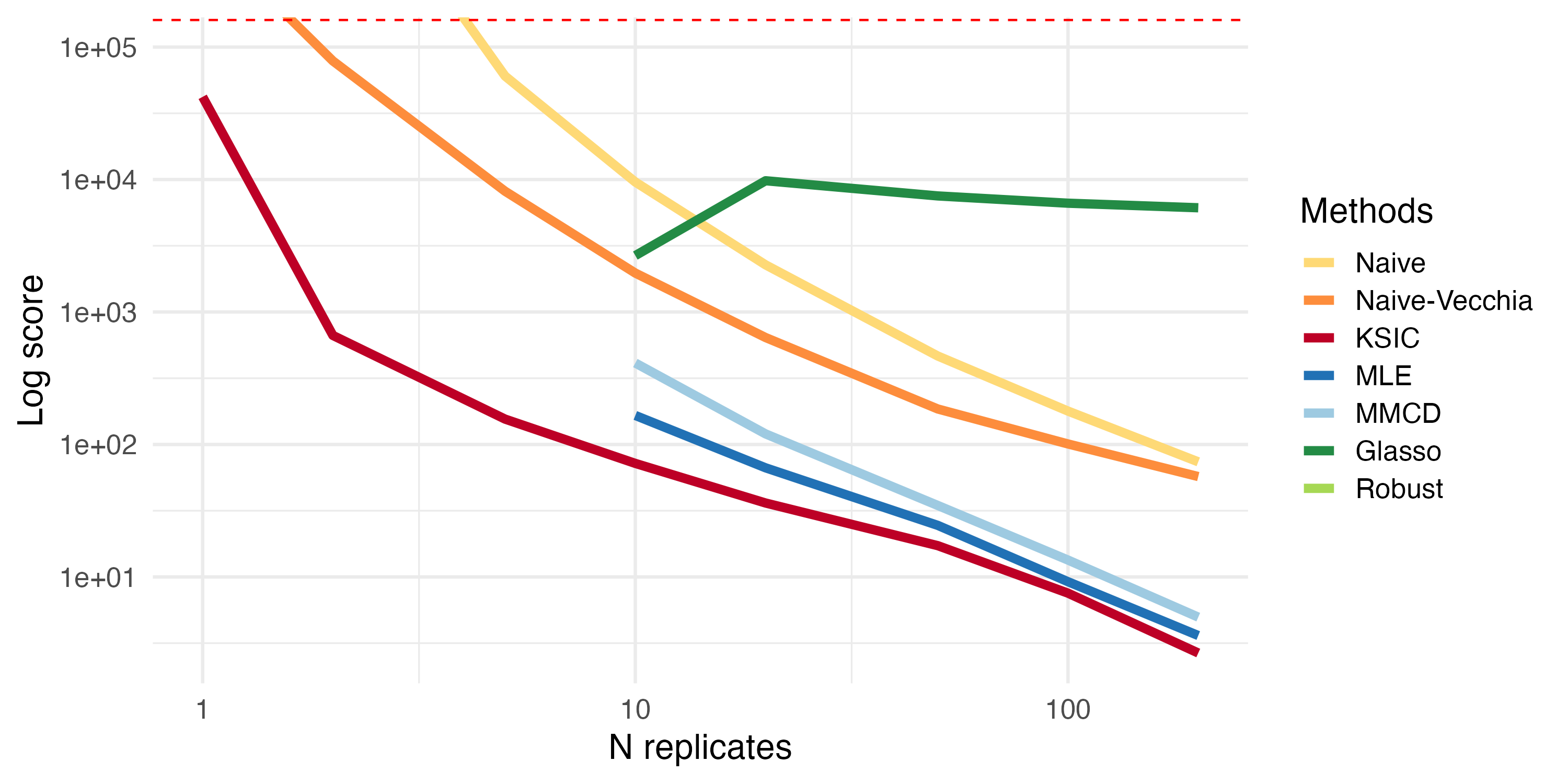}
	\caption{Log-score vs $n$.}	
	\end{subfigure}%
\caption{Estimation error and Log-score comparison when $\nu_1 = 0.5$, $\nu_2 = 1.5$.}
\label{fig:real-matern-05-15}
\end{figure}
\begin{figure}[!t]
\centering
	\begin{subfigure}{.48\textwidth}
	\centering
  	\includegraphics[width =.98\linewidth]{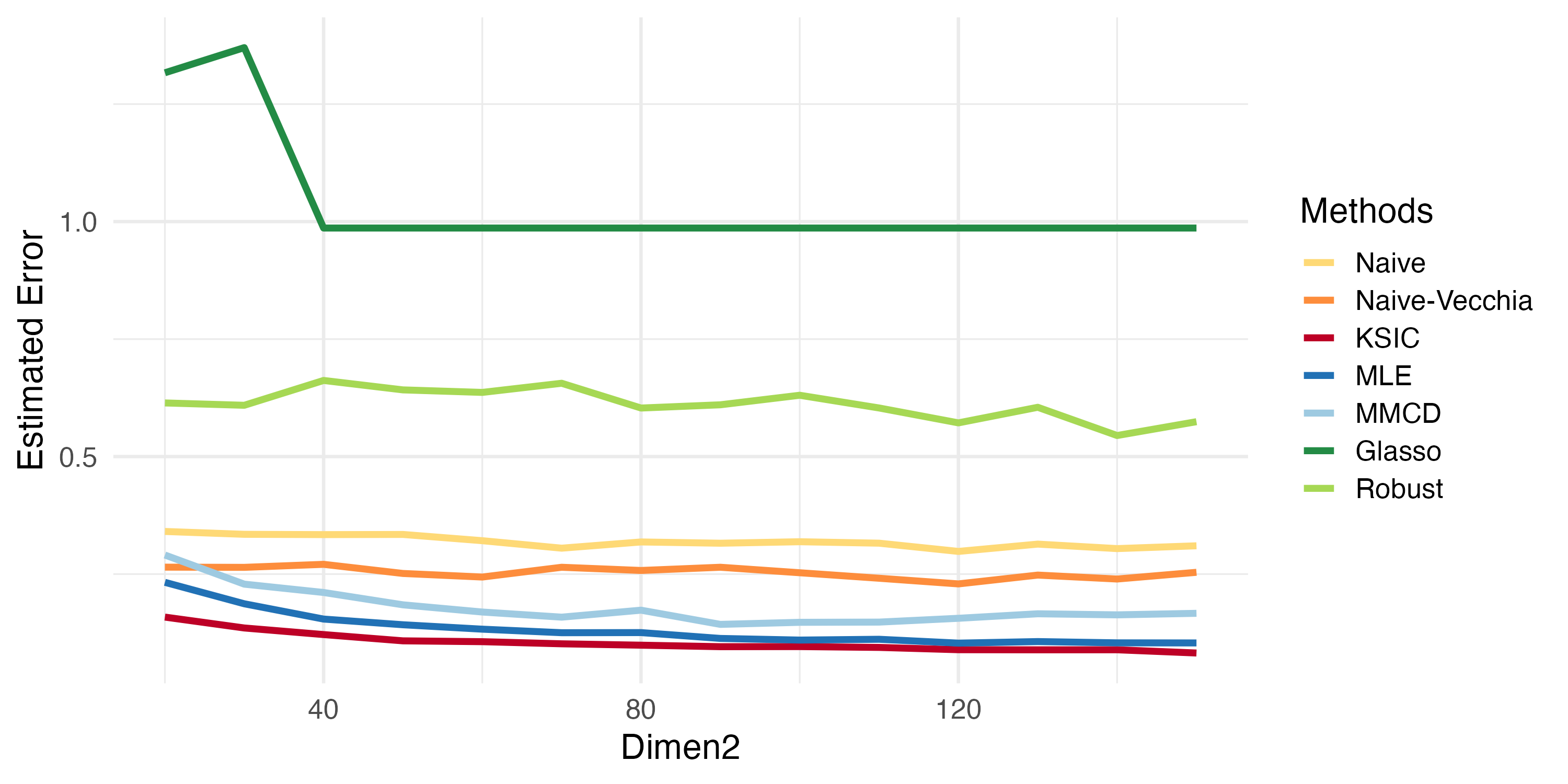}
	\caption{Frobenius loss vs $p_2$.}
	\end{subfigure}%
\hfill
	\begin{subfigure}{.48\textwidth}
	\centering
 	\includegraphics[width =.98\linewidth]{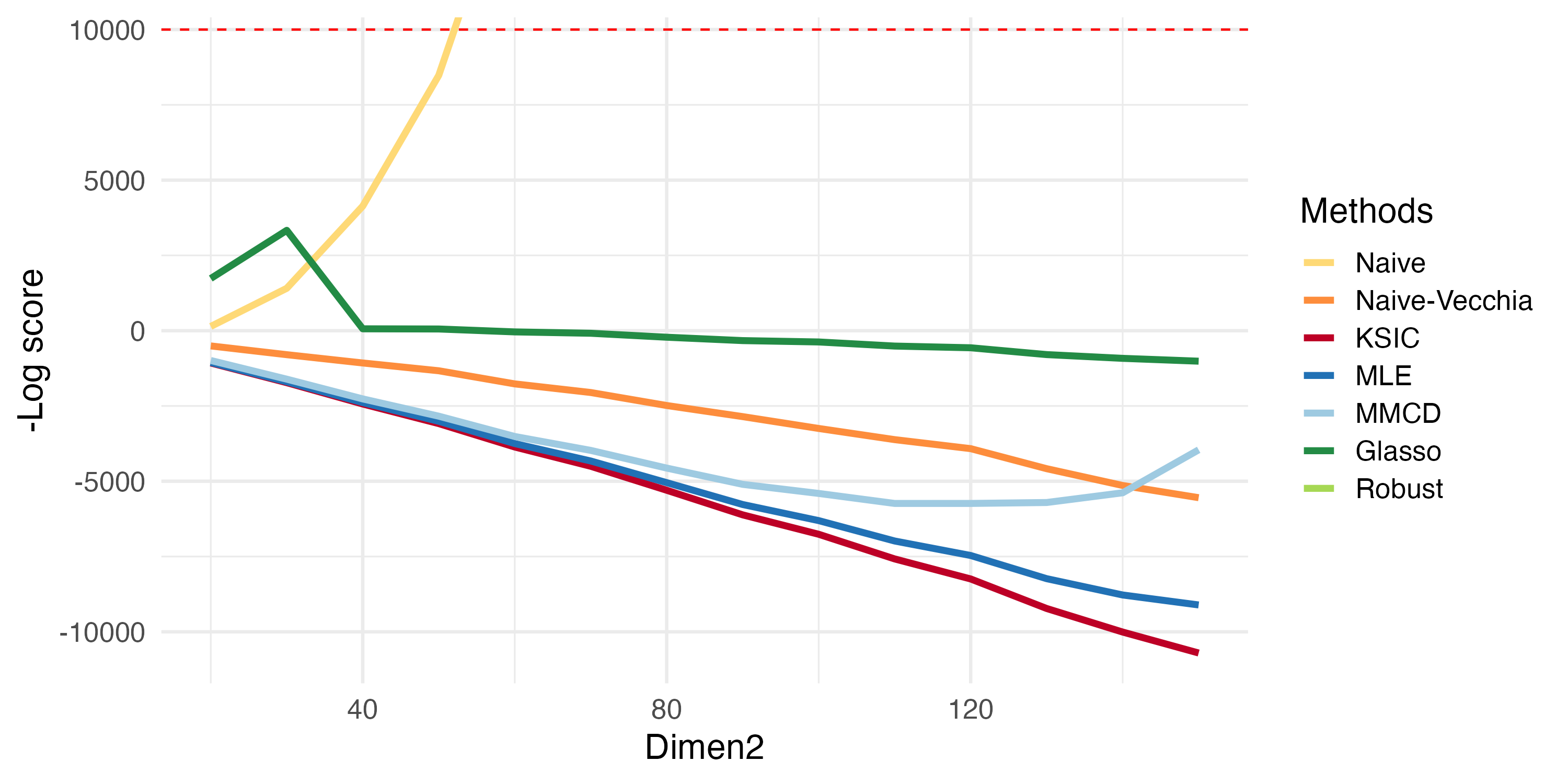}
	\caption{Log-score vs $p_2$.}	
	\end{subfigure}%
\vspace{0.5em}
\centering
	\begin{subfigure}{.48\textwidth}
	\centering
  	\includegraphics[width =.98\linewidth]{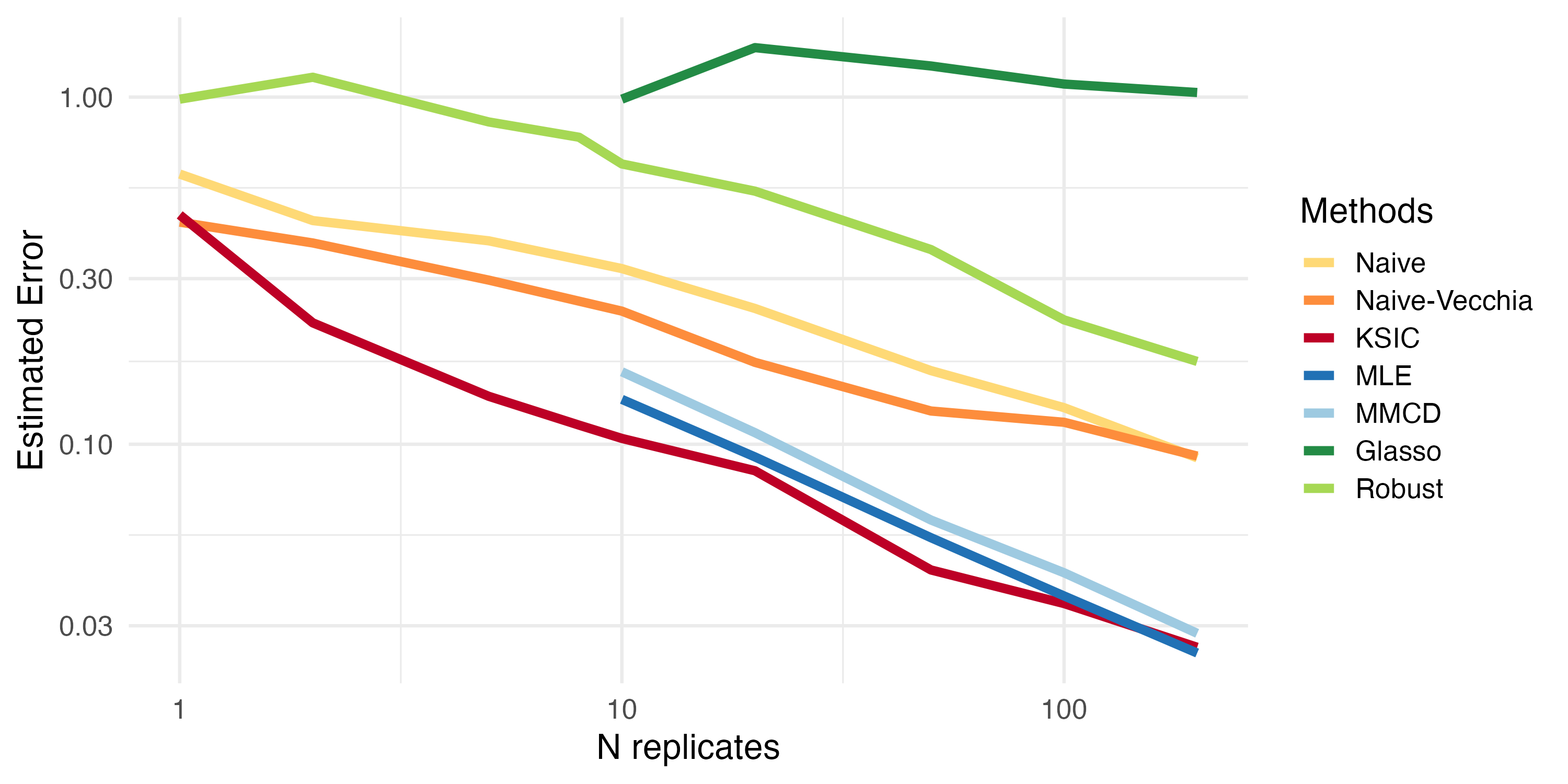}
	\caption{Frobenius loss vs $n$.}
	\end{subfigure}%
\hfill
	\begin{subfigure}{.48\textwidth}
	\centering
 	\includegraphics[width =.98\linewidth]{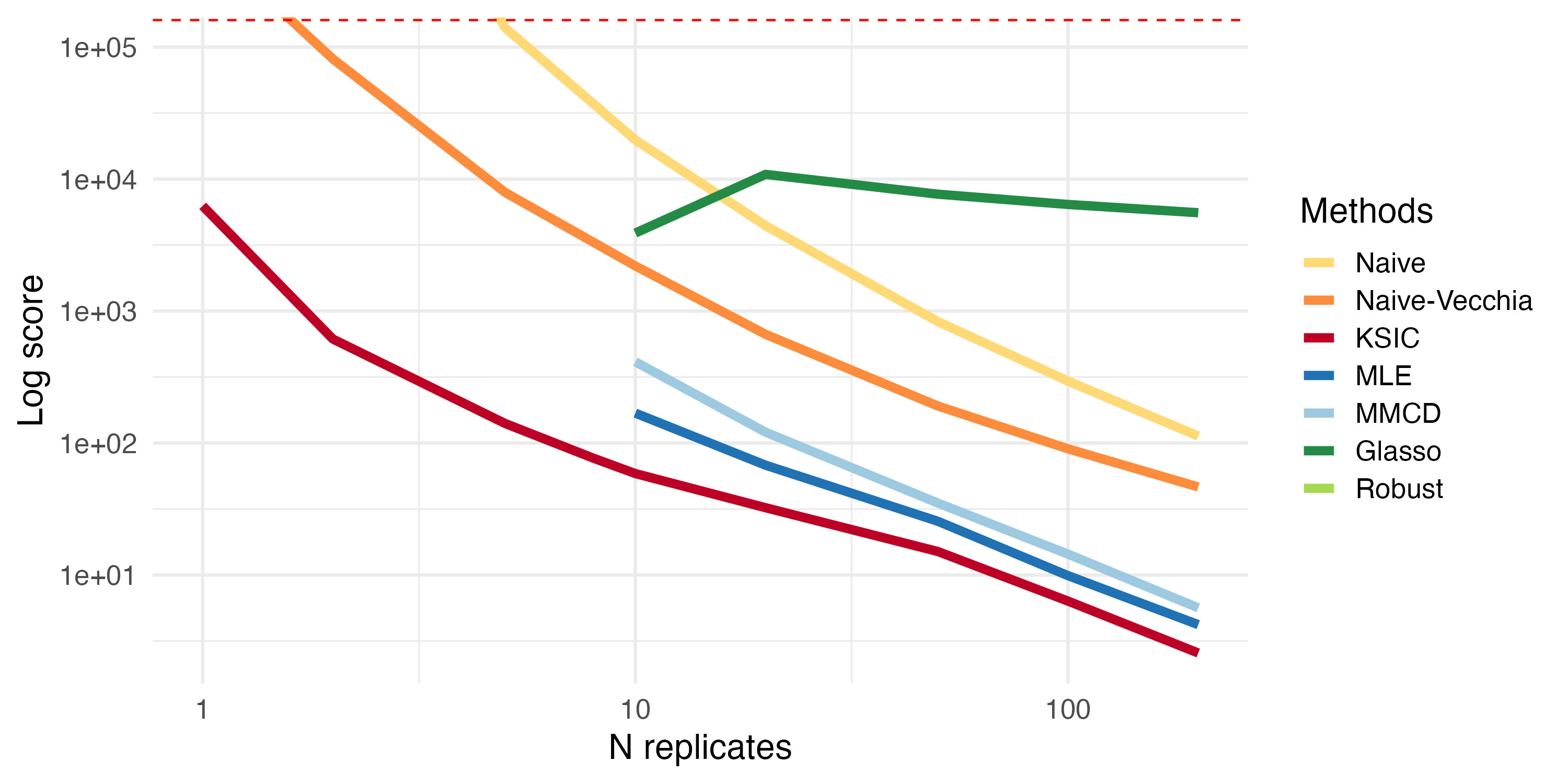}
	\caption{Log-score vs $n$.}	
	\end{subfigure}%
\caption{Estimation error and Log-score comparison when $\nu_1 = 1.5$, $\nu_2 = 0.5$.}
\label{fig:real-matern-15-05}
\end{figure}

\subsection{Supplementary fMRI Data Analysis}\label{app:real-mri}

\subsubsection{Additional Information on ABIDE Neuroimaging Data}
The ABIDE neuroimaging data used in Section \ref{sec:real-mri} were collected from 16 international imaging sites and include openly shared neuroimaging data from 539 individuals with autism spectrum disorder (ASD) and 573 typically developing controls, aged 6 to 56 years, scanned at a temporal resolution of 2 seconds. The original data consist of time series of blood-oxygen-level-dependent (BOLD) signals measured at the voxel level. The data were processed using the Configurable Pipeline for the Analysis of Connectomes (C-PAC), including motion correction and voxel-wise intensity normalization. To align voxels across samples, the images were registered to the MNI152 template. After filtering out voxels with negligible signal, which are likely to correspond to non-brain regions, 11545 voxels remain. 
The ROI level data is defined by the brain parcellation of \cite{craddock2012whole}, available as the \texttt{rois\_cc400} derivative in the ABIDE preprocessed data release.
A similar data and preprocessing pipeline has been used in recent work for statistical fMRI analysis \citep{zhao2025estimating}.

\subsubsection{Supplementary ROI-level Correlation Analysis}
In Section~\ref{sec:real-mri}, we show that the correlations estimated by KSIC respect the distance-based trend in fMRI data; that is, nearby ROIs tend to exhibit more correlated BOLD signals than distant ROIs. This pattern can arise for several reasons, including the spatial organization of neural populations and the tendency of nearby cortical regions to participate in related computations. The key point is that correlation-based KSIC captures this distance-related structure without using any physical location information for the ROIs. In addition, the example involving the left postcentral gyrus shows that correlation-based conditioning can also detect spatially distant ROI pairs with biologically meaningful functional associations. Together, these findings support the effectiveness of correlation-based conditioning in settings where the relevant geometry is unclear or not fully captured by physical distance. In this section, we provide additional correlation analyses to further support this argument.
\begin{figure}[htbp]
\centering
	\begin{subfigure}{.48\textwidth}
	\centering
  	\includegraphics[width =.98\linewidth]{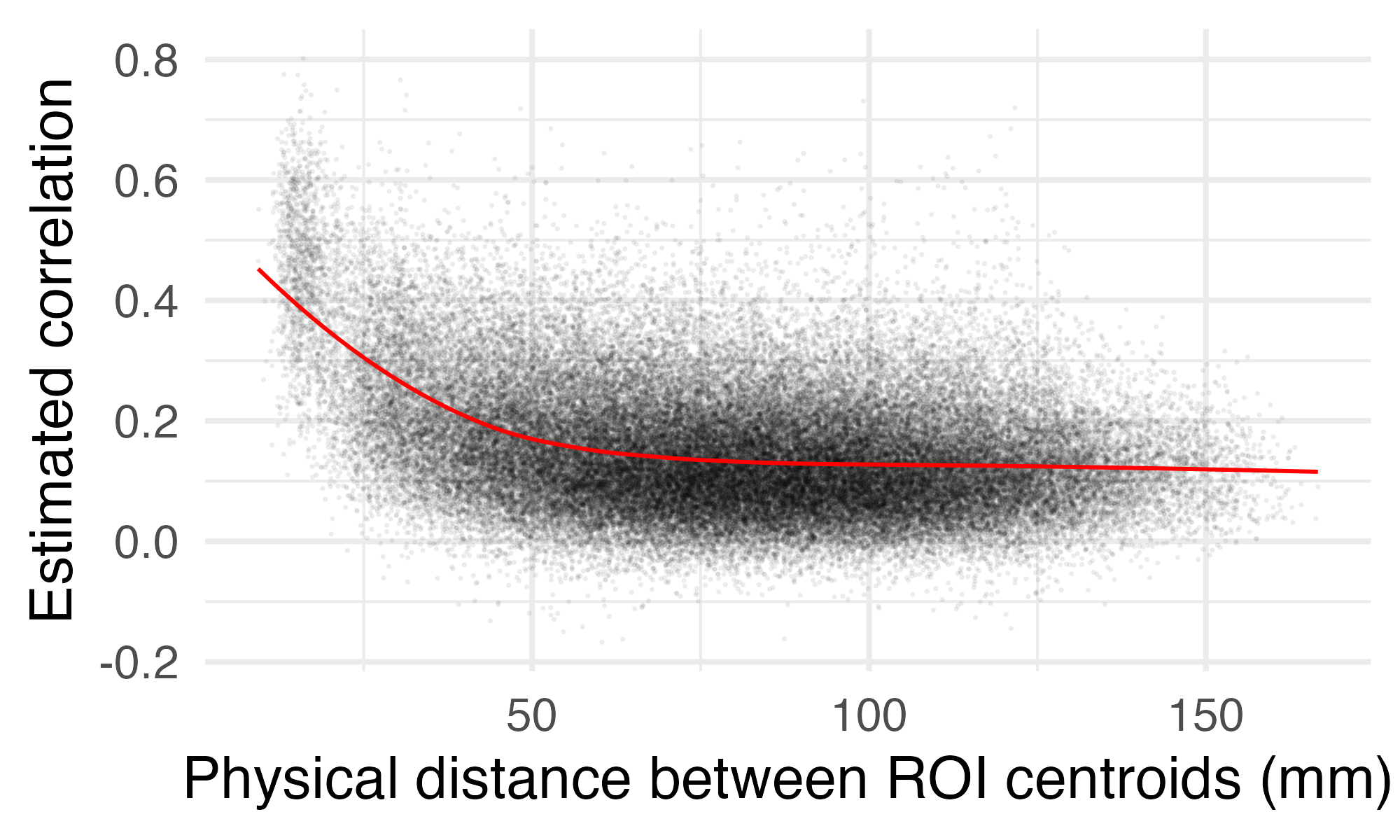}
	\caption{Estimated correlation v.s.\ physical distance}
	\end{subfigure}%
\hfill
	\begin{subfigure}{.48\textwidth}
	\centering
 	\includegraphics[width =.98\linewidth]{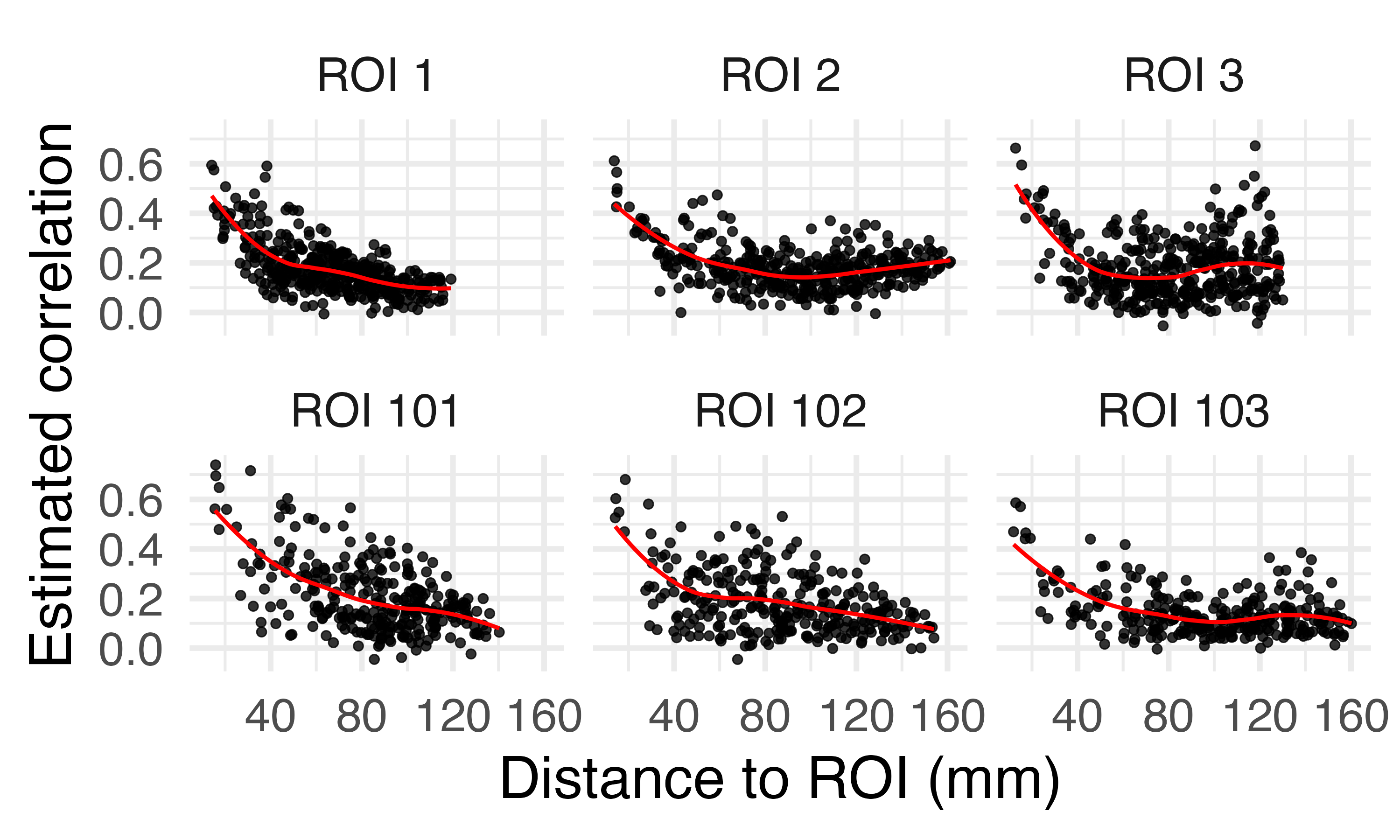}
	\caption{Estimated correlation for different ROIs.}	
	\end{subfigure}%
\caption{ROI-level FMRI analysis}
\label{fig:real-mri-app}
\end{figure}
Figure~\ref{fig:real-mri-app}(a) visualizes the estimated correlations for all ROI pairs against the physical distances between their centroids. The clear trend in the left tail confirms that KSIC generally captures the expected distance-dependent correlation structure. Figure~\ref{fig:real-mri-app}(b) presents five examples of correlations between a selected ROI and all other ROIs, showing that the strong long-range correlation observed for ROI 3 is not a common artifact.

\subsubsection{Voxel Level Analysis}
Next, we present additional results from the voxel-level fMRI data in the ABIDE release. Voxels are retained if they are active for at least $80\%$ of the time points for every subject across all sites. After filtering, $11{,}545$ voxels remain. The remaining modes and samples are processed in the same way as in the ROI-level analysis in Section~\ref{sec:real-mri}. Since the spatial dimension is much larger than the temporal dimension, which ranges from $116$ to $246$, and the available sample size, which ranges from $12$ to $73$, is substantially below the sample-size threshold required by MLE-type approaches, methods such as \emph{Naive-Vecchia} and \emph{Naive} are not stable enough for this setting. In addition, both methods require repeated full matrix inversions, which become computationally expensive when the covariance dimension reaches the order of $10^5$.

For both \emph{Naive-Vecchia} and KSIC, we assume the physical neighbors are still unknown. Both their baseline versions using random conditioning sets and the variants using correlation-based Vecchia are included. Due to the high dimensionality, we conduct the grid search on the optimal condition set size within $\{5, 10, 15, 20, 30, 40, 50\}$.

\begin{figure}[!t]
\centering
 	\includegraphics[width =.98\linewidth]{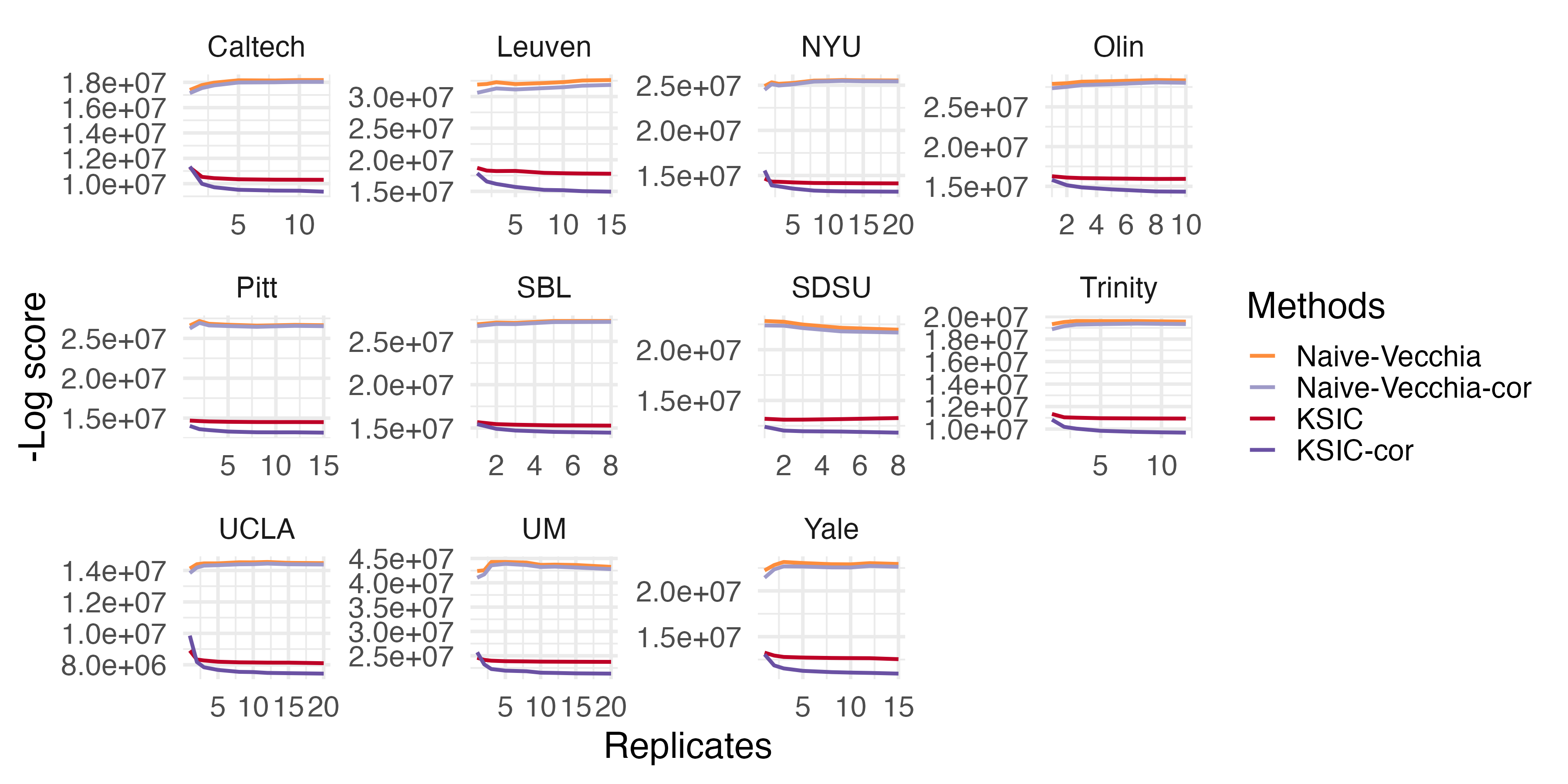}
\caption{Log-score for different approaches v.s. number of training replicates.}
\label{fig:real-mri-voxel}
\end{figure}
Figure \ref{fig:real-mri-voxel} illustrates the results for voxel-level covariance estimation. We demonstrate testing Log-score versus training sample size across 11 sites, summarized by the median over 5 random trials. The relative performance of the methods is consistent with the ROI-level analysis in Section \ref{sec:real-mri}. \emph{KSIC-cor} consistently  outperforms other methods, and the \emph{Naive} family performs significantly worse. The benefit of using data-driven conditioning sets is more pronounced in the voxel-level analysis. This is consistent with the fact that voxel-level fMRI data typically exhibit stronger local spatial and functional dependence than ROI-level summaries. During brain parcellation, voxel-wise signals are aggregated into regions defined by anatomical boundaries, functional similarity, or connectivity profiles. This aggregation reduces the dimensionality of the data and attenuates fine-scale conditional dependence among neighboring voxels. Consequently, ROI-level representations may contain weaker or smoother local dependency structure, making the advantage of data-driven conditioning less apparent than in the original voxel-level data. 

Compared to the ROI-level analysis, we highlight that this is a real case where sample size is indeed the bottleneck, which is not artificially caused by data-splitting. In this scenario, \emph{MLE} cannot provide valid result for individual sites, whereas KSIC remains stable and demonstrates its particular value under severe data scarcity.

\end{document}